\DocumentMetadata{}
\documentclass[12pt,letterpaper]{amsart}

\usepackage{amsmath,amssymb,amsthm,mathtools}
\usepackage{graphicx,booktabs}
\usepackage{float}
\usepackage[round]{natbib}
\usepackage{enumitem}
\usepackage[margin=1in]{geometry}
\usepackage{setspace}
\usepackage{microtype}
\usepackage{needspace}
\usepackage{placeins}
\usepackage{xurl}
\usepackage{hyperref}
\usepackage{tagpdf}

\tagpdfsetup{activate-all,paratagging=false}

\hypersetup{
  colorlinks=true,
  linkcolor=blue,
  citecolor=blue,
  urlcolor=blue,
  pdftitle={Priority Blockers: Welfare and Incentives in Matching under General Constraints},
  pdfauthor={Lars Ehlers and M. Bumin Yenmez},
  pdfkeywords={general upper bounds, cumulative offers, greedy choice,
    priority blockers, strategy-proofness, refugee resettlement}
}

\newtheorem{theorem}{Theorem}
\newtheorem{proposition}{Proposition}
\newtheorem{lemma}{Lemma}
\newtheorem{corollary}{Corollary}
\theoremstyle{definition}
\newtheorem{definition}{Definition}
\newtheorem{example}{Example}
\theoremstyle{remark}
\newtheorem{remark}{Remark}

\newcommand{\cL}{\mathcal{L}}
\newcommand{\cD}{\mathcal{D}}
\newcommand{\cF}{\mathcal{F}}
\newcommand{\cA}{\mathcal{A}}
\newcommand{\I}{\mathcal{I}}
\newcommand{\Sset}{\mathcal{S}}
\newcommand{\emptysetmatch}{\varnothing}

\begin{document}

\title[Priority Blockers]{Priority Blockers: Welfare and Incentives in Matching under General Constraints}

\author[Ehlers and Yenmez]{Lars Ehlers \and M. Bumin Yenmez}

\thanks{Ehlers: Universit\'e de Montr\'eal, Montreal, Canada
(\href{mailto:lars.ehlers@umontreal.ca}{\texttt{lars.ehlers@umontreal.ca}}).
Yenmez: Washington University in St. Louis, St. Louis, MO, USA; Durham
University Business School, Durham, UK; and \"Ozye\u{g}in University,
Istanbul, T\"urkiye
(\href{mailto:bumin@wustl.edu}{\texttt{bumin@wustl.edu}}).}

\thanks{The first author gratefully acknowledges financial support from the SSHRC
(Canada) under Insight Grant 435-2023-0129 and the FRQ (Qu\'ebec) under Soutien
aux \'equipes de recherche / Universitaire - nouvelle \'equipe 367853.}

\begin{abstract}
Institutions often face constraints richer than a single capacity. Protecting
the priority of a student who cannot be accommodated may then block a
lower-priority student who fits. We
introduce COP-and-Choose, a cumulative offer process followed by a choice
step, and show that its greedy version, which skips such blockers, leaves no
student worse off than any fair matching or Knapsack Deferred Acceptance
(KDA). Greedy is the unique consistent rule preserving this guarantee. With
substitutability and size monotonicity at every school, greedy is strategy-proof
and Pareto-undominated among strategy-proof mechanisms. If either fails
at a school, one additional capacity-one school yields a market where, under
every proposal order, greedy is manipulable and no strategy-proof mechanism
weakly Pareto dominates it. In a three-resource refugee-resettlement
calibration with simulated preferences, 19--33 percent of families gain
relative to KDA and none lose, while 21--83 percent strongly envy another
family.
\end{abstract}

\date{September 19, 2026}
\maketitle

\begingroup
\small
\noindent\textbf{Keywords:} general upper bounds, cumulative offers, greedy choice, priority blockers, strategy-proofness, refugee resettlement.\\
\noindent\textbf{JEL Classification:} C78, D47, D78.
\endgroup

\section{Introduction}

A refugee resettlement locality in Los Angeles has room for three children and
three adults. After it places a couple, the next family in priority
order---three children and two adults---no longer fits, though the family
behind them, a parent and child, does. Protecting the higher-priority family's
claim keeps out a family the locality could still serve.

We call such an applicant a \emph{priority blocker}. The configuration above
comes from our calibration of fiscal year 2017 resettlements by HIAS, a U.S.
refugee-resettlement agency; priority blockers arise at all twenty localities
under the simulated preferences (\autoref{sec:blocker-incidence}). Under
Knapsack Deferred Acceptance (KDA), the mechanism of \citet{DKT2023} for
refugee resettlement, 41--48 percent of families could move to a
strictly preferred locality without displacing another family. Priority
blockers cannot arise under ordinary capacity constraints,
where the next applicant fits whenever any does, but they can arise as soon as
constraints are richer than a headcount. Refugee resettlement localities are
limited by the medical, educational, and social services that different
families require;
daycare centers must satisfy age-specific teacher--child ratios; colleges face
multidimensional resource constraints 
\citep{DKT2023,KamadaKojima2024}. What these settings share is that feasibility
is downward closed: any subset of a feasible group remains feasible. Following
\citet{KamadaKojima2024}, we call such constraints \emph{general upper bounds}.

Under general upper bounds, strict priority protection can leave feasible
assignments unrealized: families placed in less preferred localities, or in
none, to preserve an ordering that cannot be acted on. This paper asks how much
can be gained by letting institutions reach past a priority blocker, what
discipline such skipping should respect, and what it costs in fairness and in
incentives.

The student-optimal fair matching of \citet{KamadaKojima2024} is the natural
benchmark for what priority protection delivers. To study what lies beyond it,
we add a terminal choice stage to the cumulative offer process
\citep{HatfieldMilgrom2005,HatfieldKojima2010} and call the result
\emph{COP-and-Choose}: students propose to schools in order of preference, each
school chooses from the cumulative set of all proposals it has received, and
when proposals stop each student takes her favorite among the offers she
holds. The mechanism is parameterized by how schools choose. A \emph{cutoff}
school scans its priority order and stops at the first applicant it cannot
accommodate; with cutoff choice, COP-and-Choose returns the benchmark under
every proposal order (\autoref{thm:cutoff-implementation}). A \emph{greedy}
school continues past that applicant and admits anyone below who still fits.

The comparison is not a routine application of deferred acceptance. Under
general upper bounds, school choice need not be substitutable: a student not
chosen from one set of proposals may be chosen when that set grows. Rejections
can therefore reverse, and a student may end the process holding offers from
several schools. Retaining every proposal and letting students resolve
overlapping offers by a terminal choice are both essential to the welfare
argument.

The main welfare result holds under every proposal order. Every outcome of
greedy COP-and-Choose weakly Pareto dominates every individually rational and
fair matching, including the Kamada--Kojima benchmark (\autoref{cor:fair-dominance}). The two
may coincide; when they differ, some student is strictly better off under
greedy. The result rests on \emph{reference-matching protection} (\autoref{thm:protection}).
Fix an individually rational reference matching. Informally, a school protects
it if the school continues to select its reference assignees whenever they
appear among proposals from students who weakly prefer that school to their
reference assignment. If every school protects the reference, then a student
matched in the reference either finishes with an offer she prefers to her
reference assignment or, having proposed to her reference school, is retained
there; a student unmatched in the reference cannot lose, by individual
rationality. Either way no student is worse off. The choice rules that protect every individually rational and
fair matching are exactly those retaining the cutoff choice from every
available set; we call them \emph{cutoff-containing} (\autoref{thm:exact-fair-dominance}). \autoref{thm:dual-protection}
identifies the reference matchings that cutoff and greedy choice protect.

The same principle reaches KDA, the mechanism
\citet[hereafter DKT]{DKT2023} designed for refugee resettlement with
multidimensional resource requirements. KDA is family-optimal among matchings
satisfying DKT's interference-free condition, a priority requirement weaker
than fairness, which is why KDA can already improve on the fair benchmark when
localities face several resource constraints. With a single resource and
strictly positive requirements, KDA and cutoff coincide. Greedy nevertheless
weakly Pareto dominates KDA, because greedy choice protects every
interference-free matching (\autoref{cor:if-dominance}).

These gains come from relaxing priority protection. Whenever a greedy
school skips a student, she could not have been accommodated by displacing any
set of selected students whom she outranks
(\autoref{prop:greedy-characterization}). This restriction applies to each
school's choice from its proposal set; the final matching may still be unfair
and need not be Pareto efficient. Within the cutoff-containing class, greedy
is distinguished by consistency: it is the unique rule whose choices do not
change when unselected students are removed
(\autoref{prop:cutoff-containing-uniqueness}).

The incentive analysis isolates one school's constraint and priority from the
market around it. We ask that the relevant property hold \emph{robustly}: in
every market in which that school uses greedy choice and the other schools
use substitutable rules (for order independence) or substitutable and
size-monotone rules (for strategy-proofness). Greedy COP-and-Choose is
robustly independent of the proposal order if and only if greedy choice at
that school is substitutable, and robustly strategy-proof under every
proposal-order convention if and only if greedy choice is substitutable and
size monotone (\autoref{thm:robust-boundary}). Because greedy choice is always
consistent, substitutability here is equivalent to path independence. If
either condition fails, the constraint-priority pair can be embedded, with
one additional capacity-one school, in a market in which, under every
proposal-order convention, greedy COP-and-Choose is manipulable and no
strategy-proof mechanism weakly Pareto dominates it
(\autoref{thm:robust-boundary} and \autoref{thm:sp-pareto-frontier}). In one
construction, a student gains by adding an unacceptable school to her list:
her proposal displaces another student whose next proposal revives her offer
at the school she wants. When both
conditions hold at every school, greedy is Pareto-undominated among
strategy-proof mechanisms (\autoref{cor:robust-greedy-frontier}), an
application of the general undominance property in
\autoref{lem:sp-nw-frontier}. Moreover, greedy weakly Pareto dominates every
strategy-proof COP-and-Choose mechanism built from substitutable,
cutoff-containing rules (\autoref{prop:greedy-dominates-sp-class}). The
outcomes of these mechanisms therefore lie weakly between the fair benchmark
and the greedy outcome.

Across priority orders, the boundaries follow the structure of the constraint.
Capacity constraints are exactly those under which cutoff COP-and-Choose is
strategy-proof at every priority, when the remaining schools also have
capacities (\autoref{prop:cutoff-incentives}); they eliminate priority blockers, so cutoff and
greedy coincide. \emph{Matroid} constraints---general upper bounds with an
exchange property, such as capacities and nested capacity structures---are
exactly those under which greedy COP-and-Choose is robustly strategy-proof at
every priority (\autoref{cor:greedy-incentives}). This builds on greedy-choice results of
\citet{Fleiner2003} and \citet{HKYY2026}, which, together with greedy
consistency, imply that matroids are exactly the constraints guaranteeing
substitutability at every priority. General upper bounds beyond matroids retain
the welfare guarantee but not, in general, the incentive guarantee. At a fixed
priority the boundary is finer: some non-matroid constraints satisfy both
incentive conditions under a suitable ranking (\autoref{ex:blocker}), and for knapsack
constraints whose requirement vectors are totally ordered
componentwise---including any one-dimensional resource constraint---every
priority refining ascending requirements works (\autoref{prop:ascending}). The chain
condition cannot be dropped in several dimensions.

We quantify these tradeoffs using DKT's resource model and HIAS data. Family
preferences are unobserved, so we use four simulated specifications spanning
common and idiosyncratic rankings with and without alignment to predicted
employment. With three resources, greedy improves 19--33 percent of families'
simulated assignments relative to KDA and worsens none, while 21--83 percent
strongly envy another family (\autoref{tab:hias-main}). A finite misreport
search finds profitable manipulations under a fixed proposal order. Ranking
families by ascending size restores the one-resource incentive guarantee but
shifts placement toward smaller families and lowers predicted employment.
These simulated gains and fairness, incentive, and distributional costs inform
the policy tradeoff; they do not estimate realized family welfare.

\autoref{sec:model} defines the environment and COP-and-Choose.
\autoref{sec:fair-benchmark} identifies the fair benchmark, \autoref{sec:welfare}
develops the welfare results, and \autoref{sec:incentives} gives the incentive and
priority-design boundaries. \autoref{sec:resources} presents the refugee
calibration. Proofs appear in \autoref{app:proofs}; additional results,
examples, and calibration details appear in the Online Appendix, which follows
the references in this document.

\subsection{Relation to the literature}

\citet{KamadaKojima2024} establish that general upper bounds admit a
student-optimal fair matching, characterize fair matchings by cutoff fixed
points, and show that a cumulative offer algorithm with cutoff choice finds the
matching (their Proposition~8). We take that matching as the benchmark and ask
what lies beyond it; \autoref{thm:cutoff-implementation}, which extends their
algorithmic result to arbitrary proposal orders, is the only result of ours
about the benchmark itself. The
mechanism, the welfare comparison, the discipline that replaces fairness, and
the incentive frontier are the contributions of this paper. Their analysis,
like ours, permits waste, because individual rationality, fairness, and
non-wastefulness can be incompatible under general upper bounds; the same
tension appears under type-specific bounds in \citet{EhlersEtAl2014}.

Reference-matching protection adapts the demand-set logic of
\citet{DoganImamuraYenmez2025}. Their extension method operates within deferred
acceptance, where rejections are permanent and each student holds at most one
tentative contract. Under general upper bounds selections can reverse and
terminal offers can overlap, so we protect reference assignees throughout
cumulative proposal sets and let students resolve overlaps by a terminal
choice. COP-and-Choose itself builds on the cumulative offer process of
\citet{HatfieldMilgrom2005} and \citet{HatfieldKojima2010}, adding that
terminal choice.

The incentive frontier complements impossibility results for strategy-proof
improvements over deferred acceptance \citep{APR2009,Kesten2010,KestenKurino2019}
and general Pareto-undominance results
\citep{HirataKasuya2017,AlvaManjunath2019,Erdil2014,AnnoKurino2016};
\autoref{lem:sp-nw-frontier} specializes \citet[Theorem~4]{HirataKasuya2017} to our environment.
The matroid boundary builds on greedy-choice results of \citet{Fleiner2003} and
\citet{HKYY2026}, which characterize matroids through the path independence of
greedy choice at every priority. Building on this characterization and standard
sufficiency results for strategy-proofness, we establish the exact robust
incentive boundary at a fixed constraint-priority pair and show that, outside
this domain, one additional capacity-one school rules out any strategy-proof
mechanism weakly Pareto dominating greedy.

A related literature reduces the efficiency cost of priorities by relaxing
priority protection. \citet{Kesten2010} introduces efficiency-adjusted
deferred acceptance, allowing students to consent to priority waivers
that do not affect their own assignments. Other papers formulate
relaxations at the level of the matching: \citet{EhlersMorrill2020} study
legal assignments, \citet{TroyanEtAl2020} ask whether a priority claimant
can retain an object after a chain of reassignments, and
\citet[Remark~1]{KamadaKojima2024} consider weak fairness, which disregards
envy that cannot be satisfied feasibly. Our discipline instead restricts
a school's choice from each available set. Disjoint terminal offers imply
greedy stability and weak fairness, but overlapping offers can violate
weak fairness even at a non-wasteful, Pareto-efficient outcome
(\autoref{prop:greedy-stability} and \autoref{ex:weak-fairness}). Whether
greedy outcomes satisfy the other matching-level relaxations remains open.

DKT introduce the refugee resource model, KDA, and the data we use. Their
family-optimality result concerns the interference-free benchmark class; our
protection theorem shows that greedy weakly Pareto dominates KDA's outcome.

We use fairness and interference-freeness as benchmarks for priority protection
without imposing stability on every outcome.\footnote{\citet[Theorem~4]{KamadaKojima2024}
establish robust non-existence: if one school's general upper bound is not a
capacity constraint, one can choose that school's priority order and students'
preferences so that no stable matching exists, regardless of the constraints
and priorities at the other schools. Their stability requires feasibility,
individual rationality, fairness, and non-wastefulness; fairness rules out
justified envy even when satisfying it would be infeasible.} \citet[pp.~2704--2705]{DKT2023}
note that under resource constraints stable matchings need not exist, and
finding one can be computationally difficult even when it exists. Our welfare
analysis asks whether students can be made weakly better off than under the
priority-protecting benchmarks. Greedy stability does hold whenever terminal
offers are disjoint, in particular when greedy choice is substitutable at
every school (\autoref{prop:greedy-stability}).

Related work studies knapsack constraints \citep{HatakeyamaDoi2026} and
distributional restrictions across institutions
\citep{KamadaKojima2015,KamadaKojima2018,GotoEtAl2017,KojimaTamuraYokoo2018}. Those
constraints couple schools but depend only on aggregate enrollments, whereas
ours are school-separable but may depend on student identities and resource
requirements. \citet{ChoEtAl2024} and \citet{KimuraEtAl2023} study hereditary
constraints on aggregate enrollment vectors, obtaining strategy-proof
mechanisms through the EF-$k$ relaxation and through multi-stage procedures
built from hereditary M$^\natural$-convex subconstraints, respectively; our
benchmark is reference-relative welfare protection under school-separable
constraints. \citet{ImamuraKawase2024} study ex post Pareto improvements from a
given matching under general constraints and provide complementary procedures.
Related refugee-assignment work studies stable housing matches
\citep{AnderssonEhlers2020} and data-driven employment optimization
\citep{BansakEtAl2018}, and computational work on matching with sizes studies
the allocation consequences of heterogeneous resource requirements
\citep{BiroMcDermid2014}.

\section{Model and COP-and-Choose}\label{sec:model}

This section develops the framework used throughout.

\subsection{Primitives and matchings}

Let $\I$ be a finite set of students and $\Sset$ a finite set of schools.
Each student $i\in\I$ has a strict preference relation $\succ_i$ over
$\Sset\cup\{\emptysetmatch\}$, where $\emptysetmatch$ denotes being unmatched.
We write $\succeq_i$ for the associated weak preference relation. School
$s\in\Sset$ is \emph{acceptable} to $i$ if $s\succ_i\emptysetmatch$.

Each school $s\in\Sset$ has a strict priority order $\pi_s$ on $\I$; for
$i,i'\in\I$, $i\mathrel{\pi_s}i'$ means that $i$ has higher priority than $i'$.
Let $G_s(i)=\{i'\in\I:i'\mathrel{\pi_s}i\}$ be the set of students who
outrank $i$ at $s$.

We impose \emph{general upper bounds}: for each school $s\in\Sset$, let
$\cA_s\subseteq 2^\I$ denote the sets of students that are feasible for $s$.
We assume that $\cA_s$ is nonempty and downward closed: if $I\in\cA_s$ and
$I'\subseteq I$, then $I'\in\cA_s$. In particular, $\emptyset\in\cA_s$.
A feasibility collection is a \emph{capacity constraint} if there is an
integer $q_s\in\{0,\ldots,|\I|\}$ such that $\cA_s=\{I\subseteq\I:|I|\leq q_s\}$.

A \emph{matching} is a function
$\mu:\I\to\Sset\cup\{\emptysetmatch\}$ such that
$\mu(s)\in\cA_s$ for every $s\in\Sset$, where
$\mu(s):=\{i\in\I:\mu(i)=s\}$ denotes the set of students assigned to
school $s$ under $\mu$.
A matching $\mu$ is \emph{individually rational} if
$\mu(i)\succeq_i\emptysetmatch$ for every $i\in\I$.

For $s\in\Sset$ and $i\in\mu(s)$, student $i'\in\I$ has
\emph{justified envy} toward $i$ if
$s\succ_{i'}\mu(i')$ and $i'\mathrel{\pi_s}i$. A matching is \emph{fair} if
it has no justified envy. A matching is \emph{non-wasteful} if there is no
$(i,s)\in\I\times\Sset$ such that $s\succ_i\mu(i)$ and
$\mu(s)\cup\{i\}\in\cA_s$.

For two matchings $\mu$ and $\mu'$, $\mu$ \emph{weakly Pareto dominates}
$\mu'$ if $\mu(i)\succeq_i\mu'(i)$ for every $i\in\I$. It \emph{Pareto
dominates} $\mu'$ if, in addition, the comparison is strict for some student.
A matching is \emph{Pareto efficient} if no matching Pareto dominates it.

For fixed constraints and priorities, a \emph{mechanism} maps every
student-preference profile to a matching. A mechanism is \emph{individually
rational} if every matching it produces is individually rational. It is
\emph{non-wasteful} if every matching it produces is non-wasteful. A mechanism
is \emph{strategy-proof} if no student can obtain a strictly preferred
assignment under her true preferences by changing her report while all other
reports are fixed.

For two mechanisms $f$ and $g$, $f$ \emph{weakly Pareto dominates} $g$ if
$f_i(\succeq)\succeq_i g_i(\succeq)$ for every preference profile
$\succeq$ and every student $i\in\I$. It \emph{Pareto dominates} $g$ if, in
addition, the comparison is strict for some student at some preference
profile. A mechanism is \emph{Pareto-undominated} within a class if no
mechanism in that class Pareto dominates it.

\subsection{Two priority-based choice rules}

We now specify the school choice rules used in the cumulative-offer
process. Fix a school $s\in\Sset$ throughout this subsection.

A choice rule is a map $C_s:2^{\I}\to\cA_s$ such that
$C_s(I)\subseteq I$ for every set $I\subseteq\I$ of available students. Thus,
the chosen set is always feasible, although the set of available students need
not be. The following two standard axioms compare a rule’s choices across different 
sets of available students. The rule is
\emph{substitutable} if
$C_s(I')\cap I\subseteq C_s(I)$ whenever $I\subseteq I'$, and it is
\emph{consistent} if $C_s(I)\subseteq J\subseteq I$ implies
$C_s(J)=C_s(I)$.\footnote{Consistency is also known as irrelevance of
rejected contracts (IRC); see \citet{AygunSonmez2013}.}
We use the following additional properties.

\begin{definition}\label{def:choice-axioms}
In each condition below, the requirement applies to every $I\subseteq\I$.
A choice rule $C_s$ is
\begin{enumerate}[label=(\roman*),leftmargin=2.5em]
\item \emph{non-wasteful} if, for every $i\in I\setminus C_s(I)$,
      $C_s(I)\cup\{i\}\notin\cA_s$;
\item \emph{weakly non-wasteful} if, whenever $I\setminus C_s(I)$ is
      nonempty and $i$ is its highest-priority student,
      $C_s(I)\cup\{i\}\notin\cA_s$;
\item \emph{priority-respecting} if every student in $C_s(I)$ has higher
      priority than every student in $I\setminus C_s(I)$; and
\item \emph{weakly priority-respecting} if, for every
      $i\in I\setminus C_s(I)$ and every nonempty $I'\subseteq C_s(I)$
      such that $i$ has higher priority than every student in $I'$,
      $(C_s(I)\setminus I')\cup\{i\}\notin\cA_s$.
\end{enumerate}
\end{definition}

We use two choice rules that depend on the priority order $\pi_s$.

\begin{definition}[Greedy choice]\label{def:greedy}
For a set of students $I\subseteq\I$ with $|I|=m$, list the students in
$I$ as $i_1,\ldots,i_m$ from highest to lowest priority under $\pi_s$. The
\emph{greedy choice rule} $C_s^g$ starts with $I_0=\emptyset$ and, for
$k=1,\ldots,m$, sets
\[
  I_k=
  \begin{cases}
    I_{k-1}\cup\{i_k\},&\text{if }I_{k-1}\cup\{i_k\}\in\cA_s,\\
    I_{k-1},&\text{otherwise}.
  \end{cases}
\]
It returns $C_s^g(I)=I_m$.
\end{definition}

Greedy choice is standard. The following characterization identifies the two
properties that uniquely distinguish it under general upper bounds.

\begin{proposition}
\label{prop:greedy-characterization}
The greedy choice rule $C_s^g$ is the unique non-wasteful and weakly
priority-respecting choice rule.
\end{proposition}

Priority respect requires the chosen set to be a top segment of the priority
order: no chosen student is outranked by an unchosen one. Weak priority respect
allows the school to skip a student only if she cannot be admitted by removing
any nonempty set of selected students below her in priority. It restricts
school choice from each available set; the final matching may still be unfair.

\citet{SonmezYenmez2022} also characterize greedy choice under matroid
feasibility; ours covers all general upper bounds.

\begin{definition}[Cutoff choice]\label{def:cutoff}
For a set of students $I\subseteq\I$ with $|I|=m$, list the students in
$I$ as $i_1,\ldots,i_m$ from highest to lowest priority under $\pi_s$. The
\emph{cutoff choice rule} $C_s^c$ scans the students in this order and stops
at the first student whose addition would violate feasibility. Equivalently,
let
\[
  k^*=\max\Bigl(\{0\}\cup
  \bigl\{k\in\{1,\ldots,m\}:\{i_1,\ldots,i_k\}\in\cA_s\bigr\}\Bigr).
\]
Then
\[
  C_s^c(I)=\{i_1,\ldots,i_{k^*}\},
\]
where the set is empty when $k^*=0$.
\end{definition}

This is the choice rule underlying the cutoff construction of
\citet{KamadaKojima2024}. The following result gives the analogous
characterization for cutoff choice.

\begin{proposition}
\label{prop:cutoff-characterization}
The cutoff choice rule $C_s^c$ is the unique choice rule that is weakly
non-wasteful and priority-respecting.
\end{proposition}

The two rules agree until the first student whose addition would violate
feasibility. Cutoff choice stops, whereas greedy choice skips that student and
continues. Hence, for every set $I\subseteq\I$ of available students,
\begin{equation}\label{eq:nesting}
  C_s^c(I)\subseteq C_s^g(I).
\end{equation}
The inclusion can be strict: cutoff choice respects priority but can be
wasteful, whereas greedy choice is non-wasteful and weakly priority-respecting
but need not respect priority.

\subsection{COP-and-Choose}

We now define COP-and-Choose formally. Its first stage is the cumulative
offer process of \citet{HatfieldMilgrom2005} and \citet{HatfieldKojima2010}:
every proposal remains in the receiving school's cumulative proposal set,
including proposals from students who are not currently selected, and each
school chooses from that set. Under general upper bounds this retention has
consequences it does not have under substitutable choice. A student rejected
at one stage may be selected later as additional proposals arrive, so she may
end the process holding offers from several schools, and the union of the
schools' final choices need not be a matching. The second stage is what we
add: each student chooses among her terminal offers. Under substitutable
choice the second stage is redundant (\autoref{prop:substitutability-cop});
without substitutability it is what makes the cumulative offer process a
well-defined mechanism.

Fix a profile of choice rules $C=(C_s)_{s\in\Sset}$. Initially, no school has
received a proposal, so $T_s^0=\emptyset$ for every $s\in\Sset$. At each stage
$k=0,1,\ldots$, each school $s\in\Sset$ selects $C_s(T_s^k)$. Students propose in 
preference order. A student currently selected by 
at least one school does not propose further, since she prefers every
school currently selecting her to every school to which she has not
yet proposed. A student is \emph{active} if she is selected by
no school and has an acceptable school to which she has not yet proposed.

A \emph{proposal-order convention} is a rule, fixed before preferences are
reported, that selects one active student after each history at which at least
one student is active. The rule is fixed, not the realized proposal sequence:
reports can change the history and the active set. For a randomized
implementation, fix the random seed independently of reports and apply the
same rule when comparing truthful and deviating reports. Fix a proposal-order
convention $\rho$. At stage $k$,
if no student is active, the proposal phase ends. Otherwise, the student
selected by $\rho$ proposes to her most-preferred acceptable school to which
she has not yet proposed. If student $i\in\I$ proposes to school $s\in\Sset$,
set $T_s^{k+1}=T_s^k\cup\{i\}$ and $T_{s'}^{k+1}=T_{s'}^k$ for every
$s'\neq s$. The process then proceeds to
stage $k+1$.

Because each student proposes to each school at most once, the proposal phase
terminates after at most $|\I||\Sset|$ proposals. Let $K$ denote the terminal
stage. For each $s\in\Sset$, let $T_s^K\subseteq\I$ be the terminal set of
proposers to $s$ and let $O_s=C_s(T_s^K)$ be the set of students receiving a
terminal offer from $s$. Because a
previously rejected student may later be selected, the sets $O_s$ need not be
disjoint. The mechanism computes each student $i$'s final choice from her
reported preferences, assigning her highest-ranked terminal offer:
\begin{equation}\label{eq:final-choice}
  \mu^{C,\rho}(i)=
  \begin{cases}
    \max_{\succ_i}\{s\in\Sset:i\in O_s\},&
       \text{if this set is nonempty},\\
    \emptysetmatch,&\text{otherwise}.
  \end{cases}
\end{equation}
For every school $s$, the set of students ultimately assigned to $s$ is a
subset of $O_s\in\cA_s$ and is therefore feasible by downward closure. Hence
$\mu^{C,\rho}$ is a matching.

\begin{definition}[COP-and-Choose]\label{def:cop-and-choose}
Given a profile $C$ of choice rules and a proposal-order convention $\rho$,
\emph{COP-and-Choose} is the mechanism just described. We call its proposal
phase the \emph{cumulative-offer process} (COP). If every school uses $C_s^g$,
we call the mechanism \emph{greedy COP-and-Choose}. If every school uses
$C_s^c$, we call it \emph{cutoff COP-and-Choose}.
\end{definition}

Because students propose only to acceptable schools, the resulting matching
is individually rational. The next example shows why the terminal choice
step matters and how retaining proposals can change welfare.

\begin{example}\label{ex:motivating}
Consider three students and two schools, $a$ and $b$, with
\[
  \cA_a=\cA_b=\{\emptyset,\{1\},\{2\},\{3\},\{1,3\}\}.
\]
Equivalently, give each school the two-dimensional resource-capacity vector
$(3,3)$ and let students $1$, $2$, and $3$ require $(0,2)$, $(3,2)$, and
$(1,1)$.
Students $1$ and $2$ find only $a$ acceptable, while
$a\succ_3 b\succ_3\emptysetmatch$. Both schools rank $1$ above $2$ above $3$.

Suppose student $2$ proposes to $a$, followed by student $3$. Greedy selects
$2$ and skips $3$, since $\{2,3\}\notin\cA_a$. Student $3$ then proposes
to $b$ and is selected there. When student $1$ subsequently proposes to $a$,
its cumulative proposal set is $\{1,2,3\}$. Greedy selects $1$, skips $2$
because $\{1,2\}\notin\cA_a$, and selects $3$ because
$\{1,3\}\in\cA_a$. Thus, $O_a=\{1,3\}$ and $O_b=\{3\}$.
Student $3$ holds two terminal offers and chooses $a$, yielding assignments
$(a,\emptysetmatch,a)$.

If rejections are permanent, as in the deferred-acceptance algorithm of
\citet{GaleShapley1962}, student $3$'s proposal to $a$ is discarded after
her initial rejection. When student $1$ arrives, $a$ chooses $1$ from
$\{1,2\}$, while $3$ remains at $b$. Cutoff
COP-and-Choose yields the same assignments as the discarded-rejection
procedure, $(a,\emptysetmatch,b)$:
from $\{1,2,3\}$, cutoff selects $1$ and stops at $2$. Greedy
COP-and-Choose Pareto dominates this outcome because only student $3$'s
assignment changes, from $b$ to $a$. It is not fair, however: unmatched
student $2$ prefers $a$ and outranks student $3$ there.
\end{example}

\section{The Fair Benchmark}\label{sec:fair-benchmark}

We first identify the fair benchmark for our welfare comparison. Cutoff
COP-and-Choose implements this benchmark exactly.

A \emph{student-optimal fair matching} is an individually rational and fair
matching that weakly Pareto dominates every other individually rational and
fair matching. Under general upper bounds,
\citet[Theorem~2]{KamadaKojima2024} establish
existence, although the proof of \autoref{thm:cutoff-implementation} below does
not invoke their existence result. Strict student preferences imply uniqueness.
We refer to this unique matching as the \emph{fair benchmark}.

\citet[Section~6.4]{KamadaKojima2024} also give a cumulative offer algorithm
that finds this matching: under simultaneous proposals, each school keeps the
longest feasible priority prefix of all students who have ever applied to
it---cutoff choice---and their Proposition~8 shows the outcome is the
student-optimal fair matching. The next result extends that observation to
every proposal-order convention and records two features used repeatedly
below.

\begin{theorem}
\label{thm:cutoff-implementation}
For every preference profile and proposal-order convention, cutoff
COP-and-Choose returns the fair benchmark. After every proposal, the sets
selected by different schools are pairwise disjoint, so the final
student-choice step is redundant. The outcome is independent of the
proposal-order convention.
\end{theorem}

Disjointness holds because cutoff choice is substitutable
(\autoref{lem:cutoff-rejection}): a student rejected by a school is never
selected there later, so she holds at most one offer at any time.

The fair benchmark can leave feasible improvements unrealized because cutoff
choice stops at the first infeasible addition: continuing past that student to
select a feasible lower-priority student would violate priority. With an
ordinary capacity, no such blocker exists. \autoref{sec:hias-results}
quantifies the resulting waste and the welfare gains from continuing past
these blockers in the refugee calibration.

Cutoff COP-and-Choose must retain rejected proposals because discarding a
blocker can cause cutoff choice to admit a lower-priority student and miss the
fair benchmark; \citet[Section~6.4]{KamadaKojima2024} make the same point
about their algorithm, and \autoref{ex:cutoff-da} in the Online Appendix gives
a minimal example.

\section{Greedy COP-and-Choose: Welfare}\label{sec:welfare}
\label{sec:greedy-welfare}

Having identified cutoff COP-and-Choose with the fair benchmark, we now turn
to greedy choice. For every set of available students, greedy choice contains
cutoff choice but continues past the first student whose addition would
violate feasibility. Reference-matching protection shows why this difference
generates the welfare guarantee and characterizes the reference matchings for
which the guarantee continues to hold.

\subsection{Reference-matching protection and greedy dominance}
\label{sec:benchmark-protection}

Following \citet{DoganImamuraYenmez2025}, for an individually rational
reference matching $\mu^0$ and a school $s\in\Sset$, define the
\emph{demand set} of $s$ at $\mu^0$ by
\[
  D_s(\mu^0)=\{i\in\I:s\succeq_i\mu^0(i)\}.
\]
Thus, $D_s(\mu^0)$ consists of the students who weakly prefer $s$ to their
assignment under $\mu^0$, including the students assigned to $s$ under
$\mu^0$.

\begin{definition}\label{def:protection}
A choice profile $C$ \emph{protects} an individually rational reference matching
$\mu^0$ if, for every $s\in\Sset$ and every
$I\subseteq D_s(\mu^0)$,
\begin{equation}\label{eq:protection}
  \mu^0(s)\cap I\subseteq C_s(I).
\end{equation}
\end{definition}

Only sets of available students contained in $D_s(\mu^0)$ are relevant for
protection.

\begin{theorem}\label{thm:protection}
If $C$ protects an individually rational reference matching $\mu^0$, every
COP-and-Choose outcome under $C$ weakly Pareto dominates $\mu^0$, under every
proposal-order convention.
\end{theorem}

Cutoff choice protects every individually rational and fair matching. A choice
profile $C$ is \emph{cutoff-containing} if
\begin{equation}\label{eq:cutoff-containment}
  C_s^c(I)\subseteq C_s(I)
  \qquad\text{for every $s\in\Sset$ and every $I\subseteq\I$}.
\end{equation}
Any cutoff-containing choice profile therefore protects every individually
rational and fair matching. By~\eqref{eq:nesting}, greedy choice is
cutoff-containing. \autoref{thm:protection} then yields the main welfare
comparison, while \autoref{thm:cutoff-implementation} identifies cutoff
COP-and-Choose with the fair benchmark.

\begin{corollary}
\label{cor:fair-dominance}
For every preference profile and proposal-order convention, greedy
COP-and-Choose weakly Pareto dominates every individually rational and fair
matching. In particular, it weakly Pareto dominates cutoff COP-and-Choose. If
the two outcomes differ, the greedy outcome Pareto dominates the cutoff
outcome.
\end{corollary}

This comparison has a useful implication. If an individually rational and fair
matching is Pareto efficient, then every greedy COP-and-Choose outcome must
coincide with it. Indeed, \autoref{cor:fair-dominance} gives weak Pareto
dominance, while Pareto efficiency rules out any strict improvement. In
particular, greedy and cutoff coincide whenever the fair benchmark is Pareto
efficient.

Greedy can Pareto improve on the fair benchmark even when no student
can feasibly move to a preferred school while all other students'
assignments remain unchanged
(\autoref{prop:beyond-benchmark-completion}). It can even make every student
strictly better off than at a non-wasteful fair benchmark: under a suitable
proposal order in \autoref{ex:universal-improvement}, and under every order
in \autoref{rem:order-robust-improvement}.

In \autoref{ex:motivating}, greedy moves student $3$ from $b$ to $a$ relative
to cutoff and leaves students $1$ and $2$ unchanged. The resulting outcome is
unfair: continuing past blocker $2$ seats lower-priority student $3$. Yet
student $2$ cannot feasibly replace student $3$ at $a$, as required by weak
priority respect (\autoref{prop:greedy-characterization}).

\subsection{The exact cutoff-containing class}\label{sec:rule-frontier}

The next result shows that cutoff containment is also necessary for protecting
every individually rational and fair matching, equivalently for guaranteeing
that every COP-and-Choose outcome weakly Pareto dominates every such matching.

\begin{theorem}\label{thm:exact-fair-dominance}
For a profile $C$ of choice rules, the following are equivalent.
\begin{enumerate}[label=(\roman*),leftmargin=2.5em]
\item $C$ is cutoff-containing.
\item For every preference profile, $C$ protects every individually rational
and fair matching.
\item For every preference profile and every proposal-order convention, every
COP-and-Choose outcome under $C$ weakly Pareto dominates every
individually rational and fair matching.
\end{enumerate}
\end{theorem}

Thus, cutoff is the pointwise smallest rule that protects every individually
rational and fair matching, while greedy is a non-wasteful rule with the same
protection property. Consistency singles out greedy within this class.

\begin{proposition}
\label{prop:cutoff-containing-uniqueness}
For each school $s\in\Sset$, the greedy rule $C_s^g$ is the unique
choice rule that is both consistent and cutoff-containing.
\end{proposition}

Consistency requires that removing unselected students leave the school's
choice unchanged. Greedy also requires only one priority scan per school choice.
It operates through student proposals, local school choices, and terminal
student choice, without first computing a benchmark or specifying an additional
welfare objective. A benchmark-preserving optimizer requires both, as well as
the full reported preferences; different objectives select different
improvements (\autoref{sec:oracle}).

The Online Appendix studies cutoff-containing alternatives that can select
more or fewer students than greedy (\autoref{cor:maximum-enrollment} and
\autoref{thm:exact-weak-claim}).

\subsection{Reference matchings protected by cutoff and greedy}\label{sec:benchmark-frontier}

We now reverse the question and ask which reference matchings are protected by
a given choice rule. For a matching $\mu^0$, a school $s\in\Sset$, and a
student $i\in\I$, let
\[
  H_s^{\mu^0}(i)=D_s(\mu^0)\cap G_s(i).
\]
Thus, $H_s^{\mu^0}(i)$ is the set of students who outrank $i$ at $s$ and weakly
prefer $s$ to their assignments under $\mu^0$.

\begin{definition}[Claim compatibility]\label{def:claim-compatibility}
An individually rational matching $\mu^0$ is \emph{claim-compatible} if
\[
  H_s^{\mu^0}(i)\cup\{i\}\in\cA_s
  \qquad\text{for every $s\in\Sset$ and every $i\in\mu^0(s)$}.
\]
It is \emph{weakly claim-compatible} if
\[
  I\cup\{i\}\in\cA_s
\]
for every $s\in\Sset$, every $i\in\mu^0(s)$, and every
$I\in\cA_s$ with $I\subseteq H_s^{\mu^0}(i)$.
\end{definition}

Claim compatibility requires $i$ to be feasible together with all higher-priority
claimants at once. Weak claim compatibility requires this only with every
feasible subset of those claimants. Thus, claim compatibility implies weak claim
compatibility by downward closure.

We next characterize exactly which reference matchings are protected by cutoff
and by greedy.

\Needspace{6\baselineskip}
\begin{theorem}\label{thm:dual-protection}
For every individually rational matching $\mu^0$,
\begin{enumerate}[label=(\roman*),leftmargin=2.5em]
\item cutoff choice protects $\mu^0$ if and only if $\mu^0$ is claim-compatible;
\item greedy choice protects $\mu^0$ if and only if $\mu^0$ is weakly
claim-compatible.
\end{enumerate}
\end{theorem}

Every individually rational and fair matching is claim-compatible. Indeed, if
$i$ is assigned to $s$, every higher-priority student who weakly prefers $s$ to
her assignment must also be assigned to $s$; otherwise she would have justified
envy. Hence $H_s^{\mu^0}(i)\cup\{i\}$ is feasible. By
\autoref{thm:exact-fair-dominance} and \autoref{thm:dual-protection},
cutoff-containing rules are therefore exactly those that protect every
claim-compatible matching.

The distinction also connects the theory to the refugee-resettlement
benchmark. Interference-free matchings in the resource model are weakly
claim-compatible, so greedy protects the KDA outcome even when cutoff need
not; see \autoref{prop:resource-hierarchy} and \autoref{cor:if-dominance}.

\section{Greedy COP-and-Choose: Incentives and the Pareto Frontier}\label{sec:incentives}

Under substitutability, the terminal Choose step is redundant, and adding
size monotonicity gives strategy-proofness by the classical cumulative-offer
results \citep{HatfieldMilgrom2005}. 
We show that this classical sufficient domain is also the exact domain of
greedy's \emph{robust} strategy-proofness. Our necessity constructions use an 
additional capacity-one school and establish manipulation under every
proposal-order convention (\autoref{thm:robust-boundary}(ii)). Such an
augmentation also rules out any strategy-proof mechanism weakly Pareto
dominating greedy (\autoref{thm:sp-pareto-frontier}). Within the positive
domain, \autoref{prop:greedy-dominates-sp-class} compares greedy with other
strategy-proof mechanisms preserving the fair-benchmark guarantee.

\subsection{The strategy-proof Pareto frontier}\label{sec:sp-pareto-frontier}

The following specialization of \citet[Theorem~4]{HirataKasuya2017} gives
Pareto-undominance once strategy-proofness and non-wastefulness are established.

\begin{lemma}
\label{lem:sp-nw-frontier}
Every individually rational, non-wasteful, and strategy-proof mechanism is
Pareto-undominated among strategy-proof mechanisms.
\end{lemma}

The additional choice condition is size monotonicity. A choice rule $C_s$ is
\emph{size monotone}\footnote{\citet{Alkan2002} calls this property cardinal
monotonicity, while \citet{AlkanGale2003} use the term size monotonicity. It is
also known as the law of aggregate demand in matching with contracts.} if
\[
  |C_s(I)|\leq |C_s(I')|
  \qquad\text{whenever }I\subseteq I'\subseteq\I.
\]

Greedy choice is non-wasteful on each available set, but overlapping terminal
offers can make the final matching wasteful. Substitutability prevents overlap.

\begin{lemma}
\label{lem:greedy-sub-nw}
If greedy choice is substitutable at every school, then greedy
COP-and-Choose is non-wasteful under every proposal-order convention.
\end{lemma}

\autoref{lem:greedy-sub-nw} supplies the non-wastefulness required by
\autoref{lem:sp-nw-frontier}.\footnote{Deleting declined terminal offers and
recomputing choices can destroy reference protection;
\autoref{ex:iterated-choose} in the Online Appendix gives an example.}
Size monotonicity also makes greedy maximize enrollment from each available
set: for every feasible $I'\subseteq I$, $C_s^g(I')=I'$ implies
$|I'|\leq|C_s^g(I)|$.

In the motivating example, substitutability fails at school $a$ because
\[
  3\in C_a^g(\{1,2,3\})
  \quad\text{and}\quad
  3\notin C_a^g(\{2,3\}).
\]

\subsection{Robust order independence and strategy-proofness}\label{sec:robust}
\label{sec:robust-order}\label{sec:robust-sp}

Robustness varies the surrounding schools while holding one school's
constraint-priority pair fixed. Fix a set $I\subseteq\I$, a feasibility
collection $\cA$ on $I$, and a strict priority $\pi$ on $I$, and write
$C^\pi_{\cA}$ for the induced greedy choice rule.

Call $(\cA,\pi)$ \emph{robustly order independent} if COP-and-Choose is
independent of the ex ante proposal-order convention in every market in which
one distinguished school uses $C^\pi_{\cA}$, extended to $\I$ by declaring
every set containing a student in $\I\setminus I$ infeasible, and every other
school uses a substitutable choice rule.

Call $(\cA,\pi)$ \emph{robustly strategy-proof} if COP-and-Choose is
strategy-proof under every ex ante proposal-order convention in every market
in which one distinguished school uses $C^\pi_{\cA}$, extended to $\I$ as
above, and every other school uses a substitutable, size-monotone choice
rule. The other schools need not use greedy choice.

\citet{HatfieldKominersWestkamp2021} use this surrounding-market approach to
characterize when a stable and strategy-proof mechanism is guaranteed to
exist.\footnote{There is at most one contract per student--school pair in our
model. Every sequence of distinct students proposing to a given school is
therefore observable: a student making her first proposal there cannot
already be held there. Hence their observable substitutability and observable
size monotonicity conditions coincide with substitutability and size
monotonicity here; their condition on manipulation through contractual terms
is automatic.}
Here the necessity results concern COP-and-Choose, which remains defined when
substitutability fails and need not select stable outcomes.

\begin{theorem}
\label{thm:robust-boundary}\label{thm:robust-canonicity}\label{thm:robust-sp}
For any constraint-priority pair $(\cA,\pi)$:
\begin{enumerate}[label=(\roman*),leftmargin=2.5em]
\item The pair is robustly order independent if and only if
$C^\pi_{\cA}$ is substitutable. If substitutability fails, adding one
capacity-one school yields two proposal-order conventions whose outcomes differ
only in that one student is matched under one and unmatched under the other;
hence one outcome Pareto dominates the other.
\item The pair is robustly strategy-proof if and only if
$C^\pi_{\cA}$ is substitutable and size monotone. If either property fails,
the constraint-priority pair can be embedded in a market with one additional
capacity-one school in which COP-and-Choose is manipulable under every
proposal-order convention.
\end{enumerate}
\end{theorem}

In the proof of (ii), when substitutability fails, a student lengthens her
list by adding an unacceptable school. This forces another student's proposal
that revives her preferred offer. At other profiles, such a report can leave
her assigned to the school she finds unacceptable.

Substitutability at every school yields greedy stability, non-wastefulness,
and weak fairness: every student who prefers a school is infeasible alongside
its higher-priority assignees (\autoref{prop:greedy-stability}). Greedy
stability can also hold without substitutability (\autoref{ex:kk-stability}),
but overlapping terminal offers can violate weak fairness even at a
Pareto-efficient outcome (\autoref{ex:weak-fairness}).

Combining classical strategy-proofness with non-wastefulness
(\autoref{lem:greedy-sub-nw}) and the Hirata--Kasuya undominance result
(\autoref{lem:sp-nw-frontier}) gives the following corollary.

\begin{corollary}
\label{cor:robust-greedy-frontier}
If greedy choice is substitutable and size monotone at every school, then
greedy COP-and-Choose is strategy-proof, proposal-order independent,
non-wasteful, and Pareto-undominated among strategy-proof mechanisms.
\end{corollary}

Undominance leaves open whether greedy is best among the strategy-proof
mechanisms with which it shares the welfare guarantee of
\autoref{thm:exact-fair-dominance}. Within the class of substitutable
cutoff-containing rules it is.

\begin{proposition}
\label{prop:greedy-dominates-sp-class}
Suppose greedy choice is substitutable and size monotone at every school, and
let $C$ be a profile of substitutable, cutoff-containing choice rules such
that COP-and-Choose under $C$ is strategy-proof. Then greedy COP-and-Choose
weakly Pareto dominates COP-and-Choose under $C$ at every preference profile.
\end{proposition}

Both mechanisms are proposal-order independent
(\autoref{lem:sub-sm-sufficient} and \autoref{prop:substitutability-cop}).
Together with \autoref{thm:exact-fair-dominance}, the proposition places these
outcomes weakly between the fair benchmark and greedy. Substitutability is
essential: \autoref{ex:nonsubstitutable-sp} gives a strategy-proof
cutoff-containing mechanism that greedy does not weakly Pareto dominate.

The comparison can be strict (\autoref{ex:strict-bracket}). Any such profile
differing from greedy must fail size monotonicity at some school:
substitutability and size monotonicity imply consistency
\citep{AygunSonmez2013}, and greedy is the unique consistent cutoff-containing
rule (\autoref{prop:cutoff-containing-uniqueness}). We do not characterize all
substitutable cutoff-containing profiles that induce strategy-proof mechanisms.

The next result asks whether any strategy-proof mechanism can preserve all of
greedy's welfare gains. The comparison allows arbitrary mechanisms, with no
stability requirement.

\begin{theorem}
\label{thm:sp-pareto-frontier}
Fix a nonempty $I\subseteq\I$, a feasibility collection $\cA$ on $I$, and a
strict priority $\pi$ on $I$. If $C^\pi_{\cA}$ violates substitutability or
size monotonicity, one additional capacity-one school yields a market in which,
under every proposal-order convention, no strategy-proof mechanism weakly
Pareto dominates greedy COP-and-Choose.
\end{theorem}

The conclusion is not pointwise outside the domain: the augmentation is part
of the statement. Together with \autoref{thm:robust-boundary}(ii), this gives a sharp
robust boundary.

The next example separates the two robustness properties while holding
feasibility fixed and varying only priorities.

\begin{example}\label{ex:blocker}
Let $I=\{1,2,3\}$ and
$\cA=\{\emptyset,\{1\},\{2\},\{3\},\{1,2\}\}$.
Under $3\mathrel{\pi}1\mathrel{\pi}2$, greedy choice selects $\{3\}$ whenever
$3$ is available and otherwise selects all available students from
$\{1,2\}$. This choice rule is substitutable. It is
not size monotone, since $C^\pi_{\cA}(\{1,2\})=\{1,2\}$ but
$C^\pi_{\cA}(\{1,2,3\})=\{3\}$.
The pair is therefore robustly order independent but not robustly strategy-proof.

Under $1\mathrel{\pi}2\mathrel{\pi}3$, by contrast, greedy choice selects all
available students from $\{1,2\}$ whenever at least one is available, and
selects $3$ only when neither $1$ nor $2$ is available. This choice rule is
both substitutable and size monotone. Hence \autoref{thm:robust-boundary}(ii) implies
that the pair is robustly strategy-proof.

Finally, under $1\mathrel{\pi}3\mathrel{\pi}2$,
$C^\pi_{\cA}(\{2,3\})=\{3\}$,
$C^\pi_{\cA}(\{1,2,3\})=\{1,2\}$. Thus student \(2\) is selected from \(\{1,2,3\}\) but not from its subset \(\{2,3\}\), so substitutability fails.
\end{example}

We next consider guarantees across all priority rankings and then the design
of priorities under fixed constraints.

\subsection{Guarantees for every priority ranking}\label{sec:boundaries}

The results so far fix a priority ranking. We now ask which constraints
deliver the guarantees for every ranking. A feasibility collection $\cA$ is a
\emph{matroid} if it satisfies exchange: whenever $I,I'\in\cA$ and $|I|<|I'|$,
some $i\in I'\setminus I$ satisfies $I\cup\{i\}\in\cA$. Every capacity
constraint is a matroid. At a fixed priority, substitutability and size
monotonicity can hold beyond matroids, as \autoref{ex:blocker} showed;
\autoref{prop:fixed-priority} in the Online Appendix gives a
deletion--contraction decomposition of these two choice conditions.

For cutoff, the following proposition translates
\citet[Theorem~5 and Section~6.3]{KamadaKojima2024} through
\autoref{thm:cutoff-implementation}.\footnote{Their
Theorem~6 also rules out a feasible, fair, unanimous, and strategy-proof
mechanism beyond capacity constraints.}

\begin{proposition}
\label{prop:cutoff-incentives}
Capacity constraints are exactly the feasibility collections under which
cutoff COP-and-Choose is strategy-proof at every priority ranking and in every
market whose remaining schools are capacity-constrained. In particular, cutoff
COP-and-Choose is strategy-proof when every school has a capacity constraint. For every non-capacity feasibility collection, there is
a priority ranking under which the resulting constraint-priority pair can be
embedded in a market with one additional capacity-one school in which
cutoff COP-and-Choose is manipulable.
\end{proposition}

Here the surrounding schools have capacity constraints, whereas
\autoref{thm:robust-boundary} allows substitutable and size-monotone rules.

At a fixed priority ranking, cutoff can be strategy-proof beyond capacity
constraints. For example, in the common-priority environment of
\citet[Section~6.3]{KamadaKojima2024}, the student-optimal fair mechanism,
and hence cutoff COP-and-Choose by
\autoref{thm:cutoff-implementation}, is strategy-proof.
\autoref{ex:cutoff-manipulation} in the Online Appendix shows that even a
non-capacity matroid can separate cutoff from greedy.

For greedy, the known characterization of matroids by substitutability at
every priority \citep{Fleiner2003,HKYY2026}, restated in
\autoref{lem:matroid-known}, combines with \autoref{thm:robust-boundary} and
\autoref{thm:sp-pareto-frontier} to yield the following corollary.

\begin{corollary}\label{cor:greedy-incentives}
Matroids are exactly the feasibility collections under which greedy
COP-and-Choose is robustly strategy-proof at every priority ranking. If every
school has a matroid constraint, then for every priority profile greedy
COP-and-Choose is strategy-proof, proposal-order independent, non-wasteful,
and Pareto-undominated among strategy-proof mechanisms. For every non-matroid
feasibility collection, there is a priority ranking under which the
constraint-priority pair can be embedded, with one additional capacity-one
school, in a market in which greedy COP-and-Choose is manipulable under every
proposal-order convention and no strategy-proof mechanism weakly Pareto
dominates it.
\end{corollary}

With capacities, greedy and cutoff coincide with student-proposing deferred
acceptance. Greedy's welfare guarantee extends from matroids to arbitrary
general upper bounds.

Changing the choice rule within the cutoff-containing class cannot generally
repair incentives. \autoref{thm:no-sp-class} in the Online Appendix gives
fixed constraints and priorities for which \emph{every} profile of
cutoff-containing rules makes COP-and-Choose manipulable under every
proposal-order convention.

\subsection{Priority design under knapsack constraints}\label{sec:priority-design}

We next design priorities under multidimensional knapsack constraints.
Unlike approaches restricting constraints while leaving priorities unrestricted
\citep{KojimaTamuraYokoo2018}, we use priority as the design instrument.
\citet{HatakeyamaDoi2026} characterize protect-and-fill choice rules under
one-dimensional constraints; the following condition delivers substitutability
and size monotonicity of greedy choice in several dimensions.

\begin{proposition}
\label{prop:ascending}
Let $\cD$ be finite, let $r_i\in\mathbb{R}_+^{\cD}$ for every $i\in\I$,
and let $q\in\mathbb{R}_+^{\cD}$. A strict priority $\pi$ \emph{refines
ascending requirements} if $r_i\leq r_j$ componentwise and $r_i\neq r_j$ imply
$i\mathrel{\pi}j$.\footnote{Componentwise total ordering of the requirement
vectors is the \emph{monotonic sizes} condition of \citet[Definition~3]{DKT2023}.} Define
\[
  \cA=\Bigl\{I\subseteq\I:\sum_{i\in I}r_i^d\le q^d
  \text{ for every }d\in\cD\Bigr\}.
\]
If the requirement vectors $(r_i)_{i\in\I}$ are totally ordered componentwise,
then every strict priority $\pi$ refining ascending requirements makes
$C^\pi_{\cA}$ substitutable and size monotone. Hence $(\cA,\pi)$ is
robustly order independent and robustly strategy-proof.
\end{proposition}

If every school satisfies \autoref{prop:ascending},
\autoref{cor:robust-greedy-frontier} places greedy on the strategy-proof
Pareto frontier. A common priority gives full Pareto efficiency.

\begin{corollary}\label{cor:ascending-efficiency}
Suppose every school has the constraints of \autoref{prop:ascending}, with
the same componentwise totally ordered requirement vectors $(r_i)_{i\in\I}$
and possibly different capacity vectors $(q_s)_{s\in\Sset}$. If all schools
use the same strict priority $\pi$ refining ascending requirements, then
greedy and cutoff COP-and-Choose coincide with serial dictatorship in order
$\pi$ and are Pareto efficient.
\end{corollary}

The common priority includes common tie-breaking among identical vectors.
\autoref{prop:common-priority-efficiency} gives the general result and proof
in the Online Appendix. The chain condition is automatic in one dimension,
as used in \autoref{sec:hias-simulations}; it cannot be dropped in several
dimensions, as the counterexample following the proof of
\autoref{prop:ascending} in \autoref{app:proofs} shows.

\section{Refugee Resettlement: Theory and Calibration}\label{sec:resources}

This section specializes the theory to refugee resettlement. We first relate
the fair and interference-free benchmarks in the multidimensional resource
model of \citet{DKT2023}, and then use the FY2017 HIAS calibration to quantify
the welfare gains and diagnose the incentive boundary.

\subsection{Multidimensional resources}\label{sec:resource-model}

We specialize our framework to the resource model of \citet{DKT2023}. Let
$\cF$ be a finite set of families and $\cL$ a finite set of localities, so that
$\I=\cF$ and $\Sset=\cL$. Let $\cD$ be a finite set of resource dimensions.
Each family $f\in\cF$ has requirement vector
$r_f\in\mathbb{Z}_{+}^{\cD}\setminus\{0\}$, and each locality
$\ell\in\cL$ has capacity
$q_\ell\in\mathbb{Z}_{+}^{\cD}$. The feasible sets at $\ell$ are
\[
  \cA_\ell
  =\Bigl\{F\subseteq\cF:
    \sum_{f\in F}r_f^d\leq q_\ell^d
    \text{ for every }d\in\cD\Bigr\}.
\]
These feasibility collections are downward closed.

For the results imported below, we retain the standing assumptions of
\citet[Section~II]{DKT2023}. Every family can be accommodated on its own by at
least one locality, and every family ranks being unmatched last. A locality
ranks every family it can accommodate on its own above every family it cannot
accommodate on its own. Institutionally incompatible family--locality pairs
are absent from both sides' lists, following the compatibility extension in
\citet[Section~II]{DKT2023}; every remaining locality is acceptable to the
family. These restrictions scope the invocations of their Theorems~5 and~6,
not the paper's earlier results for arbitrary general upper bounds.

Locality $\ell\in\cL$ \emph{weakly accommodates} family $f\in\cF$ alongside
$F\subseteq\cF\setminus\{f\}$ if, for every $d\in\cD$,
\begin{equation}\label{eq:weak-accommodation}
  r_f^d=0
  \quad\text{or}\quad
  r_f^d+\sum_{f'\in F}r_{f'}^d\leq q_\ell^d.
\end{equation}
Thus, only resource dimensions used by family $f$ constrain weak accommodation.
Recall that $H_\ell^{\mu^0}(f)$ consists of families that outrank $f$ at
$\ell$ and weakly prefer $\ell$ to their assignments under $\mu^0$.

\begin{definition}[Interference-freeness]\label{def:interference-free}
A matching $\mu^0$ is \emph{interference-free} if it is individually
rational and, for every locality $\ell\in\cL$ and every $f\in\mu^0(\ell)$,
locality $\ell$ weakly accommodates $f$ alongside $H_\ell^{\mu^0}(f)$.
\end{definition}

This is the interference-free condition of
\citet{DKT2023}.\footnote{The December~10, 2019 working-paper version,
\citet{DKT2019}, calls this condition weak envy-freeness (their
Definition~10) and defines it through a pairwise relation, strong envy (their
Definition~9); their footnote~21 notes the two formulations coincide. The
published version counts the same pairs as interference violations
\citep[Section~V.B]{DKT2023}. We follow the published terminology for the
condition and retain the working paper's name for the pairwise measure used in
\autoref{sec:hias-simulations}.}

\begin{proposition}
\label{prop:resource-hierarchy}
For individually rational matchings under multidimensional resource
constraints,
\[
  \text{fair}
  \ \Longrightarrow\
  \text{claim-compatible}
  \ \Longrightarrow\
  \text{interference-free}
  \ \Longrightarrow\
  \text{weakly claim-compatible}.
\]
Each implication can be strict.
\end{proposition}

The hierarchy connects the mechanisms studied here to those of
\citet{DKT2023}. Under the assumptions above, their Theorem~5 makes KDA the
unique family-optimal interference-free matching at every preference profile,
and their Theorem~6 makes Threshold Knapsack Deferred Acceptance (TKDA)
strategy-proof and interference-free. Both
mechanisms assign only admissible pairs, so their outcomes are individually
rational in our terminology under the outside-option convention above.

In one dimension, interference-freeness and claim compatibility coincide, and
the paper's protection results then deliver the following implementation
result, which also follows from \citet[footnote~12 and Online
Appendix~C.2]{DKT2023} together with \autoref{thm:cutoff-implementation}.

\begin{proposition}
\label{prop:kda-coincidence}
With one resource, cutoff COP-and-Choose and KDA return the same matching
at every preference profile.
Hence KDA implements the fair benchmark in one dimension.
\end{proposition}

With several resources, cutoff and KDA need not coincide. Greedy nevertheless
protects every interference-free reference matching.

\begin{corollary}\label{cor:if-dominance}
For every preference profile and proposal-order convention, greedy
COP-and-Choose weakly Pareto dominates every individually rational,
interference-free matching. In particular, it weakly Pareto dominates KDA and
TKDA. At every preference profile, KDA weakly Pareto dominates cutoff
COP-and-Choose and TKDA. With one resource, KDA and cutoff COP-and-Choose
coincide.
\end{corollary}

\subsection{HIAS calibration}\label{sec:hias-simulations}

We use the FY2017 HIAS calibration of \citet{DKT2023}: 329 refugee families,
20 localities, family composition by children, adults, and seniors, locality
capacities based on realized resettlement, the compatibility matrix, and
family--locality employment weights. These administrative data determine
constraints, compatibility, and employment-based priorities. Family
preferences are simulated as in DKT, and welfare comparisons are conditional
on the four preference specifications below.

\subsubsection{Calibration design}

For each family $f\in\cF$ and locality $\ell\in\cL$, let
$V_{f\ell}\in[0,1]$ be the normalized employment weight, $Y_\ell$ a locality
shock, and $E_{f\ell}$ an idiosyncratic family--locality shock. We generate
preferences from
\[
  u_{f\ell}=\delta V_{f\ell}+\beta Y_\ell+\gamma E_{f\ell},
\]
where $Y_\ell$ and $E_{f\ell}$ are independent
$\operatorname{Unif}[0,1]$ draws and $V_{f\ell}$ comes from the data.

We use DKT's four specifications. Type~1,
$(\delta,\beta,\gamma)=(0,1,0)$, gives families a common locality ranking;
Type~2, $(0,0,1)$, gives idiosyncratic preferences. Types~3, $(1,0,1)$, and~4,
$(1,1,0)$, add employment fit to idiosyncratic and common tastes, respectively.
These specifications vary alignment with employment priorities and whether
tastes are shared or family specific. We remove administratively incompatible pairs
and pairs in which the family cannot fit an empty locality. Localities rank
compatible families by decreasing employment weight, breaking ties by family
index.

We consider both the one-dimensional constraint on total refugees and the
three-dimensional child--adult--senior constraints. For each preference type,
we generate 100 preference draws and use the same draw under both resource
specifications, so the comparison across specifications holds family
preferences fixed. For each draw we compute cutoff COP-and-Choose, KDA, and
TKDA, and compute greedy
COP-and-Choose under 12 independently seeded proposal orders.

The sampled orders measure sensitivity to proposal order. An implementation
would commit to a publicly
announced fixed ordering of families or a public randomization protocol for
selecting the next active proposer, specified independently of preference
reports.

For greedy statistics, we average across the 12 proposal orders within each
preference draw and use the draw as the unit of observation. The 95-percent
Monte Carlo $t$ confidence intervals in \autoref{fig:hias-dominance} (Online
Appendix) and the replication tables reflect variation across simulated
preference draws.

\autoref{sec:proposal-order-dispersion} reports dispersion across the 12
orders and a robustness check extending the same profiles to 50 orders.
The minima reported below are taken over the sampled orders.

We measure priority violations using strong envy, in the sense of
\citet[Definition~9]{DKT2019}, which \citet[Section~V.B]{DKT2023} count as
interference violations. Family
$f\in\cF$ \emph{strongly envies} family $f'\in\cF$, assigned to locality
$\ell\in\cL$, if $f$ prefers $\ell$ to its assignment, has higher priority
than $f'$ at $\ell$, and $\ell$ cannot weakly accommodate $f'$ alongside all
families ranked above $f'$ who weakly prefer $\ell$ to their assignments.
When reporting strong envy, we count a family if it strongly envies at least
one other family. To measure severity, we also count, for each affected family
in each greedy outcome, the number of distinct assigned families it strongly
envies and the number of localities containing those families.

Every computed matching is feasible, and KDA and TKDA have no strong-envy
violations. In one dimension, cutoff COP-and-Choose and KDA coincide in all
400 preference draws (\autoref{prop:kda-coincidence}). In all 9,600 greedy runs, greedy
COP-and-Choose gives every family a weakly preferred locality assignment
relative to cutoff COP-and-Choose, KDA, and TKDA under its simulated preference.
KDA also weakly Pareto dominates TKDA in every profile
(\autoref{cor:if-dominance}).

We implement KDA and TKDA using Algorithms~3--5 of \citet{DKT2023}, with
preferences ranked by simulated utility, employment-weight priorities, and
capacities aligned by locality name. \autoref{sec:tkda-envy} describes the
TKDA implementation; \autoref{sec:join-sensitivity} reports sensitivity to
positional capacity alignment.

\subsubsection{Welfare gains}\label{sec:hias-results}

\begin{table}[t]
\centering
\caption{HIAS-calibrated mechanism comparison}
\label{tab:hias-main}
\resizebox{\textwidth}{!}{%
\begin{tabular}{llrrrrrrr}
\toprule
Dim. & Mechanism & \shortstack{Better\\vs.\ KDA} & \shortstack{Worse\\vs.\ KDA} & Matched & Mean rank & Employment & Strong envy & Waste \\
\midrule
1 & cutoff COP & 0.0 & 0.0 & 91.9 & 5.43 & 79.0 & 0.0 & 40.8 \\
 & KDA & 0.0 & 0.0 & 91.9 & 5.43 & 79.0 & 0.0 & 40.8 \\
 & TKDA & 0.0 & 63.6 & 57.0 & 11.76 & 63.6 & 0.0 & 80.4 \\
 & greedy COP & 9.5 & 0.0 & 93.3 & 5.04 & 79.2 & 9.3 & 14.8 \\
\addlinespace
3 & cutoff COP & 0.0 & 9.9 & 76.1 & 8.33 & 79.1 & 0.0 & 50.5 \\
 & KDA & 0.0 & 0.0 & 79.9 & 7.66 & 80.6 & 0.0 & 48.2 \\
 & TKDA & 0.0 & 59.0 & 44.5 & 13.94 & 57.7 & 0.0 & 85.4 \\
 & greedy COP & 25.9 & 0.0 & 90.6 & 5.53 & 83.7 & 51.8 & 11.8 \\
\bottomrule
\end{tabular}
}\par\medskip
\begin{minipage}{0.98\linewidth}\footnotesize
\emph{Notes:} Entries except mean rank are percentages. Better and Worse compare each mechanism with KDA family by family according to simulated family preferences. KDA and TKDA follow the formal algorithms of \citet{DKT2023} under employment-weight priorities. Greedy outcomes are first averaged over 12 proposal orders; all entries then average over 100 draws for each of four preference types. Mean rank assigns rank 21, one worse than the last of the 20 localities, to an unmatched family. Employment is the share of the employment-maximizing feasible assignment. Strong envy is defined in the text. Waste is the share of families for whom some strictly preferred locality could feasibly add the family to its assigned set. Ninety-five percent Monte Carlo confidence intervals across simulated preference draws are reported in the replication tables.
\end{minipage}
\end{table}

\begin{table}[t]
\centering
\caption{Greedy COP-and-Choose relative to KDA, by preference specification}
\label{tab:hias-by-type}
\resizebox{\textwidth}{!}{%
\begin{tabular}{llrrrrrrr}
\toprule
Dim. & Type & \shortstack{Better\\vs.\ KDA} & \shortstack{Matched\\(pp)} & \shortstack{Mean rank\\(positions)} & \shortstack{Employment\\(pp)} & \shortstack{Strong envy\\(greedy)} & \shortstack{Waste\\(greedy)} & \shortstack{Waste\\(KDA)} \\
\midrule
1 & Type 1 & 11.5 & $+1.1$ & $-0.35$ & $-0.1$ & 11.4 & 28.2 & 61.8 \\
 & Type 2 & 8.4 & $+1.4$ & $-0.42$ & $-0.3$ & 6.9 & 3.6 & 18.8 \\
 & Type 3 & 8.1 & $+1.8$ & $-0.47$ & $+0.8$ & 6.6 & 4.9 & 25.9 \\
 & Type 4 & 9.9 & $+1.1$ & $-0.34$ & $+0.2$ & 12.3 & 22.6 & 56.7 \\
\addlinespace
3 & Type 1 & 33.3 & $+11.3$ & $-2.13$ & $+3.3$ & 83.0 & 23.4 & 69.2 \\
 & Type 2 & 21.9 & $+11.5$ & $-2.45$ & $+2.8$ & 21.3 & 2.8 & 28.3 \\
 & Type 3 & 19.0 & $+8.5$ & $-1.84$ & $+2.8$ & 26.7 & 3.8 & 31.9 \\
 & Type 4 & 29.5 & $+11.4$ & $-2.10$ & $+3.5$ & 75.9 & 17.3 & 63.5 \\
\bottomrule
\end{tabular}
}\par\medskip
\begin{minipage}{0.98\linewidth}\footnotesize
\emph{Notes:} Better, strong envy, and waste are percentages of families. Matched, mean rank, and employment are greedy minus KDA: percentage points of families placed, positions of mean simulated-preference rank (negative is better; an unmatched family is assigned rank 21, one worse than the last of the 20 localities), and percentage points of the employment-maximizing feasible total. Types~1 and 4 include a common locality component; Types~2 and 3 are purely idiosyncratic or idiosyncratic plus employment fit. Greedy outcomes are first averaged over 12 proposal orders; all entries then average over 100 draws. Ninety-five percent Monte Carlo confidence intervals are reported in the replication tables.
\end{minipage}
\end{table}

The calibration distinguishes welfare gains under simulated preferences,
strong-envy violations, and profitable misreports. A family can gain relative
to KDA and still strongly envy another family; strong envy establishes neither
a welfare loss relative to KDA nor manipulability at the given profile.
\autoref{sec:hias-incentives} examines incentives separately.

\autoref{tab:hias-main} pools the four preference specifications;
\autoref{tab:hias-by-type} reports greedy against KDA for each. Pooled, greedy
COP-and-Choose gives 9.5 percent of families a strictly preferred simulated
assignment relative to KDA with one resource and 25.9 percent with three, and
no family a worse one under either specification; the reverse share is zero in
every one of the 9,600 runs. Among the sampled market--order pairs, the
minimum strictly improved share is 2.1 percent in one dimension and 12.8
percent in three. The gains come with higher placement and better ranks:
relative to KDA, 1.1--1.8 percentage points more families are placed with one
resource and 8.5--11.5 with three, and mean simulated-preference rank improves
by 0.34--0.47 and 1.84--2.45 positions. Predicted employment rises by 2.8--3.5
percentage points with three resources but moves by $-0.3$ to $+0.8$ with one.
Family-preference welfare and aggregate predicted employment happen to move
together in the three-dimensional calibration; the Pareto result does not imply
this, and the one-dimensional losses and the ceiling comparison below show that
greedy need not advance the planner's employment objective. Additional
placements account for the larger three-dimensional employment gain, more than
offsetting the loss from reassigning already placed families
(\autoref{sec:oracle} in the Online Appendix gives the decomposition).

The improvement comes with weaker priority protection. Cutoff COP-and-Choose
is fair and KDA is interference-free, and both have zero strong envy. Greedy
creates strong envy for 6.6--12.3 percent of families with one resource and
21.3--83.0 percent with three. Incidence understates the intensive margin:
among affected families in the two three-dimensional specifications with the
highest incidence, the mean numbers of strongly envied assignees are 7.37 and
7.65, the 90th percentiles are 23 and 24, and the maxima are 78 and 73
(\autoref{sec:envy-severity} in the Online Appendix reports all specifications
and a locality-count measure).

Ordinary justified envy gives nearly the same three-resource incidence for
greedy, 51.9 percent versus 51.8 percent under strong envy. For KDA it is
6.1 percent, despite zero strong envy; cutoff has neither
(\autoref{tab:ordinary-envy} in the Online Appendix).

The three-resource results reveal a common pattern across preference types.
Under the common-component specifications (Types~1 and~4), greedy improves
33.3 and 29.5 percent of families relative to KDA, while 83.0 and 75.9 percent
strongly envy an assigned family. Under the idiosyncratic specifications
(Types~2 and~3), 21.9 and 19.0 percent improve and 21.3 and 26.7 percent
strongly envy. KDA leaves 63--69 percent of families with a feasible strictly
preferred locality under common-component preferences, against 28--32 percent
under idiosyncratic preferences; greedy reduces these shares to 17--23 and
3--4 percent. Common rankings concentrate proposals at the same localities,
where skipping priority blockers both frees capacity and overrides priority
claims. Greedy's welfare gain and fairness cost are largest in the same
environments.

Within each locality's proposal set, a skipped family cannot be accommodated
by displacing any nonempty set of selected families below it in employment
priority (\autoref{prop:greedy-characterization}).

Waste moves in the opposite direction. Greedy COP-and-Choose is not guaranteed
to be non-wasteful here, because \autoref{lem:greedy-sub-nw} requires
substitutability and the calibrated constraint-priority pairs do not satisfy
it. Measured waste is nevertheless far lower than under the KDA benchmark: the
share of families for whom some strictly preferred locality could feasibly
accommodate them is 14.8 percent in one dimension and 11.8 percent in three,
against 40.8 and 48.2 percent under KDA.

Residual waste can be removed by repeatedly moving one family to a strictly
preferred locality that can feasibly add it, leaving every other assignment
unchanged (\autoref{prop:completion}). This preserves greedy's benchmark
dominance. Completion is vacuous when greedy choice is substitutable at every
school (\autoref{lem:greedy-sub-nw}); outside the positive domain,
\autoref{thm:sp-pareto-frontier} rules out strategy-proof weak improvements in
its constructed markets. \autoref{sec:completion-incentives} tests a fixed
completion rule in the calibrated markets. Under the lowest-index proposal
convention, moving the lowest-indexed family with an available improvement to
its best feasible preferred locality raises beneficiary shares relative to
KDA from 8.0 to 10.7 percent with one resource and from 25.1 to 26.2 percent
with three, with no losses. A finite search finds profitable reports for this
completed mechanism in 24 and 27 of the respective 400 profiles. Because the
searches for greedy and completed greedy can select different candidate
families, these counts alone do not establish how completion changes
manipulability. Reports now affect both the proposals and the subsequent
improving moves. Completion leaves the ascending-size outcome unchanged.

Overlapping terminal offers occur in 94.4 and 93.1
percent of greedy runs with one and three resources, respectively
(\autoref{sec:overlap}). The deferred-acceptance analogue that permanently
discards greedy rejections produces different outcomes, with no general
Pareto ranking (\autoref{sec:retention}). For cutoff, discarding rejections
can miss the fair benchmark (\autoref{ex:cutoff-da}).

TKDA supplies a strategy-proof comparison under the same employment
priorities. Greedy improves 66.0 percent of families relative to TKDA with one
resource and 70.8 percent with three, worsening none. KDA improves 63.6 and
59.0 percent, also worsening none (\autoref{tab:tkda-pairwise} in the Online
Appendix). These gaps
measure the comparison with one strategy-proof mechanism, not an unavoidable
cost of strategy-proofness.

Among all matchings that leave no family worse off than under KDA, the largest
attainable number of beneficiaries averages 66.7 families with one resource and
123.2 with three; greedy realizes 48.9 and 70.1 percent of those ceilings on
average and attains neither in any run. The corresponding placement captures
are 27.7 and 67.4 percent (\autoref{sec:oracle}). A larger beneficiary count
does not, however, mean that the maximizing matching Pareto dominates greedy:
different improvements benefit different families, and the
beneficiary-maximizing matching reaches the unconditional placement ceiling in
only 9 of 400 one-resource markets and 24 of 400 three-resource markets.
Since KDA is family-optimal among individually
rational interference-free matchings, any KDA-preserving optimizer also weakly
Pareto dominates every such matching, including every individually rational
fair matching.

\subsubsection{Incentive diagnosis and priority design}\label{sec:hias-incentives}

With one resource, a family's requirement is its total size, so \autoref{prop:ascending}
applies: ranking compatible families by ascending size, rather than by
predicted employment, makes greedy choice substitutable and size monotone at
every locality. The code breaks size ties by family index at every locality,
giving a common strict order. Greedy COP-and-Choose is therefore serial
dictatorship in that order (\autoref{cor:ascending-efficiency} and
\autoref{prop:common-priority-efficiency}), and is strategy-proof,
Pareto efficient, and proposal-order independent. In the calibration, all 12 sampled proposal orders
then produce the same outcome in each of the 400 preference draws, whereas
under employment-based priorities they produce more than one outcome in every
draw. The reordering has distributional consequences. Relative to the outcome
under employment-based priorities, ascending-size priority increases the number
of families placed by 3.1 percent and improves mean simulated-preference rank
from 5.04 to 3.46, but lowers total predicted employment by 3.6 percent and
shifts placement toward smaller families: on average, 6.5 fewer families of
five or more members are placed per draw (\autoref{sec:priority-design-detail}
reports the full incidence by family size). The same prescription does not apply to the
three-dimensional specification, because the requirement vectors are not
totally ordered componentwise.

Under the employment-based priorities the calibration exhibits actual
manipulation as well as failure of the robust incentive conditions.
\citet[Proposition~3]{DKT2023} already show that KDA is not strategy-proof
because the choice function it induces at localities fails cardinal
monotonicity; the diagnosis here is the analogous one for greedy, made exact by
\autoref{thm:robust-boundary}. Under the fixed rule that the lowest-indexed active
family proposes next, a finite search finds profitable reports in 10 of 400
one-resource profiles (2.5 percent) and 16 of 400 three-resource profiles
(4.0 percent). For each profile, the search tests 20 poorly assigned families
and a subset of their possible reports. These rates are therefore lower
bounds on manipulable profiles in the calibrated sample, not estimates of
manipulability over all reports. \autoref{sec:deviation-search} gives the
report classes, witnesses, and replication protocol.

The domain diagnosis is distinct. The calibrated constraints fail the matroid
condition that guarantees robust order independence and strategy-proofness of
greedy COP-and-Choose for every priority ranking. At every locality under both
resource specifications, some compatible-family set has maximal feasible
subsets of different cardinalities; hence the corresponding feasibility
collection is not a matroid. The employment-based priorities also fail the
fixed-priority condition: direct checks show that greedy choice is not
substitutable at any calibrated locality under either resource specification.
Hence no calibrated constraint-priority pair is robustly strategy-proof.

The Los Angeles configuration in the Introduction reproduces
\autoref{ex:motivating}, and hence its substitutability violation. The
replication package supplies witnesses for every other locality. By itself,
failure of the robust conditions would not establish manipulation in this
network: the necessity theorem constructs a surrounding market rather than
testing a given preference profile.

\section{Conclusion}

General upper bounds make priority protection costly even when additional
assignments are feasible. COP-and-Choose separates that cost from feasibility:
cutoff implements the student-optimal fair benchmark, while greedy weakly
Pareto dominates every individually rational and fair matching under every
proposal order. Reference protection identifies exactly which choice rules
preserve this comparison and which matchings they protect.

The incentive results delimit when these welfare gains are compatible with
strategy-proofness. At a fixed priority, greedy requires substitutability and
size monotonicity for a robust guarantee; across priorities, the exact
constraint class is matroids. Inside the domain, greedy is not merely
undominated but dominant among the strategy-proof substitutable rules that share its
welfare guarantee. Outside that domain, a surrounding-market
construction rules out a strategy-proof mechanism weakly dominating greedy.
These are robust boundaries, not pointwise impossibility statements about
every market outside the domain.

The calibration shows that the welfare margin can be substantial, but so can
the priority violations. It also supplies direct manipulation witnesses and
shows the distributional cost of restoring the one-resource incentive
guarantee through priority design. The design question is therefore which
priority claims an institution is willing to relax, and what incentive and
distributional safeguards it requires in return.
\appendix

\section{Proofs and Supporting Results}
\label{app:proofs}

This appendix follows the order of the main text. Within each subsection,
supporting results are grouped with the main-text proof that uses them.

\subsection{Choice rules and the fair benchmark}

\begin{proof}[Proof of \autoref{prop:greedy-characterization}]
Fix $I\subseteq\I$. By construction, $C_s^g(I)\in\cA_s$, so $C_s^g$ is a
choice rule. Suppose student $i\in I$ is skipped, and let $X$ be the set
selected before $i$ is considered. Then $X\cup\{i\}\notin\cA_s$. Since
$X\subseteq C_s^g(I)$, feasibility of $C_s^g(I)\cup\{i\}$ would imply
feasibility of $X\cup\{i\}$ by downward closure. Hence
$C_s^g(I)\cup\{i\}\notin\cA_s$, so greedy choice is non-wasteful.

Now suppose that $i$ has higher priority than every student in a nonempty
$I'\subseteq C_s^g(I)$. Every student in $X$ has higher priority than $i$, so
$X\cap I'=\emptyset$. Thus, $X\subseteq C_s^g(I)\setminus I'$. If
$(C_s^g(I)\setminus I')\cup\{i\}$ were feasible, downward closure would imply
that $X\cup\{i\}$ is feasible, a contradiction. Hence greedy choice is weakly
priority-respecting.

For uniqueness, let $C$ be a non-wasteful and weakly priority-respecting
choice rule, and suppose that $C(I)\neq C_s^g(I)$. Neither chosen set can be a
proper subset of the other. Indeed, if one were, an element of the larger
feasible set could be added to the smaller one, contradicting
non-wastefulness. Hence both $C(I)\setminus C_s^g(I)$ and
$C_s^g(I)\setminus C(I)$ are nonempty.

Let $i$ be the highest-priority student in $C(I)\setminus C_s^g(I)$ and let
$j$ be the highest-priority student in $C_s^g(I)\setminus C(I)$. Suppose first
that $i\mathrel{\pi_s}j$. Every student selected by greedy before $i$ must
belong to $C(I)$; otherwise such a student would belong to
$C_s^g(I)\setminus C(I)$ and have higher priority than $j$. Let $X$ be the set
selected by greedy before $i$. Then $X\cup\{i\}\subseteq C(I)$, so
$X\cup\{i\}\in\cA_s$ by downward closure. Greedy choice must therefore select
$i$, a contradiction.

Suppose instead that $j\mathrel{\pi_s}i$. Since $i$ is the highest-priority
student in $I'=C(I)\setminus C_s^g(I)$, student $j$ has higher priority than
every student in $I'$. Moreover,
$(C(I)\setminus I')\cup\{j\}=(C(I)\cap C_s^g(I))\cup\{j\}\subseteq C_s^g(I)$,
so this set is feasible. This contradicts weak priority respect.
\end{proof}

\begin{proof}[Proof of \autoref{prop:cutoff-characterization}]
Fix $I\subseteq\I$. By construction, $C_s^c(I)\in\cA_s$, so $C_s^c$ is a
choice rule, and it is priority-respecting. If $C_s^c(I)\neq I$, let $i$ be
the highest-priority student in $I\setminus C_s^c(I)$. Then
$C_s^c(I)\cup\{i\}\notin\cA_s$, since $i$ is the first student whose addition
would violate feasibility. Hence $C_s^c$ is weakly non-wasteful.

For uniqueness, let $C$ be a weakly non-wasteful and priority-respecting
choice rule. Priority respect implies that $C(I)$ consists of the
highest-priority students in $I$. Let $m=|C(I)|$ and
$k=|C_s^c(I)|$. If $m<k$, the highest-priority student in
$I\setminus C(I)$ can be added feasibly to $C(I)$, contrary to weak
non-wastefulness. If $m>k$, then $C(I)$ contains the first $k+1$ students in
the priority order, but that set is infeasible by the definition of cutoff
choice. This contradicts feasibility of $C(I)$. Hence $m=k$, and therefore
$C(I)=C_s^c(I)$.
\end{proof}

The supporting results below establish cutoff substitutability, its process
implications, and protection of fair matchings. They appear in dependency
order immediately before the implementation proof that combines them.

\begin{lemma}\label{lem:cutoff-rejection}
Cutoff choice is substitutable.
\end{lemma}

\begin{proof}[Proof of \autoref{lem:cutoff-rejection}]
For $I\subseteq\I$ and $i\in I$, write $P_s(i;I)=(I\cap G_s(i))\cup\{i\}$
for the \emph{priority prefix} of $I$ through $i$: the students in $I$ whose
priority at $s$ is at least as high as $i$'s. This notation is used only in
the proofs. Fix $s\in\Sset$ and $I\subseteq I'\subseteq\I$. For every
$i\in I$, downward closure gives
\begin{equation}\label{eq:cutoff-prefix-test}
  i\in C_s^c(I)
  \quad\Longleftrightarrow\quad
  P_s(i;I)\in\cA_s.
\end{equation}
Indeed, if $P_s(i;I)$ is feasible, then every earlier priority prefix is
feasible, so cutoff choice reaches and selects $i$. The converse is immediate
from the definition of cutoff choice.

Now let $i\in C_s^c(I')\cap I$. By
\eqref{eq:cutoff-prefix-test}, $P_s(i;I')$ is feasible. Since
$P_s(i;I)\subseteq P_s(i;I')$, downward closure implies that
$P_s(i;I)$ is feasible. Another application of
\eqref{eq:cutoff-prefix-test} gives $i\in C_s^c(I)$. Hence
$C_s^c(I')\cap I\subseteq C_s^c(I),$
so cutoff choice is substitutable.
\end{proof}

\begin{proposition}
\label{prop:substitutability-cop}
Suppose every school's choice rule is substitutable. If a student who has
proposed to a school is not selected from its cumulative set of proposers at
some stage of COP-and-Choose, then that student is never selected by that
school at any later stage. Moreover, the sets selected by different schools are disjoint at every
stage. Consequently, each student receives at most one terminal offer, the
final student-choice step is redundant, and the outcome is independent of the
proposal-order convention.
\end{proposition}

\begin{proof}[Proof of \autoref{prop:substitutability-cop}]
Pairwise disjointness holds initially. Suppose it holds just before an active
student $i\in\I$ proposes to school $s\in\Sset$. Student $i$ is then selected
by no school, and only $s$'s cumulative set of proposers changes, from $T_s$
to $T_s\cup\{i\}$. Substitutability gives
$C_s(T_s\cup\{i\})\subseteq C_s(T_s)\cup\{i\},$
so the selected sets remain pairwise disjoint.

Next, suppose student $i$ has proposed to $s$ and is not selected from some
cumulative set $T_s$. For every later cumulative set $T'_s\supseteq T_s$, if
$i\in C_s(T'_s)$, substitutability would imply
$i\in C_s(T'_s)\cap T_s\subseteq C_s(T_s)$, a contradiction. Thus, a student
who is rejected by a school is never selected by that school later.

At termination, the schools' selected sets are therefore pairwise disjoint, so
each student receives at most one offer and the final student-choice step is
redundant.

For order independence, fix two runs at the same preference profile, and let
$(U_s)_{s\in\Sset}$ be the second run's terminal proposal sets. Suppose the
first run includes a proposal absent from the second, and consider the first
such proposal, by student $i$ to school $s$. Immediately before it, every
cumulative proposal set $T_{s'}$ in the first run is contained in $U_{s'}$.
Student $i$ has already proposed to every school $s'$ she prefers to $s$ and,
being active, is not selected there. Thus
$i\in T_{s'}\setminus C_{s'}(T_{s'})$, and substitutability gives
$i\notin C_{s'}(U_{s'})$ for each such school. In the second run, $i$ never
proposes to $s$ and therefore cannot propose to any less-preferred school.
She consequently has no terminal offer, although the acceptable school $s$
remains untried, contradicting termination. Hence every proposal in the first
run also occurs in the second. Reversing the argument gives identical terminal
proposal sets and therefore identical outcomes.
\end{proof}

\begin{lemma}
\label{lem:cutoff-protects-fair}
Cutoff choice protects every individually rational and fair matching.
\end{lemma}

\begin{proof}[Proof of \autoref{lem:cutoff-protects-fair}]
Fix a student-preference profile, an individually rational and fair matching
$\mu^0$, a school $s\in\Sset$, a set $I\subseteq D_s(\mu^0)$, and
$i\in\mu^0(s)\cap I$. Every $j\in I$ with
$j\mathrel{\pi_s}i$ must belong to $\mu^0(s)$. Otherwise,
$I\subseteq D_s(\mu^0)$ gives $s\succeq_j\mu^0(j)$, and since
$j\notin\mu^0(s)$ and preferences are strict,
$s\succ_j\mu^0(j)$. Thus, $j$ has justified envy toward $i$, contradicting
fairness.

Hence the priority prefix $P_s(i;I)$ is contained in the feasible set
$\mu^0(s)$ and is therefore feasible by downward closure.
By \eqref{eq:cutoff-prefix-test}, $i\in C_s^c(I)$.
\end{proof}

\begin{proof}[Proof of \autoref{thm:cutoff-implementation}]
By \autoref{lem:cutoff-rejection}, cutoff choice is substitutable.
\autoref{prop:substitutability-cop} therefore implies that a student rejected
by a school is never selected there later, the schools' selected sets remain
pairwise disjoint, the final student-choice step is redundant, and the outcome
is independent of the proposal-order convention. Feasibility and individual
rationality follow from the definition of COP-and-Choose.

Let $\mu^c$ denote the resulting matching. Fix $s\in\Sset$ and suppose that
$i$ has justified envy toward $i'\in\mu^c(s)$. Since
$s\succ_i\mu^c(i)$ and students propose in preference order, student $i$ must
have proposed to $s$. Thus, the terminal set of proposers at $s$ contains both
$i$ and $i'$. Because $i\mathrel{\pi_s}i'$ and cutoff choice selects $i'$,
it also selects $i$. Pairwise disjointness of the terminal selected sets and
redundancy of the final student-choice step then imply $\mu^c(i)=s$, a
contradiction. Hence $\mu^c$ is fair.

By \autoref{lem:cutoff-protects-fair}, cutoff choice protects every
individually rational and fair matching. \autoref{thm:protection} therefore
implies that $\mu^c$ weakly Pareto dominates every such matching. Since
$\mu^c$ is itself individually rational and fair, it is the student-optimal
fair matching. This matching is unique under strict student preferences:
any two student-optimal fair matchings weakly Pareto dominate one another and
therefore assign every student identically.
\end{proof}

\subsection{Greedy COP-and-Choose: welfare}

\begin{proof}[Proof of \autoref{thm:protection}]
We first show by induction that whenever a student $i$ proposes to a school
$s$, she weakly prefers $s$ to her assignment under $\mu^0$. The claim holds
for a student's first proposal because it goes to her most-preferred acceptable
school and $\mu^0$ is individually rational. Suppose the claim holds for all
previous proposals and an active student $i$ is about to make a later proposal
to $s$.
If $s\prec_i\mu^0(i)$, then $\mu^0(i)=s'\in\Sset$. Since students
propose in preference order, $i$ must already have proposed to $s'$. By the
induction hypothesis, the current cumulative proposal set at $s'$ is contained
in $D_{s'}(\mu^0)$. Since $i\in\mu^0(s')$ and $i$ has proposed to
$s'$, protection implies that $s'$ currently selects $i$. This contradicts the
fact that $i$ is active.

Now consider a student $i\in\I$ at termination. If
$\mu^0(i)=\emptysetmatch$, individual rationality of the COP-and-Choose
outcome gives $\mu^{C,\rho}(i)\succeq_i\mu^0(i)$. Suppose instead that
$\mu^0(i)=s$. If $i$ proposed to $s$, then
$T_s\subseteq D_s(\mu^0)$ by the first part of the proof. Since
$i\in\mu^0(s)\cap T_s$, protection gives $i\in C_s(T_s)$. Thus, $i$
receives a terminal offer from $s$ and chooses an assignment weakly preferred
to $s$.

If $i$ did not propose to $s$, any school from which she receives a terminal
offer must be preferred to $s$, since students propose in preference order.
She must receive such an offer at termination; otherwise she would be active
with the acceptable school $s$ still untried. Hence
$\mu^{C,\rho}(i)\succeq_i\mu^0(i)$ for every $i\in\I$.
\end{proof}

\begin{proof}[Proof of \autoref{cor:fair-dominance}]
By \eqref{eq:nesting}, greedy choice is cutoff-containing.
\autoref{thm:exact-fair-dominance} therefore implies that every greedy
COP-and-Choose outcome weakly Pareto dominates every individually rational and
fair matching. By \autoref{thm:cutoff-implementation}, the cutoff outcome is
individually rational and fair, so greedy weakly Pareto dominates it.

If the greedy and cutoff outcomes differ, some student receives different
assignments. Since preferences are strict and no student is worse off under
greedy, that student is strictly better off. Hence the greedy outcome Pareto
dominates the cutoff outcome.
\end{proof}

\begin{proof}[Proof of \autoref{thm:exact-fair-dominance}]
Suppose first that (i) holds. By \autoref{lem:cutoff-protects-fair}, cutoff
choice protects every individually rational and fair matching. Cutoff
containment therefore implies that $C$ protects every such matching, proving
(ii). Statement (ii) and \autoref{thm:protection} imply (iii).

It remains to prove that (iii) implies (i). Suppose (i) fails. Then, for some
school $s\in\Sset$, set $I\subseteq\I$, and student $i\in I$,
$i\in C_s^c(I)\setminus C_s(I)$. Let $P=C_s^c(I)$. Give every student in
$I$ a preference under which $s$ is the only acceptable school, and make every
school unacceptable to every student outside $I$. Let $\mu^0$ assign the
students in $P$ to $s$ and leave everyone else unmatched.

The matching $\mu^0$ is individually rational and feasible. It is also fair:
$P$ is a top-priority prefix of $I$, so every student in $I\setminus P$ has
lower priority at $s$ than every student in $P$, while students outside $I$
find $s$ unacceptable.

Under every proposal-order convention, every student in $I$ eventually
proposes to $s$, since an unproposed student remains active with her only
acceptable school untried. Hence the terminal set of proposers at $s$ is
$I$. Because $i\notin C_s(I)$, student $i$ receives no terminal offer from
$s$ and therefore remains unmatched, whereas $\mu^0(i)=s$. Thus, the
COP-and-Choose outcome does not weakly Pareto dominate $\mu^0$, contradicting
(iii). Hence (iii) implies (i).
\end{proof}

\begin{lemma}\label{lem:greedy-consistency}
Every greedy choice rule is consistent.
\end{lemma}

\begin{proof}[Proof of \autoref{lem:greedy-consistency}]
Fix $s\in\Sset$ and $I,I'\subseteq\I$ with
$C_s^g(I)\subseteq I'\subseteq I$. Every student in $I\setminus I'$ was
skipped when greedy scanned $I$. Removing such students does not change the
set accumulated before any remaining student is considered, and therefore
does not change any subsequent feasibility decision. Hence the greedy runs on
$I$ and $I'$ select the same students, so
$C_s^g(I')=C_s^g(I)$.\qedhere
\end{proof}

\begin{proof}[Proof of \autoref{prop:cutoff-containing-uniqueness}]
Greedy is cutoff-containing by \eqref{eq:nesting} and consistent by
\autoref{lem:greedy-consistency}. Conversely, let $C_s$ be consistent and
cutoff-containing. We prove $C_s(I)=C_s^g(I)$ by induction on $|I|$.
If $I$ is feasible, cutoff containment gives $C_s(I)=I=C_s^g(I)$.
Otherwise, let $j$ be the first student skipped by greedy and let $P$ be the
priority prefix preceding $j$. Then $P=C_s^c(I)\subseteq C_s(I)$ and
$P\cup\{j\}$ is infeasible. Downward closure and feasibility of $C_s(I)$
imply $j\notin C_s(I)$. Consistency gives
$C_s(I)=C_s(I\setminus\{j\})$, while deleting a skipped student leaves
greedy choice unchanged. The induction hypothesis therefore yields
\[
  C_s(I)=C_s(I\setminus\{j\})
        =C_s^g(I\setminus\{j\})=C_s^g(I).
\]
\end{proof}

\begin{lemma}\label{lem:greedy-prefix}
For every $s\in\Sset$, every $I\subseteq\I$, and every $i\in I$,
\[
  C_s^g(I)\cap G_s(i)=C_s^g(I\cap G_s(i)),
\]
and
\[
  i\in C_s^g(I)
  \quad\Longleftrightarrow\quad
  C_s^g(I\cap G_s(i))\cup\{i\}\in\cA_s.
\]
\end{lemma}

\begin{proof}[Proof of \autoref{lem:greedy-prefix}]
Fix $s\in\Sset$, $I\subseteq\I$, and $i\in I$. The greedy runs on $I$ and on
$I\cap G_s(i)$ encounter the same higher-priority students in the same order
and therefore make the same decisions on them. Hence
$C_s^g(I)\cap G_s(i)=C_s^g(I\cap G_s(i)).$
The set selected before $i$ is considered in the run on $I$ is therefore
$C_s^g(I\cap G_s(i))$. By the definition of greedy choice, $i$ is selected
if and only if
$C_s^g(I\cap G_s(i))\cup\{i\}\in\cA_s.$
\end{proof}

\begin{proof}[Proof of \autoref{thm:dual-protection}]
Fix an individually rational matching $\mu^0$ and a school $s\in\Sset$.

\emph{(i).} Suppose first that $\mu^0$ is claim-compatible. Fix
$I\subseteq D_s(\mu^0)$ and
$i\in\mu^0(s)\cap I$. The priority prefix $P_s(i;I)$ is contained in
$H_s^{\mu^0}(i)\cup\{i\}$, which is feasible by claim compatibility. Downward
closure therefore makes $P_s(i;I)$ feasible, so cutoff choice selects $i$.
Thus, cutoff choice protects $\mu^0$.

Conversely, suppose cutoff choice protects $\mu^0$. Fix
$i\in\mu^0(s)$ and take
$I=H_s^{\mu^0}(i)\cup\{i\}$. Then $I\subseteq D_s(\mu^0)$, so protection
implies $i\in C_s^c(I)$. Since every other member of $I$ has higher priority
than $i$, the priority prefix through $i$ is all of $I$. Hence $I$ is
feasible, proving claim compatibility.

\emph{(ii).} Suppose next that $\mu^0$ is weakly claim-compatible. Fix
$I\subseteq D_s(\mu^0)$ and
$i\in\mu^0(s)\cap I$, and let
$X=C_s^g(I\cap G_s(i))$. Then $X$ is feasible and
$X\subseteq H_s^{\mu^0}(i)$. Weak claim compatibility therefore makes
$X\cup\{i\}$ feasible, and \autoref{lem:greedy-prefix} gives
$i\in C_s^g(I)$. Thus, greedy choice protects $\mu^0$.

Conversely, suppose greedy choice protects $\mu^0$. Fix
$i\in\mu^0(s)$ and a feasible set
$I'\subseteq H_s^{\mu^0}(i)$. Take $I=I'\cup\{i\}$. Every member of $I'$ has
higher priority than $i$, and because $I'$ is feasible, greedy selects every
member of $I'$. Protection also gives $i\in C_s^g(I)$. By
\autoref{lem:greedy-prefix}, $I'\cup\{i\}$ is therefore feasible. Hence
$\mu^0$ is weakly claim-compatible.
\end{proof}

\subsection{Greedy COP-and-Choose: incentives and the Pareto frontier}

\begin{proof}[Proof of \autoref{lem:sp-nw-frontier}]
Let $g$ be an individually rational, non-wasteful, and strategy-proof
mechanism. Suppose, toward a contradiction, that a strategy-proof mechanism
$f$ weakly Pareto dominates $g$ at every preference profile and strictly
improves some student at some profile. Thus, for some preference profile
$\succeq^0$, student $i_0$, and school $s_0$,
\[
  s_0=f_{i_0}(\succeq^0)
  \succ^0_{i_0}
  g_{i_0}(\succeq^0).
\]

We construct a sequence recursively. Suppose that, at profile $\succeq^t$,
\[
  f_{i_t}(\succeq^t)=s_t
  \succ^t_{i_t}
  g_{i_t}(\succeq^t).
\]
Let $\succeq^{t+1}$ differ from $\succeq^t$ only in the preference of
$i_t$, who now declares $s_t$ to be her unique acceptable school.
Strategy-proofness of $f$, applied with $\succeq^{t+1}_{i_t}$ as the true
preference, implies
$f_{i_t}(\succeq^{t+1})=s_t$: otherwise $i_t$ could report her previous
preference and obtain her uniquely preferred acceptable school $s_t$.

Strategy-proofness of $g$, now applied with $\succeq^t_{i_t}$ as the true
preference, implies
$g_{i_t}(\succeq^{t+1})\neq s_t$. Otherwise the truncation would be a
profitable deviation from $g_{i_t}(\succeq^t)$ to $s_t$. Since $g$ is
individually rational and $s_t$ is the only acceptable school under
$\succeq^{t+1}_{i_t}$,
$g_{i_t}(\succeq^{t+1})=\emptysetmatch.$

Let $N_t=g(\succeq^{t+1})(s_t)$. Since $i_t$ is unmatched under $g$ and
strictly prefers $s_t$ to being unmatched, non-wastefulness of $g$ gives
$N_t\cup\{i_t\}\notin\cA_{s_t}.$
On the other hand, $f(\succeq^{t+1})(s_t)$ is feasible and contains $i_t$.
Therefore some student in $N_t$ cannot remain at $s_t$ under $f$; otherwise
$N_t\cup\{i_t\}$ would be a subset of the feasible set
$f(\succeq^{t+1})(s_t)$ and hence feasible by downward closure.

Choose such a student $i_{t+1}\in N_t$. Under $g$,
$i_{t+1}$ receives $s_t$. Since $f$ weakly Pareto dominates $g$ at
$\succeq^{t+1}$ and does not assign $i_{t+1}$ to $s_t$, strict preferences
imply
\[
  f_{i_{t+1}}(\succeq^{t+1})
  \succ^{t+1}_{i_{t+1}}
  s_t
  =
  g_{i_{t+1}}(\succeq^{t+1}).
\]
Set
$s_{t+1}=f_{i_{t+1}}(\succeq^{t+1})$ and continue.

The students $i_0,i_1,\ldots$ are all distinct. Indeed, once a student
$i_r$ has been truncated, $s_r$ remains her unique acceptable school at all
later profiles. If she belonged to some later set $N_t$, individual
rationality of $g$ would imply that the school to which $g$ assigns her is
$s_r$. But no assignment is strictly preferred to $s_r$ under her truncated
preference, so she could not be chosen as $i_{t+1}$.
Because every later profile retains that truncation, $i_r$ therefore cannot
reappear as a later strict beneficiary.

The construction therefore produces infinitely many distinct students,
contradicting finiteness of $\I$. Hence no strategy-proof mechanism Pareto
dominates $g$.
\end{proof}

The following choice facts connect substitutability to path independence for
greedy choice and supply the standard sufficiency result used in the robust
strategy-proofness proof. A choice rule $C_s$ is \emph{path independent} if
\[
  C_s(I\cup I')=C_s(C_s(I)\cup I')
  \qquad\text{for all }I,I'\subseteq\I.
\]
The following equivalence is due to \citet{AizermanMalishevski1981}.

\begin{lemma}\label{lem:aizerman}
A choice rule is path independent if and only if it is consistent and
substitutable.
\end{lemma}

\begin{lemma}
\label{lem:sub-sm-sufficient}
If every school's choice rule is substitutable and size monotone, then once a
student is rejected by a school she is never selected by that school later, the
final student-choice step is redundant, and COP-and-Choose is
proposal-order independent and strategy-proof.
\end{lemma}

\begin{proof}[Proof of \autoref{lem:sub-sm-sufficient}]
The assumptions imply consistency: if $C_s(I)\subseteq J\subseteq I$,
substitutability gives $C_s(I)\subseteq C_s(J)$, while size monotonicity gives
$|C_s(J)|\leq |C_s(I)|$. Hence $C_s(J)=C_s(I)$.
\autoref{prop:substitutability-cop} implies that a student rejected by a
school is never selected there later, the schools' selected sets remain
pairwise disjoint, and the final student-choice step is redundant.
COP-and-Choose therefore coincides with the ordinary cumulative-offer process
with substitutable, size-monotone choice rules. Proposal-order independence
and strategy-proofness then follow from the standard deferred-acceptance
results; see
\citet{HatfieldMilgrom2005,HatfieldKominersWestkamp2021}.
\end{proof}

\begin{proof}[Proof of \autoref{lem:greedy-sub-nw}]
\autoref{prop:substitutability-cop} implies that a student rejected
by a school is never selected there later, the terminal selected sets are
pairwise disjoint, and the final student-choice step is redundant.

Let $\mu$ be the resulting greedy outcome and suppose
$s\succ_i\mu(i)$. Student $i$ must have proposed to $s$: if she is matched,
she proposes to every more-preferred school before reaching her assignment,
and if she is unmatched, she exhausts all acceptable schools. Thus,
$i\in T_s$. Since the final student-choice step is redundant,
$\mu(s)=C_s^g(T_s).$
Because $\mu(i)\neq s$, we have $i\notin C_s^g(T_s)$. By
\autoref{prop:greedy-characterization}, greedy choice is non-wasteful, so
$C_s^g(T_s)\cup\{i\}\notin\cA_s.$
Hence
$\mu(s)\cup\{i\}\notin\cA_s$. Since this holds whenever
$s\succ_i\mu(i)$, the matching $\mu$ is non-wasteful.
\end{proof}

\begin{proof}[Proof of \autoref{thm:robust-boundary}(i)]
If $C^\pi_{\cA}$ is substitutable, the distinguished school and every other
school use substitutable choice rules. \autoref{prop:substitutability-cop} then
implies
that the COP-and-Choose outcome is independent of the proposal-order
convention.

For the converse, let $C=C^\pi_{\cA}$ and suppose that $C$ is not
substitutable. Thus, there are
$K\subseteq L\subseteq I$ and $u\in K$ such that
$u\notin C(K)$ but $u\in C(L)$. Add the students in $L\setminus K$ one at a
time. At the first addition that causes $u$ to be selected, there are
$M\subseteq I$ and $j\in I\setminus M$ such that
\begin{equation}\label{eq:revival}
  u\in M,\qquad
  u\notin C(M),\qquad
  u\in C(M\cup\{j\}).
\end{equation}

The set $C(M)$ is inclusion-maximal feasible in $M$. Indeed, if
$x\in M\setminus C(M)$, greedy skipped $x$ when the accumulated selected set
was some $P\subseteq C(M)$ with $P\cup\{x\}\notin\cA$. If
$C(M)\cup\{x\}$ were feasible, downward closure would make
$P\cup\{x\}$ feasible, a contradiction.

Consequently, some student selected from $M$ must be displaced when $j$ is
added. That is, there is
\begin{equation}\label{eq:displaced}
  v\in C(M)\setminus C(M\cup\{j\}).
\end{equation}
Otherwise $C(M)\subseteq C(M\cup\{j\})$; together with
$u\in C(M\cup\{j\})$, this would make
$C(M)\cup\{u\}$ feasible, contradicting the maximality of $C(M)$.
The students $u,v,j$ are distinct.

Construct a market with schools $a$ and $b$. School $a$ uses the greedy rule
generated by $(\cA,\pi)$, extended to the full student set by making students
outside $I$ loops.\footnote{A student $i$ is a \emph{loop} of a feasibility
collection $\cA$ if $\{i\}\notin\cA$, so no feasible set contains her. Making
the students outside $I$ loops extends $\cA$ from $I$ to $\I$ without creating
any new feasible set.} School $b$ has capacity one and ranks $u$ above $v$.
Students $u$ and $v$ rank $a$ first and $b$ second. Every other student in
$M\cup\{j\}$ finds only $a$ acceptable, and every student outside
$M\cup\{j\}$ finds neither school acceptable.

Consider first a proposal-order convention under which every student in
$(M\setminus\{u\})\cup\{j\}$ proposes to $a$ before $u$, after which $u$
proposes to $a$, and rejected students move only afterward. The terminal set
of proposers at $a$ is $M\cup\{j\}$. By
\eqref{eq:revival}--\eqref{eq:displaced}, $a$ selects $u$ and rejects $v$.
Student $v$ then proposes to $b$ and is matched there.

Now consider a convention under which every student in $M$ proposes to $a$
before $j$. School $a$ then selects $C(M)$ and rejects $u$, who proposes to
$b$. When $j$ later proposes to $a$, its cumulative set of proposers becomes
$M\cup\{j\}$. School $a$ now selects $u$ and rejects $v$. Student $v$
proposes to $b$, but $b$ retains $u$. At the final student-choice step, $u$
chooses $a$, so $v$ remains unmatched.

School $a$ has the same terminal set of proposers, $M\cup\{j\}$, under both
conventions, and every student other than $v$ receives the same final
assignment. Student $v$ is assigned to $b$ under the first convention and is
unmatched under the second. Hence the outcomes differ only in $v$'s
assignment, and the first strictly Pareto dominates the second.
\end{proof}

\begin{proof}[Proof of \autoref{thm:robust-boundary}(ii)]
\emph{Sufficiency.}
Suppose $C^\pi_{\cA}$ is substitutable and size monotone. Extend the rule
to $\I$ by making every student in $\I\setminus I$ a loop. For every
$J\subseteq\I$, the extended greedy rule chooses
$C^\pi_{\cA}(J\cap I),$
since every student outside $I$ is skipped. Substitutability and size
monotonicity therefore carry directly to the extended rule. Every other school
has a substitutable, size-monotone choice rule by assumption, so
\autoref{lem:sub-sm-sufficient} implies strategy-proofness under every
proposal-order convention.

\emph{Necessity.}
Let $C=C^\pi_{\cA}$ and suppose that $C$ fails substitutability or size
monotonicity. Construct a market with schools $a$ and $b$. School $a$ uses
$C$, extended to $\I$ by making students outside $I$ loops. School $b$ will
have capacity one.

\emph{Case 1: $C$ is not substitutable.}
The construction for part~(i) of \autoref{thm:robust-boundary} gives
$M\subseteq I$, $j\in I\setminus M$, and $u\in M$ such that
\begin{equation}\label{eq:greedy-revival}
  u\notin C(M)
  \quad\text{and}\quad
  u\in C(M\cup\{j\}).
\end{equation}
School $b$ ranks $u$ above $j$ and has capacity one. Student $u$ finds only
$a$ acceptable, student $j$ ranks $b$ above $a$, every student in
$M\setminus\{u\}$ finds only $a$ acceptable, and every remaining student finds
neither school acceptable.

Under truthful reporting, every member of $M$ eventually proposes to $a$,
while $j$ remains at $b$. Thus, the terminal set of proposers at $a$ is $M$,
and $u$ is unmatched by \eqref{eq:greedy-revival}.

Now let $u$ report
$a\succ'_u b\succ'_u\emptysetmatch$. Since $u\notin C(M)$, she is eventually
rejected by $a$ and proposes to $b$, where she displaces $j$. Student $j$ then
proposes to $a$, so the terminal set of proposers there becomes
$M\cup\{j\}$. By \eqref{eq:greedy-revival}, $a$ now selects $u$. School $b$
also selects $u$, and at the final student-choice step her reported preference
makes her choose $a$. She therefore obtains her true first choice instead of
remaining unmatched. These eventual proposals and rejections are forced by the
reports, so the manipulation succeeds under every proposal-order convention.

\emph{Case 2: $C$ is substitutable but not size monotone.}
There are $J\subseteq J'\subseteq I$ with $|C(J)|>|C(J')|$. Adding the
students in $J'\setminus J$ one at a time, we obtain
$M\subseteq I$ and $j\in I\setminus M$ such that
$|C(M)|>|C(M\cup\{j\})|.$
We first record two consequences. If $j\notin C(M\cup\{j\})$, then
$C(M\cup\{j\})\subseteq M\subseteq M\cup\{j\}$, and consistency would imply
$C(M)=C(M\cup\{j\})$, a contradiction. Hence
$j\in C(M\cup\{j\})$. By substitutability,
$C(M\cup\{j\})\setminus\{j\}\subseteq C(M)$, so
\[
  \bigl|C(M)\setminus C(M\cup\{j\})\bigr|
  =
  |C(M)|-|C(M\cup\{j\})|+1
  \geq 2.
\]
Choose distinct
$u,v\in C(M)\setminus C(M\cup\{j\}).$

Let school $b$ have capacity one and rank
$u\mathrel{\pi_b}v\mathrel{\pi_b}j$ above every other student. Give students
$u,v,j$ preferences
\[
  a\succ_u b\succ_u\emptysetmatch,
  \qquad
  b\succ_v a\succ_v\emptysetmatch,
  \qquad
  b\succ_j a\succ_j\emptysetmatch.
\]
Every student in $M\setminus\{u,v\}$ finds only $a$ acceptable, and every
remaining student finds neither school acceptable.

Consider truthful reporting. Students $v$ and $j$ eventually propose to $b$,
which retains $v$ and rejects $j$. Thus, $j$ eventually proposes to $a$.
Before $j$ arrives, every cumulative set of proposers at $a$ is contained in
$M$. Whenever such a set $N$ contains $u$, substitutability and
$u\in C(M)$ imply
$u\in C(M)\cap N\subseteq C(N),$
so $u$ is not rejected before $j$ arrives.

Every student in $M\setminus\{v\}$ eventually proposes to $a$, so its
cumulative set of proposers reaches
$(M\setminus\{v\})\cup\{j\}$. Since
$v\notin C(M\cup\{j\})$,
\[
  C(M\cup\{j\})
  \subseteq (M\setminus\{v\})\cup\{j\}
  \subseteq M\cup\{j\}.
\]
Consistency therefore gives
$C\bigl((M\setminus\{v\})\cup\{j\}\bigr) =C(M\cup\{j\}).$
In particular, $u$ is rejected by $a$ and proposes to $b$, where she displaces
$v$. Student $v$ then proposes to $a$, whose terminal set of proposers becomes
$M\cup\{j\}$. Since $v\notin C(M\cup\{j\})$, student $v$ is unmatched under
truthful reporting.

Now let $v$ report that only $a$ is acceptable. We claim that $j$ is never
rejected by $b$ and therefore never proposes to $a$. Indeed, with $v$ absent
from $b$, student $j$ can be rejected only if $u$ proposes there, which in
turn requires $u$ to be rejected by $a$. But as long as $j$ has not proposed
to $a$, its cumulative proposer set is some $N\subseteq M$; whenever
$u\in N$, substitutability and $u\in C(M)$ imply $u\in C(N)$. Thus, $u$ is
never rejected, so $j$ remains at $b$.

Consequently the terminal set of proposers at $a$ is $M$. Since
$v\in C(M)$, student $v$ receives $a$, which she strictly prefers to being
unmatched. The same argument depends only on eventual proposals and
rejections, not on the order in which active students are chosen. Hence the
manipulation succeeds under every proposal-order convention.
\end{proof}

\begin{proof}[Proof of \autoref{cor:robust-greedy-frontier}]
\autoref{lem:sub-sm-sufficient} gives strategy-proofness and
proposal-order independence, while \autoref{lem:greedy-sub-nw} gives
non-wastefulness. \autoref{lem:sp-nw-frontier} then implies that greedy
COP-and-Choose is Pareto-undominated among strategy-proof mechanisms.
\end{proof}

\begin{proof}[Proof of \autoref{prop:greedy-dominates-sp-class}]
Let $f$ denote greedy COP-and-Choose and $g$ denote COP-and-Choose under $C$.
By \autoref{lem:sub-sm-sufficient}, $f$ is strategy-proof and proposal-order
independent, a student not selected by a school is never selected there later,
and the final student-choice step is redundant under $f$; by
\autoref{lem:greedy-sub-nw}, $f$ is non-wasteful. Since $C$ is substitutable,
\autoref{prop:substitutability-cop} gives proposal-order independence,
permanent rejection, and redundancy of the final choice step for $g$;
strategy-proofness of $g$ is a hypothesis. Call
$i$ a \emph{strict beneficiary} at $\succeq$ if $g_i(\succeq)\succ_i
f_i(\succeq)$, and suppose toward a contradiction that some profile has one.

\emph{Step 1: truncation.} There is a profile $\succeq$ with a strict
beneficiary at which every strict beneficiary reports a single acceptable
school. Given $\succeq^t$ with a strict beneficiary $i_t$ whose report is not a
singleton, let $a_t=g_{i_t}(\succeq^t)$ and let $\succeq^{t+1}$ change only
$i_t$'s report, declaring $a_t$ her unique acceptable school. Strategy-proofness
of $g$, with $\succeq^{t+1}_{i_t}$ as the true preference, gives
$g_{i_t}(\succeq^{t+1})=a_t$: otherwise individual rationality leaves her
unmatched and reporting $\succeq^t_{i_t}$ yields $a_t$. Strategy-proofness of
$f$, with $\succeq^t_{i_t}$ as the true preference, gives
$f_{i_t}(\succeq^{t+1})\neq a_t$: otherwise reporting $\succeq^{t+1}_{i_t}$ at
$\succeq^t$ yields $a_t\succ^t_{i_t}f_{i_t}(\succeq^t)$. So $i_t$ remains a
strict beneficiary, now with a singleton report. Each step converts one
non-singleton report and never reverses it, so the sequence terminates within
$|\I|$ steps, and the student truncated last is a strict beneficiary at the
terminal profile $\succeq$. If $i$ is a strict beneficiary at $\succeq$
reporting only $b$, individual rationality confines $f_i(\succeq)$ and
$g_i(\succeq)$ to $\{b,\emptysetmatch\}$, so $g_i(\succeq)=b$ and
$f_i(\succeq)=\emptysetmatch$.

\emph{Step 2: construction.} Fix a strict beneficiary $j$ at $\succeq$ and
write $a=g_j(\succeq)$, $\mu=f(\succeq)$, and $\nu=g(\succeq)$, so $j\in\nu(a)$
and $\mu(j)=\emptysetmatch$. Let $T_a$ and $U_a$ be the terminal sets of
proposers at $a$ under $f$ and under $g$. Since $j$ finds only $a$ acceptable
and is unmatched under $f$, she exhausts her acceptable schools, so $j\in T_a$.
Put $S=\mu(a)\cup(\nu(a)\cap T_a)$, and note $j\in S\setminus\mu(a)$.

Because rejections are permanent and the final choice step is redundant under
$f$, $\mu(a)=C_a^g(T_a)$; since $\mu(a)\subseteq S\subseteq T_a$,
\autoref{lem:greedy-consistency} gives
\begin{equation}\label{eq:p4-consistency}
  C_a^g(S)=\mu(a).
\end{equation}
Next, $S\subseteq U_a$. Let $k\in\mu(a)\setminus\nu(a)$. Were $\nu(k)\succ_k a$,
student $k$ would be a strict beneficiary and Step 1 would give
$\mu(k)=\emptysetmatch$, contradicting $\mu(k)=a$; hence $a\succ_k\nu(k)$ and
$k$ proposed to $a$ under $g$. Every student in $\nu(a)$ proposed to $a$ as
well, so $S\subseteq U_a$. Since students assigned to $a$ under $g$ receive a
terminal offer there, $\nu(a)\subseteq C_a(U_a)$, and substitutability of $C_a$
yields
\begin{equation}\label{eq:p4-substitutability}
  \nu(a)\cap S\subseteq C_a(U_a)\cap S\subseteq C_a(S).
\end{equation}

\emph{Step 3: contradiction.} As $\mu(a)$ and $\nu(a)$ are feasible, downward
closure gives $\{i\}\in\cA_a$ for every $i\in S$. By
\eqref{eq:p4-consistency} and $j\in S\setminus\mu(a)$, greedy skips at least
one student when scanning $S$. List $S$ in decreasing $\pi_a$-priority as
$i_1,\ldots,i_m$ and let $i_{l+1}$ be the first student skipped, so
$\{i_1,\ldots,i_l\}\in\cA_a$, $\{i_1,\ldots,i_{l+1}\}\notin\cA_a$, and
$l\geq1$. By \autoref{def:cutoff}, $C_a^c(S)=\{i_1,\ldots,i_l\}$, so cutoff
containment gives $\{i_1,\ldots,i_l\}\subseteq C_a(S)$. Were
$i_{l+1}\in C_a(S)$, downward closure would make
$\{i_1,\ldots,i_{l+1}\}$ feasible; hence $i_{l+1}\notin C_a(S)$. But
\eqref{eq:p4-consistency} places $i_{l+1}$ in
$S\setminus\mu(a)\subseteq\nu(a)\cap S$, so \eqref{eq:p4-substitutability}
gives $i_{l+1}\in C_a(S)$, a contradiction.
\end{proof}

\begin{proof}[Proof of \autoref{thm:sp-pareto-frontier}]
Use the two-school markets constructed in the proof of
\autoref{thm:robust-boundary}(ii). School $a$ uses the given greedy rule
$C=C^\pi_{\cA}$, and school $b$ is the additional capacity-one school.
Suppose, toward a contradiction, that a strategy-proof mechanism $f$ weakly
Pareto dominates greedy COP-and-Choose $g$.

\emph{Case 1: $C$ is not substitutable.}
Let $M,j,u$ satisfy \eqref{eq:greedy-revival}, and let $b$ rank $u$ above
$j$. At the truthful profile used in the proof of \autoref{thm:robust-boundary}(ii),
greedy leaves $u$ unmatched, assigns $j$ to $b$, and assigns $C(M)$ to $a$.

Weak Pareto dominance requires $f$ to leave $j$ at $b$ and every member of
$C(M)$ at $a$. Since $u\notin C(M)$ and greedy choice is non-wasteful on $M$,
$C(M)\cup\{u\}\notin\cA.$
Feasibility therefore rules out assigning $u$ to $a$, while capacity one at
$b$ rules out assigning her there. Hence $f$ also leaves $u$ unmatched.

Now let $u$ report
$a\succ'_u b\succ'_u\emptysetmatch$. Under greedy COP-and-Choose she receives
$a$, so weak Pareto dominance forces $f$ to assign her to $a$ at the reported
profile. Thus, under her true preference, $u$ can obtain $a$ by misreporting
instead of remaining unmatched, contradicting strategy-proofness of $f$.

\emph{Case 2: $C$ is substitutable but not size monotone.}
As in the proof of \autoref{thm:robust-boundary}(ii), there are
$M\subseteq I$, $j\in I\setminus M$, and distinct
$u,v\in C(M)\setminus C(M\cup\{j\}),$
with $j\in C(M\cup\{j\})$. Let $b$ rank
$u\mathrel{\pi_b}v\mathrel{\pi_b}j$, and use the preferences constructed
there. At the truthful profile $\succeq$, greedy assigns $u$ to $b$, leaves
$v$ unmatched, and assigns $C(M\cup\{j\})$ to $a$.

Let $\succeq'$ be the profile obtained when $v$ reports that only $a$ is
acceptable. The proof of \autoref{thm:robust-boundary}(ii) shows that
$g_v(\succeq')=a$. Weak Pareto dominance therefore gives
$f_v(\succeq')=a$. Since $v$ truly ranks $b$ above $a$ above being unmatched,
strategy-proofness of $f$ implies
$f_v(\succeq)\in\{a,b\}.$

Suppose first that $f_v(\succeq)=b$. Since $b$ has capacity one and greedy
assigns $u$ to $b$, weak Pareto dominance forces $f_u(\succeq)=a$. Every
member of $C(M\cup\{j\})$ must also remain at $a$: student $j$ cannot receive
$b$, and every other member finds only $a$ acceptable. But
$u\notin C(M\cup\{j\})$, so non-wastefulness of greedy choice on
$M\cup\{j\}$ gives
\[
  C(M\cup\{j\})\cup\{u\}\notin\cA,
\]
a contradiction.

Hence $f_v(\succeq)=a$. Similarly,
$C(M\cup\{j\})\cup\{v\}\notin\cA.$
All members of $C(M\cup\{j\})$ other than $j$ find only $a$ acceptable, so
feasibility and weak Pareto dominance force $j$ to receive $b$. Since greedy
assigns $u$ to $b$ and $b$ has capacity one, weak Pareto dominance then forces
$f_u(\succeq)=a$.

Let $\succeq''$ be obtained from $\succeq$ by having $u$ report that only
$a$ is acceptable. If $f_u(\succeq'')\neq a$, then, treating
$\succeq''_u$ as her true preference, $u$ could report her preference at
$\succeq$ and obtain $a$. Strategy-proofness therefore implies
$f_u(\succeq'')=a$.

Under greedy COP-and-Choose at $\succeq''$, student $v$ receives $b$ and
school $a$ assigns $C(M\cup\{j\})$. Weak Pareto dominance forces $v$ to
remain at $b$ and every member of $C(M\cup\{j\})$ to remain at $a$: each
member other than $j$ finds only $a$ acceptable, and $j$ cannot receive $b$
because $b$ is occupied by $v$. Together
with $f_u(\succeq'')=a$, feasibility would therefore require
$C(M\cup\{j\})\cup\{u\}\in\cA,$
contradicting the displayed infeasibility above.

Thus, no strategy-proof mechanism weakly Pareto dominates greedy
COP-and-Choose in either case.
\end{proof}

\begin{proof}[Proof of \autoref{prop:cutoff-incentives}]
Under capacity constraints, cutoff and greedy choice coincide, so cutoff
COP-and-Choose is ordinary student-proposing deferred acceptance and is
strategy-proof.

Conversely, let $\cA$ be a non-capacity feasibility collection. Since $\cA$
is downward closed, it is a general upper-bound in the terminology of
\citet{KamadaKojima2024}. Apply their Theorem~5 to a two-school market in
which one school has constraint $\cA$. Their construction in Online
Appendix~A.5 keeps the student set fixed and gives the other school capacity
one, yielding a priority profile for which the student-optimal fair matching
mechanism is not strategy-proof.
By \autoref{thm:cutoff-implementation}, cutoff COP-and-Choose implements the
student-optimal fair matching at every preference profile. Hence cutoff
COP-and-Choose is not strategy-proof in this two-school market.
\end{proof}

\begin{lemma}
\label{lem:matroid-known}
Greedy choice is substitutable for every strict priority ranking if and
only if the feasibility collection is a matroid.
\end{lemma}

\citet{Fleiner2003} shows that matroid feasibility implies path independence
for every priority ranking, while \citet{HKYY2026} establish the converse.
The stated formulation follows from greedy consistency and
\autoref{lem:aizerman}.

\begin{proof}[Proof of \autoref{cor:greedy-incentives}]
Suppose first that every school has a matroid constraint. By
\autoref{lem:matroid-known}, greedy choice is substitutable under every
priority ranking. Moreover, for every school $s$ and
set $I\subseteq\I$,
greedy choice selects a basis of the restriction of $\cA_s$ to
$I$.\footnote{A \emph{basis} of the restriction of $\cA_s$ to $I$ is a maximal
feasible subset of $I$. In a matroid all such sets have the same cardinality,
called the \emph{rank} of $I$ and written $\operatorname{rank}_{\cA_s}(I)$,
which is monotone in $I$.} Hence
\[
  |C_s^g(I)|=\operatorname{rank}_{\cA_s}(I),
\]
so monotonicity of matroid rank implies size monotonicity. Therefore
\autoref{lem:sub-sm-sufficient} gives strategy-proofness and proposal-order
independence, and \autoref{cor:robust-greedy-frontier} gives non-wastefulness
and Pareto-undominance among strategy-proof mechanisms.

Conversely, let $\cA$ be a non-matroid feasibility collection. By
\autoref{lem:matroid-known}, there is a priority ranking $\pi$ for which
$C^\pi_{\cA}$ is not substitutable. The necessity
direction of
\autoref{thm:robust-boundary}(ii) then embeds $(\cA,\pi)$ in a market with one
additional capacity-one school in which greedy COP-and-Choose is manipulable
under every proposal-order convention; \autoref{thm:sp-pareto-frontier} uses
the same market and shows that no strategy-proof mechanism weakly Pareto
dominates greedy there.
\end{proof}

\begin{proof}[Proof of \autoref{prop:ascending}]
List students in $\pi$-order. Because the requirement vectors are totally
ordered componentwise and $\pi$ refines ascending requirements, requirements
are weakly increasing componentwise along the list.

Consider greedy choice from a set $I\subseteq\I$. Suppose that, after selecting
$X\subseteq I$, greedy cannot add student $i$. Then $X\cup\{i\}$ violates
some resource constraint. Every lower-priority student $j\in I$ satisfies
$r_j\geq r_i$ componentwise, so $X\cup\{j\}$ violates the same constraint.
Thus, once greedy skips a student, it skips every lower-priority student.
Greedy therefore selects a feasible priority prefix and coincides with cutoff
choice on every set.

The common rule is substitutable by \autoref{lem:cutoff-rejection}.

It remains to establish size monotonicity. Let $I\subseteq I'\subseteq\I$ and
write $k=|C^\pi_{\cA}(I)|$. Since greedy selects a priority prefix, the first
$k$ students of $I$ form a feasible set. For each $m\leq k$, the $m$th
student of $I'$ has a requirement vector weakly smaller componentwise than
the $m$th student of $I$. Hence the sum of the first $k$ requirement vectors
in $I'$ is componentwise no larger than the corresponding sum in $I$ and is
therefore feasible. Greedy consequently selects at least $k$ students from
$I'$, so
$|C^\pi_{\cA}(I')|\geq |C^\pi_{\cA}(I)|.$
Thus, greedy choice is size monotone. The final conclusion follows from
\autoref{thm:robust-boundary}(i) and \autoref{thm:robust-boundary}(ii).
\end{proof}

The chain assumption cannot simply be dropped in several dimensions. For
example, let $\I=\{1,2,3\}$, $\cD=\{1,2\}$, $q=(2,3)$, and
\[
  r_1=(1,1),\qquad r_2=(1,2),\qquad r_3=(2,1).
\]
The only componentwise comparisons between distinct vectors are
$r_1\leq r_2$ and $r_1\leq r_3$; $r_2$ and $r_3$ are incomparable. Thus,
$1\mathrel{\pi}3\mathrel{\pi}2$ refines ascending requirements: it ranks
student 1 first, and either order of students 2 and 3 is permitted. These
requirements induce the feasibility collection in \autoref{ex:blocker};
under this priority, its third case shows that greedy choice is not substitutable.

\subsection{Refugee resettlement}

\begin{proof}[Proof of \autoref{prop:resource-hierarchy}]
By \autoref{lem:cutoff-protects-fair} and
\autoref{thm:dual-protection}(i), every individually rational and fair matching
is claim-compatible.

Suppose next that an individually rational matching $\mu^0$ is
claim-compatible. Fix a locality $\ell\in\cL$ and a family
$f\in\mu^0(\ell)$. Then
$H_\ell^{\mu^0}(f)\cup\{f\}$ is feasible. Hence, for every resource dimension
$d$ with $r_f^d>0$,
\[
  r_f^d+\sum_{f'\in H_\ell^{\mu^0}(f)}r_{f'}^d
  \leq q_\ell^d.
\]
For dimensions with $r_f^d=0$, weak accommodation imposes no restriction.
Thus, locality $\ell$ weakly accommodates $f$ alongside
$H_\ell^{\mu^0}(f)$, so $\mu^0$ is interference-free.

Now suppose that $\mu^0$ is interference-free. Fix
$\ell\in\cL$, $f\in\mu^0(\ell)$, and a feasible set
$F\subseteq H_\ell^{\mu^0}(f)$. For every $d\in\cD$ with $r_f^d>0$,
interference-freeness gives
\[
  r_f^d+\sum_{f'\in F}r_{f'}^d
  \leq
  r_f^d+\sum_{f'\in H_\ell^{\mu^0}(f)}r_{f'}^d
  \leq q_\ell^d.
\]
If $r_f^d=0$, adding $f$ does not change resource use in dimension $d$.
Therefore $F\cup\{f\}$ is feasible. Hence $\mu^0$ is weakly
claim-compatible.

It remains to show that each implication can be strict. Let
$\cL=\{\ell\}$ and suppose every family finds $\ell$ acceptable.

First, let $\cF=\{f_1,f_2\}$, with one resource, capacity $2$,
$r_{f_1}=r_{f_2}=1$, and priority $f_1\mathrel{\pi_\ell}f_2$. The matching
that assigns only $f_2$ is claim-compatible, since $\{f_1,f_2\}$ is feasible,
but it is not fair because the unmatched family $f_1$ has higher priority than $f_2$.
Thus, the first implication can be strict.

Second, let $\cF=\{f_1,f_2,f_3\}$, with capacity vector $(2,1)$,
requirements
\[
  r_{f_1}=r_{f_2}=(2,0),
  \qquad
  r_{f_3}=(0,1),
\]
and priority
$f_1\mathrel{\pi_\ell}f_2\mathrel{\pi_\ell}f_3$. The matching that assigns
only $f_3$ is interference-free because $f_3$ uses only the second resource,
but it is not claim-compatible because
$\{f_1,f_2,f_3\}$ is infeasible. Thus, the second implication can be strict.

Finally, let there be one resource with capacity $3$,
$r_{f_1}=r_{f_2}=2$, $r_{f_3}=1$, and the same priority. The matching that
assigns only $f_3$ is weakly claim-compatible: every feasible subset of
$\{f_1,f_2\}$ can be combined with $f_3$. It is not interference-free,
because accommodating $f_3$ alongside both higher-priority families would
require five units of capacity. Thus, the third implication can be strict.
\end{proof}

\begin{proof}[Proof of \autoref{prop:kda-coincidence}]
Fix an individually rational matching $\mu^0$, a locality $\ell\in\cL$, and
$f\in\mu^0(\ell)$. With one resource, the model requires $r_f>0$, so
$\ell$ weakly accommodates $f$ alongside
$H_\ell^{\mu^0}(f)$ if and only if
$H_\ell^{\mu^0}(f)\cup\{f\}$ is feasible. Hence interference-freeness and
claim compatibility coincide.

By \autoref{thm:cutoff-implementation}, the cutoff outcome is individually
rational and fair. By \autoref{prop:resource-hierarchy}, it is therefore
interference-free. Moreover, claim compatibility and
\autoref{thm:dual-protection} imply that cutoff choice protects every
individually rational, interference-free matching. Hence
\autoref{thm:protection} implies that the cutoff outcome weakly Pareto
dominates every individually rational, interference-free matching. It is
therefore family-optimal among such matchings.

\citet[Theorem~5]{DKT2023} gives the same family-optimality property to the
KDA outcome under the standing assumptions stated in
\autoref{sec:resources}. Since both outcomes are interference-free, each weakly Pareto
dominates the other. Strict family preferences then imply that they assign
every family to the same locality. Thus, cutoff COP-and-Choose and KDA
coincide.
\end{proof}

\begin{proof}[Proof of \autoref{cor:if-dominance}]
By \autoref{prop:resource-hierarchy}, every individually rational,
interference-free matching is weakly claim-compatible. Hence
\autoref{thm:dual-protection} implies that greedy choice protects every such
matching, and \autoref{thm:protection} implies that every greedy
COP-and-Choose outcome weakly Pareto dominates it. By Theorems~5 and~6 of
\citet{DKT2023}, respectively, KDA and TKDA are interference-free on the
domain stated in \autoref{sec:resources}. Their outcomes are individually
rational by the admissible-pair and outside-option conventions imposed there.
Thus, greedy COP-and-Choose weakly Pareto dominates both.

By \autoref{thm:cutoff-implementation}, the cutoff outcome is individually
rational and fair, and therefore interference-free by
\autoref{prop:resource-hierarchy}. By \citet[Theorem~5]{DKT2023}, KDA is
family-optimal among individually rational, interference-free matchings.
TKDA is such a matching by \citet[Theorem~6]{DKT2023}. Thus, KDA weakly
Pareto dominates both TKDA and cutoff COP-and-Choose.

The final statement follows from \autoref{prop:kda-coincidence}.
\end{proof}

\section*{AI usage disclosure}

The authors used ChatGPT (OpenAI) and Claude (Anthropic) to assist with
manuscript preparation, including exposition, LaTeX editing, and checking  
mathematical arguments. The authors take full responsibility for the content 
of the paper, including its results, proofs, computations, and references.

\begin{spacing}{1.15}

\end{spacing}

\clearpage
\section*{Online Appendix}
\addcontentsline{toc}{section}{Online Appendix}
\renewcommand{\thesection}{OA\arabic{section}}
\renewcommand{\thesubsection}{OA\arabic{section}.\arabic{subsection}}
\renewcommand{\thesubsubsection}{OA\arabic{section}.\arabic{subsection}.\arabic{subsubsection}}
\renewcommand*{\theHsection}{supp.\arabic{section}}
\renewcommand*{\theHsubsection}{supp.\arabic{section}.\arabic{subsection}}
\renewcommand*{\theHsubsubsection}{supp.\arabic{section}.\arabic{subsection}.\arabic{subsubsection}}
\setcounter{section}{0}
\setcounter{subsection}{0}
\setcounter{subsubsection}{0}

\section{A Cutoff-Containing Incentive Limit}

\begin{theorem}
\label{thm:no-sp-class}
There are constraints and priorities such that, for every cutoff-containing
profile of choice rules, COP-and-Choose is manipulable under every
proposal-order convention.
\end{theorem}

\begin{proof}[Proof of \autoref{thm:no-sp-class}]
Let $\I=\{1,2,3\}$ and $\Sset=\{b,c\}$, with feasibility collections
\[
\begin{split}
  \cA_b&=\{\emptyset,\{1\},\{2\},\{3\},\{2,3\}\},\\
  \cA_c&=\{\emptyset,\{1\},\{2\},\{3\},\{1,2\}\},
\end{split}
\]
and priorities
\[
  1\mathrel{\pi_b}2\mathrel{\pi_b}3,
  \qquad
  2\mathrel{\pi_c}3\mathrel{\pi_c}1.
\]
Both feasibility collections are downward closed and fail the matroid
augmentation property: $\{1\}$ and $\{2,3\}$ witness failure for $\cA_b$,
while $\{3\}$ and $\{1,2\}$ do so for $\cA_c$.

Feasibility and cutoff containment determine the choice rule for every set of
students except $\{1,2,3\}$ at school $c$. For that set, cutoff choice is
$\{2\}$, and the only feasible supersets containing it are $\{2\}$ and
$\{1,2\}$. Hence there are exactly two cutoff-containing profiles of choice
rules.

Suppose first that
\[
  C_c(\{1,2,3\})=\{2\}.
\]
Let student $1$ rank $c$ above $b$, let student $2$ find only $c$
acceptable, and let student $3$ rank $c$ above $b$. Under truthful reporting,
all three students eventually propose to $c$. Its terminal choice is $\{2\}$,
so students $1$ and $3$ eventually propose to $b$. From $\{1,3\}$, school
$b$ selects student $1$. Thus, the outcome assigns student $1$ to $b$ and
student $2$ to $c$, leaving student $3$ unmatched.

If student $3$ instead reports that only $b$ is acceptable, then only students
$1$ and $2$ propose to $c$, which selects $\{1,2\}$, while $b$ selects
student $3$. Student $3$ therefore obtains $b$ rather than remaining
unmatched, a profitable deviation.

Suppose instead that
\[
  C_c(\{1,2,3\})=\{1,2\}.
\]
Let student $1$ find only $c$ acceptable, let student $2$ rank $b$ above
$c$, and let student $3$ find only $c$ acceptable. Under truthful reporting,
students $1$ and $3$ propose to $c$, which selects student $3$, while student
$2$ is selected by $b$. Hence student $1$ is unmatched.

Now let student $1$ report $c\succ'_1 b\succ'_1\emptysetmatch$. After being
rejected by $c$, student $1$ proposes to $b$ and displaces student $2$.
Student $2$ then proposes to $c$, whose cumulative set of proposers becomes
$\{1,2,3\}$ and whose choice is $\{1,2\}$. Student $1$ is therefore selected
by both schools and chooses $c$ at the final step. She obtains her true first
choice instead of remaining unmatched.

In both cases, the relevant eventual proposals and rejections are forced by
the reported preferences, so the manipulation succeeds under every
proposal-order convention.
\end{proof}

The next example exhibits a rule to which
\autoref{prop:greedy-dominates-sp-class} applies nonvacuously, and over which
greedy's dominance is strict.

\begin{example}\label{ex:strict-bracket}
Let $\I=\{1,2,3\}$ and $\Sset=\{a,b\}$ with two resources, $r_1=(0,1)$,
$r_2=r_3=(1,0)$, $q_a=(2,0)$, and $q_b=(0,1)$, and let both priorities rank
students in the order $1,2,3$. Let $C_b$ be greedy choice and let $C_a$ select
the cutoff set together with the first student greedy would add beyond it, so
that $C_a(\{2,3\})=\{2,3\}$ while $C_a(\{1,2,3\})=\{2\}$. Then $C_a$ is
substitutable and cutoff-containing but not size monotone, while greedy choice
is substitutable and size monotone at both schools. Each student finds only one
school feasible and $C_a$ selects student~2 from every proposal set containing
her, so COP-and-Choose under $C$ is strategy-proof. When all three students
find only $a$ acceptable it returns $(\emptysetmatch,a,\emptysetmatch)$, while
greedy COP-and-Choose returns $(\emptysetmatch,a,a)$.
\end{example}

The next example shows that the substitutability hypothesis in
\autoref{prop:greedy-dominates-sp-class} cannot be dropped: a cutoff-containing
rule that fails substitutability can induce a strategy-proof COP-and-Choose
that greedy does not weakly Pareto dominate.

\begin{example}\label{ex:nonsubstitutable-sp}
Let $\I=\{1,2,3,4\}$ and $\Sset=\{s\}$, with priority $1\pi_s 2\pi_s 3\pi_s 4$
and feasibility collection consisting of the sets that contain at most one of
$\{1,2\}$ and at most one of $\{3,4\}$. This is a partition matroid, so greedy
choice is substitutable and size monotone. From the full set $\{1,2,3,4\}$,
cutoff choice selects $\{1\}$ and greedy choice selects $\{1,3\}$.

Let $C_s$ agree with greedy choice on every proper subset of $\I$ and select
$\{1,4\}$ from $\I$ itself. The set $\{1,4\}$ is feasible and contains the
cutoff set, so $C_s$ is a cutoff-containing choice rule. It is not
substitutable: student $4$ is selected from $\{1,2,3,4\}$ but not from its
subset $\{1,3,4\}$, where $C_s$ coincides with greedy and selects $\{1,3\}$.

With a single school, a student's report affects the process only through
whether she proposes to $s$, and her assignment is $s$ if and only if she is
selected from the terminal set of proposers. Declaring $s$ unacceptable leaves
her unmatched, and declaring it acceptable when it is not can only assign her
an unacceptable school. Hence COP-and-Choose under $C_s$ is strategy-proof, as
is greedy COP-and-Choose.

At the profile in which all four students find $s$ acceptable, greedy assigns
$\{1,3\}$ to $s$ and COP-and-Choose under $C_s$ assigns $\{1,4\}$. Student $4$
strictly prefers the latter, so greedy COP-and-Choose does not weakly Pareto
dominate COP-and-Choose under $C_s$. The only hypothesis of
\autoref{prop:greedy-dominates-sp-class} that fails is substitutability of
$C_s$.
\end{example}

\Needspace{7\baselineskip}
\section{Additional Welfare Results}
\subsection{Why cutoff retains rejected proposals}

\begin{example}
\label{ex:cutoff-da}
Let $\I=\{1,2,3\}$ and $\Sset=\{s\}$. Every student finds the school
acceptable. The priority is
\[
  1\mathrel{\pi_s}2\mathrel{\pi_s}3,
\]
and the feasible sets are
\[
  \cA_s=\{\emptyset,\{1\},\{2\},\{3\},\{1,3\}\}.
\]

Consider the natural deferred-acceptance analogue that applies cutoff choice to
the currently held students together with each new proposer and discards
rejected proposals. Suppose students propose in the order $1,2,3$. Student $1$
is initially held. When student $2$ proposes, cutoff choice applied to
$\{1,2\}$ selects only student $1$, so student $2$ is rejected and her proposal
is discarded. Student $3$ then proposes. Since $\{1,3\}$ is feasible, cutoff
choice selects both students. The resulting matching assigns students $1$ and
$3$ to $s$ and is unfair because student $2$ has justified envy toward student
$3$.

Cutoff COP-and-Choose produces a different outcome because it retains student
$2$'s proposal. Its terminal set of proposers is $\{1,2,3\}$. Cutoff choice
selects student $1$ and stops at student $2$, whose addition would violate
feasibility. Thus, only student $1$ is assigned to $s$, which is the fair
benchmark by \autoref{thm:cutoff-implementation}.

The difference comes from retaining student $2$ in the set of proposers:
removing her changes cutoff choice from
$C_s^c(\{1,2,3\})=\{1\}$ to $C_s^c(\{1,3\})=\{1,3\}$.
\end{example}

\subsection{Maximum-cardinality cutoff completion}

The class characterized in \autoref{thm:exact-fair-dominance} is not limited
to priority-based objectives. For each school $s\in\Sset$ and set
$I\subseteq\I$ of available students, compare subsets by their indicator
vectors in descending $\pi_s$-priority order. Let $C_s^m(I)$ be the
lexicographically highest maximum-cardinality member of
\[
  \{I'\in\cA_s:C_s^c(I)\subseteq I'\subseteq I\}.
\]
This collection is nonempty because it contains $C_s^c(I)$, so the
\emph{maximum-cardinality completion} is well defined.

\begin{corollary}\label{cor:maximum-enrollment}
For every preference profile and every proposal-order convention, every
COP-and-Choose outcome under $(C_s^m)_{s\in\Sset}$ weakly Pareto dominates
every individually rational and fair matching. Each $C_s^m$ is non-wasteful,
and $|C_s^m(I)|\geq |C_s^g(I)|$ for every $s\in\Sset$ and $I\subseteq\I$.
\end{corollary}

The maximization is performed separately for each set of available students
and therefore need not maximize enrollment in the resulting matching.
Computing $C_s^m(I)$ requires solving a maximum-cardinality feasible-subset
problem. Greedy uses a single priority-ordered pass and, whenever
it is size monotone, selects a maximum-cardinality feasible subset from every
available set. On that domain, the completion provides no increase in the
number of students selected.

Whenever $C_s^m\neq C_s^g$, \autoref{prop:cutoff-containing-uniqueness}
implies that $C_s^m$ violates consistency: removing unselected students can
change the school's choice. More generally, substitutability and size
monotonicity together imply consistency \citep{AygunSonmez2013}, so a
cutoff-containing rule with both properties exists at a given
constraint-priority pair if and only if greedy has both.

\begin{proof}[Proof of \autoref{cor:maximum-enrollment}]
For every school $s\in\Sset$ and set $I\subseteq\I$,
$C_s^c(I)\subseteq C_s^m(I)$ by construction. Hence the profile
$(C_s^m)_{s\in\Sset}$ is cutoff-containing, and
\autoref{thm:exact-fair-dominance} gives the Pareto-dominance conclusion.

To see that $C_s^m$ is non-wasteful, suppose that
$i\in I\setminus C_s^m(I)$ and $C_s^m(I)\cup\{i\}\in\cA_s$. Then
$C_s^m(I)\cup\{i\}$ is a feasible subset of $I$ containing $C_s^c(I)$ and has
strictly larger cardinality than $C_s^m(I)$, contradicting the definition of
$C_s^m(I)$.

Finally, by \eqref{eq:nesting}, $C_s^g(I)$ is a feasible subset of $I$
containing $C_s^c(I)$. It is therefore one of the sets over which
$C_s^m(I)$ maximizes cardinality, so
$|C_s^m(I)|\geq |C_s^g(I)|$.
\end{proof}

\subsection{The minimal weak-claim protector}

The class of weakly claim-compatible matchings protected by greedy also has a
pointwise smallest protecting choice rule. For each school $s\in\Sset$ and set
$I\subseteq\I$, define \emph{weak-claim choice} by
\begin{equation}\label{eq:weak-claim-choice}
  C_s^w(I)=
  \bigl\{i\in I:
  I'\cup\{i\}\in\cA_s
  \text{ for every $I'\in\cA_s$ with }
  I'\subseteq I\cap G_s(i)\bigr\}.
\end{equation}

\begin{lemma}\label{lem:weak-claim-nesting}
For every $s\in\Sset$ and $I\subseteq\I$,
\[
  C_s^c(I)\subseteq C_s^w(I)\subseteq C_s^g(I).
\]
In particular, $C_s^w$ is a choice rule for every school $s$.
\end{lemma}

\begin{theorem}\label{thm:exact-weak-claim}
For a profile $C$ of choice rules, the following are equivalent.
\begin{enumerate}[label=(\roman*),leftmargin=2.5em]
\item $C_s^w(I)\subseteq C_s(I)$ for every school $s\in\Sset$ and every
      $I\subseteq\I$;
\item for every student-preference profile, $C$ protects every individually
      rational, weakly claim-compatible matching;
\item for every student-preference profile and every proposal-order convention,
      every COP-and-Choose outcome under $C$ weakly Pareto dominates every
      individually rational, weakly claim-compatible matching.
\end{enumerate}
\end{theorem}

Thus, $C^w$ is the pointwise smallest weak-claim protector. On a matroid,
$C^w=C^g$: the greedy set on $I\cap G_s(i)$ is a basis, so if it can be
augmented by $i$, every feasible subset can. The inclusion is strict in
\autoref{ex:cutoff-da}, where $C^w(\{1,2,3\})=\{1\}$ and
$C^g(\{1,2,3\})=\{1,3\}$.

Whenever $C_s^w(I)\subsetneq C_s^g(I)$, weak-claim choice is wasteful,
since any student in the difference can be added feasibly. Together with
the maximum-cardinality completion, it brackets greedy within the
cutoff-containing class: $C_s^w$ can select fewer students and $C_s^m$
more. By \autoref{prop:cutoff-containing-uniqueness}, either rule violates
consistency whenever it differs from greedy.

\subsection{Limits and completions}\label{sec:completion}

Starting from an individually
rational matching $\mu$, an \emph{improvement step} selects a student
$i\in\I$ and a school $s\in\Sset$ such that
\[
  s\succ_i\mu(i)
  \quad\text{and}\quad
  \mu(s)\cup\{i\}\in\cA_s,
\]
moves $i$ to $s$, and removes $i$ from her former school, if any.

\begin{proposition}
\label{prop:beyond-benchmark-completion}
There is a market in which the student-optimal fair matching admits no
improvement step, while every greedy COP-and-Choose outcome Pareto dominates
it and is Pareto efficient. Hence no sequence of improvement steps starting
from the student-optimal fair matching can reach or weakly Pareto dominate a
greedy outcome.
\end{proposition}

\begin{proposition}\label{prop:completion}
Every maximal sequence of improvement steps from an individually rational
matching terminates at an individually rational, non-wasteful matching that
weakly Pareto dominates its starting point. Every cutoff-containing
COP-and-Choose outcome admits both a non-wasteful completion and a
Pareto-efficient weak Pareto improvement; each weakly Pareto dominates every
individually rational and fair matching.
\end{proposition}

\begin{proof}[Proof of \autoref{prop:completion}]
Each step preserves feasibility and individual rationality, strictly improves
one student, and changes no other assignment. Finiteness implies termination,
and the absence of a further step is non-wastefulness. The welfare guarantee
follows from \autoref{thm:exact-fair-dominance}. For Pareto efficiency, choose
a Pareto-maximal element of the finite set of feasible matchings that weakly
Pareto dominate the completed matching.
\end{proof}

The local procedure need not reach a Pareto-efficient matching because a
further gain may require coordinated reassignment; see
\citet{ImamuraKawase2024}. By
\autoref{prop:beyond-benchmark-completion}, greedy also cannot generally be
reduced to cutoff followed by unilateral improvements.

\begin{proof}[Proof of \autoref{lem:weak-claim-nesting}]
Fix $s\in\Sset$ and $I\subseteq\I$.

For the first inclusion, let $i\in C_s^c(I)$. Then the priority prefix of $I$
through $i$ is feasible. Hence, for every feasible
$I'\subseteq I\cap G_s(i)$, the set $I'\cup\{i\}$ is a subset of that prefix
and is therefore feasible by downward closure. Thus, $i\in C_s^w(I)$.

For the second inclusion, let $i\in C_s^w(I)$ and set
$X=C_s^g(I\cap G_s(i))$. The set $X$ is feasible and satisfies
$X\subseteq I\cap G_s(i)$. By the definition of weak-claim choice,
$X\cup\{i\}$ is feasible. \autoref{lem:greedy-prefix} therefore gives
$i\in C_s^g(I)$.

Thus,
\[
  C_s^c(I)\subseteq C_s^w(I)\subseteq C_s^g(I).
\]
Since $C_s^g(I)$ is feasible and $\cA_s$ is downward closed,
$C_s^w(I)$ is feasible. Hence $C_s^w$ is a choice rule.
\end{proof}

\begin{proof}[Proof of \autoref{thm:exact-weak-claim}]
Suppose first that (i) holds. Fix an individually rational, weakly
claim-compatible matching $\mu^0$, a school $s\in\Sset$, a set
$I\subseteq D_s(\mu^0)$, and $i\in\mu^0(s)\cap I$. Every feasible
subset of $I\cap G_s(i)$ is a feasible subset of
$H_s^{\mu^0}(i)$. Weak claim compatibility therefore gives
$i\in C_s^w(I)\subseteq C_s(I)$. Hence $C$ protects $\mu^0$, proving (ii).
Statement (ii) and \autoref{thm:protection} imply (iii).

It remains to prove that (iii) implies (i). Suppose (i) fails, so
$i\in C_s^w(I)\setminus C_s(I)$ for some school $s\in\Sset$, set
$I\subseteq\I$, and student $i\in I$. Apply the construction in the proof of
\autoref{thm:exact-fair-dominance} with $P=C_s^w(I)$ in place of $C_s^c(I)$:
every student in $I$ finds only $s$ acceptable, every school is unacceptable
outside $I$, and $\mu^0$ assigns $P$ to $s$. Here $P$ is feasible by
\autoref{lem:weak-claim-nesting}, so $\mu^0$ is an individually rational
matching, and it is weakly claim-compatible: for $j\in P$ the higher-priority
students who weakly prefer $s$ to their assignments are exactly
$I\cap G_s(j)$, and $j\in C_s^w(I)$ with
\eqref{eq:weak-claim-choice} gives $I'\cup\{j\}\in\cA_s$ for every feasible
$I'\subseteq I\cap G_s(j)$. As in that proof, the terminal set of proposers at
$s$ is $I$, so $i$ remains unmatched while $\mu^0(i)=s$, contradicting (iii).
\end{proof}

\begin{example}
\label{ex:iterated-choose}
Let $\I=\{1,2,3\}$ and $\Sset=\{a,b\}$. School $a$ has capacity one, while
\[
  \cA_b=\{\emptyset,\{1\},\{2\},\{3\},\{1,2\}\}.
\]
Priorities and preferences are
\[
  1\mathrel{\pi_a}2\mathrel{\pi_a}3,
  \qquad
  2\mathrel{\pi_b}3\mathrel{\pi_b}1,
\]
\[
  b\succ_1 a\succ_1\emptysetmatch,
  \qquad
  a\succ_i b\succ_i\emptysetmatch
  \quad\text{for }i\in\{2,3\}.
\]

Students $2$ and $3$ first propose to $a$, and student $1$ to $b$. Once
$2$ and $3$ have proposed to $a$, student $3$ is rejected and proposes to
$b$. Greedy choice from $\{1,3\}$ at $b$ selects $3$, causing student $1$
to propose to $a$ and displace $2$, who then proposes to $b$. Thus every
student proposes to both schools under every proposal-order convention.
The terminal offers are $O_a=\{1\}$ and $O_b=\{1,2\}$. Student $1$
chooses $b$, yielding
\[
  \mu=(b,b,\emptysetmatch).
\]
The outcome is wasteful: school $a$ is empty and can feasibly admit student
$2$, who prefers $a$ to $b$.

Now suppose each school permanently deletes declined terminal offers and
re-chooses from its remaining proposers. After student $1$ declines $a$,
school $a$ selects $2$ from $\{2,3\}$. Student $2$ chooses $a$, so $b$
deletes $2$ and selects $3$ from $\{1,3\}$, revoking student $1$'s offer.
The resulting matching is
$(\emptysetmatch,a,b)$, which makes student $1$ worse off than at the
individually rational and fair matching $(a,b,\emptysetmatch)$.
Thus this deletion-and-rechoice completion loses reference protection.
Allowing deleted proposals to return would require a separate termination
and welfare argument. The improvement steps in \autoref{prop:completion}
preserve every student's assignment unless she strictly improves.
\end{example}

\begin{proof}[Proof of \autoref{prop:beyond-benchmark-completion}]
Let $\I=\{1,2,3\}$ and $\Sset=\{a,b\}$. Student sizes are
\[
  (z_1,z_2,z_3)=(1,3,1),
\]
and school capacities are $q_a=4$ and $q_b=1$. Preferences are
\[
  a\succ_1 b\succ_1\emptysetmatch,
  \qquad
  b\succ_i a\succ_i\emptysetmatch
  \quad\text{for }i\in\{2,3\},
\]
and priorities are
\[
  3\mathrel{\pi_a}2\mathrel{\pi_a}1,
  \qquad
  1\mathrel{\pi_b}2\mathrel{\pi_b}3.
\]

Under cutoff choice, students $1,2,3$ initially propose to $a,b,b$,
respectively. Student $2$ cannot fit at $b$, and because she has higher
priority than student $3$, cutoff choice stops before reaching student $3$.
Thus, students $2$ and $3$ eventually propose to $a$. Once all three students
have proposed to $a$, cutoff choice selects $\{2,3\}$ and rejects student $1$,
who then proposes to $b$. The terminal sets of proposers are therefore
$\{1,2,3\}$ at both schools. Cutoff choice selects $\{2,3\}$ at $a$ and
$\{1\}$ at $b$, so \autoref{thm:cutoff-implementation} gives the
student-optimal fair matching
\[
  \mu^0=(b,a,a).
\]

Under greedy choice, student $2$ is skipped at $b$ because she cannot fit even
alone, and therefore eventually proposes to $a$. At $b$, greedy continues
after skipping student $2$ and selects student $3$. Students $1$ and $2$ fit
together at $a$. Hence every proposal-order convention yields
\[
  \mu=(a,a,b).
\]
Students $1$ and $3$ receive their first choices at $\mu$ rather than their
second choices at $\mu^0$, while student $2$ remains at $a$. Thus, $\mu$
strictly Pareto dominates $\mu^0$.

No improvement step is available at $\mu^0$: school $a$ uses all four
units of capacity and school $b$ uses its one unit. Hence neither school
can admit an additional student while keeping its current assignees,
so every sequence of improvement steps starting from $\mu^0$ is trivial.

Finally, $\mu$ is Pareto efficient. Students $1$ and $3$ receive their first
choices, and student $2$ cannot be assigned to her first choice $b$ because
her size exceeds its capacity. Therefore no student can be made strictly
better off without making another student worse off.
\end{proof}

\begin{example}
\label{ex:universal-improvement}
Let $\I=\{1,2,3,4\}$ and $\Sset=\{a,b\}$, with
\[
  \cA_a=\{\emptyset,\{1\},\{2\},\{3\},\{4\},\{1,2\}\},\qquad
  \cA_b=\{\emptyset,\{1\},\{2\},\{3\},\{4\},\{3,4\}\}.
\]
Every singleton is feasible. Equivalently, use two resources with
$r_1=r_2=(1,2)$, $r_3=r_4=(2,1)$, $q_a=(2,4)$, and $q_b=(4,2)$.
Priorities and preferences are
\[
  3\mathrel{\pi_a}1\mathrel{\pi_a}4\mathrel{\pi_a}2,
  \qquad 2\mathrel{\pi_b}3\mathrel{\pi_b}4\mathrel{\pi_b}1,
\]
\[
  a\succ_i b\succ_i\emptysetmatch\quad(i=1,2,4),
  \qquad b\succ_3 a\succ_3\emptysetmatch.
\]
Cutoff COP-and-Choose returns the fair benchmark
$\mu^c=(\emptysetmatch,b,a,\emptysetmatch)$. Its assigned sets $\{3\}$ at
$a$ and $\{2\}$ at $b$ are inclusion-maximal feasible, so no improvement
step of \autoref{prop:completion} is available.

Under greedy, selecting the lowest-indexed active student gives the proposal
sequence $1a,2a,3b,4a,4b$, where $is$ denotes a proposal by $i$ to $s$.
School $a$ selects $\{1,2\}$ from $\{1,2,4\}$, skipping student $4$,
and $b$ selects $\{3,4\}$. The outcome is
$\mu^g=(a,a,b,b)$, which every student strictly prefers to $\mu^c$.
It is Pareto efficient: students $1,2,3$ receive their first choices, and
student $4$ cannot join $\{1,2\}$ at $a$.

Instead, selecting the first active student in the fixed ordering $2,4,3,1$
gives the sequence
\[
  2a,\ 4a,\ 2b,\ 3b,\ 3a,\ 4b,\ 1a,\ 1b,
\]
and greedy returns $\mu^c$. Thus one order improves every student, while
another leaves the fair benchmark unchanged. Since $\mu^c$ admits no
improvement step, reaching $\mu^g$ requires coordinated reassignment,
as in \autoref{prop:beyond-benchmark-completion}.
\end{example}

\begin{remark}\label{rem:order-robust-improvement}
The universal improvement above is order dependent, but it need not be. Let
$\I=\{1,2,3\}$ and $\Sset=\{a,b\}$ with two resources,
$r_1=(2,0)$, $r_2=(1,1)$, $r_3=(0,2)$, $q_a=(1,2)$, and $q_b=(2,2)$, so that
$\cA_a=\{\emptyset,\{2\},\{3\}\}$ and
$\cA_b=\{\emptyset,\{1\},\{2\},\{3\},\{1,3\}\}$, with priorities
$3\mathrel{\pi_a}1\mathrel{\pi_a}2$ and
$2\mathrel{\pi_b}1\mathrel{\pi_b}3$ and preferences
$a\succ_1 b\succ_1\emptysetmatch$, $a\succ_2 b\succ_2\emptysetmatch$, and
$b\succ_3 a\succ_3\emptysetmatch$. Cutoff COP-and-Choose returns the
non-wasteful fair benchmark $\mu^c=(\emptysetmatch,b,a)$, while every
proposal-order convention makes greedy COP-and-Choose return the
Pareto-efficient matching $\mu^g=(b,a,b)$, which every student strictly
prefers: student $1$ is infeasible at $a$ even alone, so cutoff stalls at her
under every order while greedy continues to student $2$. Here $\cA_b$ is not a
matroid, since $\{2\}$ and $\{1,3\}$ violate exchange, and size monotonicity
fails at $b$, since $C_b^g(\{1,3\})=\{1,3\}$ while
$C_b^g(\{1,2,3\})=\{2\}$; the market therefore lies outside the domain of
\autoref{thm:robust-boundary}(ii).
\end{remark}

\subsection{Stability under greedy choice}\label{app:greedy-stability}

\begin{definition}\label{def:greedy-stability}
A matching $\mu$ is \emph{greedy-stable} if it is individually rational and
there is no school $s$ and nonempty set $B\subseteq\I\setminus\mu(s)$ such
that
\[
  s\succ_i\mu(i)\quad\text{for every }i\in B,
  \qquad B\subseteq C_s^g\bigl(\mu(s)\cup B\bigr).
\]
\end{definition}

School individual rationality, $C_s^g(\mu(s))=\mu(s)$, is automatic because
every assigned set is feasible. This stability notion uses the greedy choice
rules to determine which proposed deviations a school accepts.

Following \citet[Remark~1]{KamadaKojima2024}, a matching is \emph{weakly fair}
if no student $i$ has justified envy toward an assignee $i'$ at school $s$
with $(\mu(s)\setminus\{i'\})\cup\{i\}\in\cA_s$. Thus, only envy
that can be satisfied by a feasible replacement counts.

\Needspace{10\baselineskip}
\begin{proposition}\label{prop:greedy-stability}
\begin{enumerate}[label=(\roman*),leftmargin=2.5em]
\item An individually rational matching $\mu$ is greedy-stable if and only if
\begin{equation}\label{eq:greedy-stability-prefix}
  \bigl(\mu(s)\cap G_s(i)\bigr)\cup\{i\}\notin\cA_s
  \quad\text{whenever }s\succ_i\mu(i).
\end{equation}
Every greedy-stable matching is non-wasteful and weakly fair.
\item Every greedy COP-and-Choose outcome with disjoint terminal offers is
greedy-stable. In particular, this holds if greedy choice is substitutable
at every school, and hence if every constraint-priority pair is robustly
strategy-proof.
\end{enumerate}
\end{proposition}

\begin{proof}[Proof of \autoref{prop:greedy-stability}]
\emph{(i).} In the greedy scan of $\mu(s)\cup\{i\}$, every student
ahead of $i$ is selected, because these students form a subset of the feasible
set $\mu(s)$. Thus $i$ is selected exactly when the set in
\eqref{eq:greedy-stability-prefix} is feasible. Greedy stability therefore
implies that condition. Conversely, if a nonempty set $B$ blocks at $s$, let
$i$ be its highest-priority member. In the scan of $\mu(s)\cup B$, the
students ahead of $i$ are exactly $\mu(s)\cap G_s(i)$ and all are selected.
Since $i$ is also selected, the set in \eqref{eq:greedy-stability-prefix} is
feasible, a contradiction.

If $\mu(s)\cup\{i\}$ is feasible, so is its subset
$(\mu(s)\cap G_s(i))\cup\{i\}$. The same holds if
$(\mu(s)\setminus\{i'\})\cup\{i\}$ is feasible and $i\mathrel{\pi_s}i'$.
Condition \eqref{eq:greedy-stability-prefix} therefore implies both
non-wastefulness and weak fairness.

\emph{(ii).} Fix a preference profile and a proposal-order convention, and let
$T_s$ be school $s$'s terminal proposal set. Disjoint offers give
$\mu(s)=C_s^g(T_s)$. The matching is individually rational because students
propose only to acceptable schools.

Every student who strictly prefers $s$ to her assignment has proposed to $s$:
proposals follow preference order, and an unmatched student exhausts her
acceptable schools. For any set $B$ of such students,
\[
  C_s^g(T_s)=\mu(s)\subseteq\mu(s)\cup B\subseteq T_s.
\]
Greedy consistency (\autoref{lem:greedy-consistency}) therefore gives
$C_s^g(\mu(s)\cup B)=\mu(s)$, ruling out every nonempty blocking set
$B\subseteq\I\setminus\mu(s)$. Substitutability at every school makes
terminal offers disjoint by \autoref{prop:substitutability-cop}, and
\autoref{thm:robust-boundary}(ii) implies substitutability when every pair
is robustly strategy-proof.
\end{proof}

Condition \eqref{eq:greedy-stability-prefix} says that every student who
prefers a school is infeasible alongside its higher-priority assignees.
Size monotonicity is unnecessary for this conclusion. The next example
separates greedy stability from fairness.

\begin{example}\label{ex:kk-stability}
Consider the market in \citet[Example~3]{KamadaKojima2024} (hereafter KK).
There are four students and two schools, with preferences and priorities
\[
  s_2\succ_1 s_1\succ_1\emptysetmatch,
  \qquad s_1\succ_i s_2\succ_i\emptysetmatch\quad(i=2,3,4),
\]
\[
  1\mathrel{\pi_{s_1}}2\mathrel{\pi_{s_1}}3\mathrel{\pi_{s_1}}4,
  \qquad
  2\mathrel{\pi_{s_2}}1\mathrel{\pi_{s_2}}3\mathrel{\pi_{s_2}}4,
\]
and feasibility collections
\[
\begin{split}
  \cA_{s_1}&=\{\emptyset,\{1\},\{2\},\{3\},\{4\},\{2,4\}\},\\
  \cA_{s_2}&=\{\emptyset,\{1\},\{2\},\{3\},\{4\}\}.
\end{split}
\]
In KK's terminology, a stable matching is feasible, individually rational,
fair, and non-wasteful. The relevant matchings are
\begin{center}
\begin{tabular}{lccc}
\toprule
Matching & $s_1$ & $s_2$ & Unmatched\\
\midrule
$\mu$: KK student-optimal stable & $\{1\}$ & $\{2\}$ & $\{3,4\}$\\
$\mu'$: cutoff COP-and-Choose & $\{2\}$ & $\{1\}$ & $\{3,4\}$\\
$\mu''$: greedy COP-and-Choose & $\{2,4\}$ & $\{1\}$ & $\{3\}$\\
\bottomrule
\end{tabular}
\end{center}

Both COP outcomes are independent of proposal order in this market.
Students 1 and 2 retain their first choices: neither can be displaced before
the other leaves her first choice, so neither is the first to do so.
Students 3 and 4 eventually propose to $s_1$, which chooses from
$\{2,3,4\}$. Cutoff stops at student 3 and selects only student 2; greedy
skips student 3 and also selects student 4. School $s_2$ retains student 1.
Thus, $\mu'$ Pareto dominates $\mu$ by improving students 1 and 2, and
$\mu''$ Pareto dominates $\mu'$ by improving student 4.
Student optimality of $\mu$ is restricted to KK-stable matchings.
The cutoff outcome $\mu'$ is wasteful because student $4$ can join
student $2$ at $s_1$.

At $\mu''$, only student 3 wants to move. Greedy chooses $\{2,4\}$ from
$\{2,3,4\}$ at $s_1$ and $\{1\}$ from $\{1,3\}$ at $s_2$, so there is
no blocking set. Nevertheless, student 3 has justified envy toward student 4
at $s_1$ under KK's fairness criterion, which does not require the replacement
to be feasible. Hence $\mu''$ is greedy-stable but is not fair.

This market is outside the robust strategy-proofness domain:
\[
  C_{s_1}^g(\{3,4\})=\{3\},
  \qquad C_{s_1}^g(\{2,3,4\})=\{2,4\},
\]
which violates substitutability. Terminal offers are nevertheless disjoint,
and \autoref{prop:greedy-stability}(ii) therefore applies.
\end{example}

For a greedy COP-and-Choose outcome, a weak-fairness violation requires a
higher-priority terminal offer holder to decline the school. To see this,
suppose $i$ has feasible justified envy toward $i'$ at $s$, and let
$O_s=C_s^g(T_s)$. Student $i$ proposed to $s$ but has no terminal offer there.
Put $H=O_s\cap G_s(i)$. Since greedy skipped $i$, $H\cup\{i\}$ is
infeasible. If $H\subseteq\mu(s)$, then
$H\cup\{i\}\subseteq(\mu(s)\setminus\{i'\})\cup\{i\}$, contradicting
feasibility of the replacement. Some $h\in H\setminus\mu(s)$ therefore
declined $s$ in favor of another terminal offer. Overlap is necessary for
such a violation, but need not produce one.

\begin{example}\label{ex:weak-fairness}
Let $\I=\{1,2,3\}$ and $\Sset=\{a,b\}$, with
\[
\begin{split}
  \cA_a&=\{\emptyset,\{1\},\{2\},\{3\},\{1,3\}\},\\
  \cA_b&=\{\emptyset,\{1\},\{2\},\{3\},\{1,2\}\}.
\end{split}
\]
These constraints can be represented by school-independent resource
requirements $r_1=(1,1)$, $r_2=(2,1)$, $r_3=(1,2)$ and capacities
$q_a=(2,3)$, $q_b=(3,2)$. Priorities and preferences are
\[
  1\mathrel{\pi_a}2\mathrel{\pi_a}3,
  \qquad 2\mathrel{\pi_b}3\mathrel{\pi_b}1,
\]
\[
  b\succ_1 a\succ_1\emptysetmatch,
  \qquad a\succ_i b\succ_i\emptysetmatch\quad(i=2,3).
\]
Under every proposal order, students 2 and 3 propose first to $a$, which
chooses $\{2\}$ from $\{2,3\}$. Student 3 then proposes to $b$, where
she displaces student 1 or prevents her selection when she arrives. Student 1
therefore proposes to $a$, whose choice becomes $\{1,3\}$, displacing
student 2. Student 2 proposes to $b$, whose choice becomes $\{1,2\}$.
Both terminal proposal sets are $\{1,2,3\}$, and student 1 chooses $b$.
Thus every greedy outcome is $\mu^g=(b,b,a)$, whereas cutoff yields
$\mu^c=(a,b,\emptysetmatch)$. Greedy Pareto dominates the fair benchmark.

Student 2 prefers $a$ and outranks its assignee 3, whom she can feasibly
replace. Hence $\mu^g$ is not weakly fair. The violation arises when
student 1 declines her offer from $a$. Yet $\mu^g$ is non-wasteful:
only student 2 wants to move, and $\{2,3\}\notin\cA_a$. It is also
Pareto efficient, since students 1 and 3 receive their first choices and
student 2 cannot join student 3 at $a$. Neither unilateral non-wasteful
completion nor any Pareto improvement can repair this violation.
\end{example}

\Needspace{6\baselineskip}
\subsection{Common priorities and efficiency}\label{app:common-priority-efficiency}

\begin{proposition}\label{prop:common-priority-efficiency}
Suppose every school uses the same strict priority $\pi$ and each greedy
choice rule $C_s^g=C^\pi_{\cA_s}$ is substitutable. At every preference
profile and under every proposal-order convention, greedy COP-and-Choose
coincides with serial dictatorship in order $\pi$ and is Pareto efficient.
\end{proposition}

\begin{proof}[Proof of \autoref{prop:common-priority-efficiency}]
By \autoref{prop:substitutability-cop}, the outcome is proposal-order
independent. Use the convention that always selects the highest-priority
active student, and write the common order as
$i_1\mathrel{\pi}\cdots\mathrel{\pi}i_n$.

Suppose $i_1,\ldots,i_{k-1}$ have completed their proposals. Their selections
remain unchanged because adding lower-priority proposers cannot change
greedy's decisions about higher-priority students. Let $H_s$ be the set of
these earlier students assigned to $s$. Earlier proposals rejected by $s$
are still skipped. Thus, when $i_k$ proposes to $s$, she is selected exactly
when $H_s\cup\{i_k\}\in\cA_s$. Until selected, she remains the
highest-priority active student and proposes in preference order. She obtains
her favorite acceptable school that can accommodate her with its earlier
assignees, or remains unmatched if none exists. Later proposals do not alter
her selection. This is serial dictatorship.

Let $\mu$ be its outcome and suppose a feasible matching $\nu$ Pareto
dominates $\mu$. Let $i_k$ be the highest-priority student whose assignment
differs. All earlier students keep their assignments, while strict
preferences imply $\nu(i_k)\succ_{i_k}\mu(i_k)$. Since $\mu$ is
individually rational, $s=\nu(i_k)$ is an acceptable school. Its earlier
assignees satisfy $H_s\cup\{i_k\}\subseteq\nu(s)$. Feasibility of $\nu$
and downward closure imply $H_s\cup\{i_k\}\in\cA_s$, so serial
dictatorship could not assign $i_k$ to the worse alternative $\mu(i_k)$.
\end{proof}

\begin{proof}[Proof of \autoref{cor:ascending-efficiency}]
\autoref{prop:ascending} gives substitutability at every school, so
\autoref{prop:common-priority-efficiency} applies to greedy COP-and-Choose.
The proof of \autoref{prop:ascending} also shows that greedy and cutoff
choice coincide on every available set under its assumptions. Hence their
COP-and-Choose outcomes coincide under every proposal-order convention.
\end{proof}

\section{Additional Incentive Results}\label{app:incentives-additional}

\subsection{Deletion and contraction at a fixed priority}
\label{app:fixed-priority-proof}

Fix a nonempty $I\subseteq\I$, a feasibility collection $\cA$ on $I$, and a
strict priority $\pi$ on $I$. Let $i$ be the highest-priority student and
$\pi_{-i}:=\pi|_{I\setminus\{i\}}$.

\begin{definition}[Deletion and contraction]\label{def:minors}
\[
  \cA\setminus i:=\{I'\in\cA:i\notin I'\},
  \qquad
  \cA/i:=\{I'\subseteq I\setminus\{i\}:I'\cup\{i\}\in\cA\}.
\]
Write $C_{\setminus i}:=C^{\pi_{-i}}_{\cA\setminus i}$ and, when
$\{i\}\in\cA$, $C_{/i}:=C^{\pi_{-i}}_{\cA/i}$.
\end{definition}

\begin{proposition}\label{prop:fixed-priority}
If $\{i\}\notin\cA$, then $C^\pi_{\cA}$ is substitutable if and only if
$C_{\setminus i}$ is substitutable, and $C^\pi_{\cA}$ is size monotone if
and only if $C_{\setminus i}$ is size monotone.

If $\{i\}\in\cA$, then:
\begin{enumerate}[label=(\roman*),leftmargin=2.5em]
\item $C^\pi_{\cA}$ is substitutable if and only if $C_{\setminus i}$ and
      $C_{/i}$ are substitutable and
      $C_{/i}(I')\subseteq C_{\setminus i}(I')$ for every
      $I'\subseteq I\setminus\{i\}$;
\item $C^\pi_{\cA}$ is size monotone if and only if
      $C_{\setminus i}$ and $C_{/i}$ are size monotone and
      $|C_{\setminus i}(I')|\leq 1+|C_{/i}(I')|$ for every
      $I'\subseteq I\setminus\{i\}$;
\item $C^\pi_{\cA}$ is substitutable and size monotone if and only if
      $C_{\setminus i}$ and $C_{/i}$ satisfy both properties,
      $C_{/i}(I')\subseteq C_{\setminus i}(I')$, and
      $|C_{\setminus i}(I')\setminus C_{/i}(I')|\leq 1$ for every
      $I'\subseteq I\setminus\{i\}$.
\end{enumerate}
Recursion over deletion and contraction minors is finite because each step
reduces the ground set by one student.
\end{proposition}

\begin{proof}[Proof of \autoref{prop:fixed-priority}]
Every rule in the proposition is greedy and therefore consistent by
\autoref{lem:greedy-consistency}. By \autoref{lem:aizerman}, substitutability
is equivalent here to path independence,
$C(J\cup J')=C(C(J)\cup J')$. The greedy algorithm gives
\[
  C^\pi_{\cA}(J)=
  \begin{cases}
    C_{\setminus i}(J),&i\notin J,\\
    \{i\}\cup C_{/i}(J\setminus\{i\}),&i\in J,
  \end{cases}
  \qquad\text{if }\{i\}\in\cA,
\]
and $C^\pi_{\cA}(J)=C_{\setminus i}(J\setminus\{i\})$ otherwise.

Suppose first that $\{i\}\notin\cA$. Then
\[
  C^\pi_{\cA}(J)=C_{\setminus i}(J\setminus\{i\})
  \qquad\text{for every }J\subseteq I.
\]
Thus, adding or removing the loop $i$ affects neither path independence nor
size monotonicity. Assume henceforth that $\{i\}\in\cA$ and write
$C=C^\pi_{\cA}$.

\emph{(i).} Consider path independence,
$C(J\cup J')=C(C(J)\cup J')$. There are three cases.

If $i\in J$, then $i$ is selected on both sides. After removing $i$ and using
the displayed decomposition, the identity is exactly path independence of
$C_{/i}$.
If $i\notin J\cup J'$, the identity is exactly path independence of
$C_{\setminus i}$.

It remains to consider the case $i\notin J$ and $i\in J'$. Write
$X=J$ and $Y=J'\setminus\{i\}$. By the displayed decomposition, path independence
in this case is equivalent to
\begin{equation}\label{eq:cross}
  C_{/i}(X\cup Y)
  =
  C_{/i}\bigl(C_{\setminus i}(X)\cup Y\bigr)
  \qquad
  \text{for all }X,Y\subseteq I\setminus\{i\}.
\end{equation}

Suppose first that \eqref{eq:cross} holds. Setting $Y=\emptyset$ gives
\[
  C_{/i}(X)
  =
  C_{/i}\bigl(C_{\setminus i}(X)\bigr)
  \subseteq C_{\setminus i}(X).
\]
Thus, the containment condition is necessary.

Conversely, suppose $C_{/i}$ is path independent and
$C_{/i}(X)\subseteq C_{\setminus i}(X)$ for every
$X\subseteq I\setminus\{i\}$. Greedy choice is consistent by
\autoref{lem:greedy-consistency}, so
\[
  C_{/i}\bigl(C_{\setminus i}(X)\bigr)=C_{/i}(X).
\]
Hence, for every $X,Y\subseteq I\setminus\{i\}$,
\[
  C_{/i}(X\cup Y)
  =C_{/i}\bigl(C_{/i}(X)\cup Y\bigr)
  =C_{/i}\bigl(C_{/i}(C_{\setminus i}(X))\cup Y\bigr)
  =C_{/i}\bigl(C_{\setminus i}(X)\cup Y\bigr),
\]
where the first and last equalities use path independence of $C_{/i}$.
Thus, \eqref{eq:cross} holds, proving (i).

\emph{(ii).} Suppose first that $C$ is size monotone. Restricting attention to
sets that exclude $i$ gives size monotonicity of $C_{\setminus i}$.
Restricting attention to sets that contain $i$ and using the displayed
decomposition gives size monotonicity of $C_{/i}$. Finally, for every
$X\subseteq I\setminus\{i\}$, size monotonicity applied to
$X\subseteq X\cup\{i\}$ gives
\[
  |C_{\setminus i}(X)|
  \leq
  1+|C_{/i}(X)|.
\]

Conversely, suppose $C_{\setminus i}$ and $C_{/i}$ are size monotone and the
displayed inequality holds for every $X\subseteq I\setminus\{i\}$. Fix
$J\subseteq J'\subseteq I$. If $i\notin J'$, size monotonicity follows from
that of $C_{\setminus i}$. If $i\in J$, it follows from that of $C_{/i}$.
In the remaining case, $i\notin J$ and $i\in J'$. Let
$K=J'\setminus\{i\}$. Then
\[
  |C(J)|
  =|C_{\setminus i}(J)|
  \leq |C_{\setminus i}(K)|
  \leq 1+|C_{/i}(K)|
  =|C(J')|.
\]
Thus, $C$ is size monotone, proving (ii).

\emph{(iii).} Combine (i) and (ii). Under
$C_{/i}(X)\subseteq C_{\setminus i}(X)$,
\[
  |C_{\setminus i}(X)|
  \leq 1+|C_{/i}(X)|
\]
is equivalent to
\[
  |C_{\setminus i}(X)\setminus C_{/i}(X)|\leq 1.
\]
This proves (iii).

Applying the same decomposition recursively to each deletion and contraction
minor, always at its highest-priority student, yields the finite recursive
characterization. Each step reduces the ground set by one student.
\end{proof}

\subsection{A cutoff manipulation on a matroid}

\begin{example}
\label{ex:cutoff-manipulation}
Let $\I=\{1,2,3\}$ and $\Sset=\{a,b\}$. School $a$ has capacity one, while
\[
  \cA_b=\{\emptyset,\{2\},\{3\}\}.
\]
Thus, $\cA_b$ is a rank-one matroid, so every nonempty feasible set is a
singleton, and student $1$ is a loop because $\{1\}\notin\cA_b$. Priorities
are
\[
  2\mathrel{\pi_a}1\mathrel{\pi_a}3,
  \qquad
  1\mathrel{\pi_b}2\mathrel{\pi_b}3,
\]
and preferences are
\[
  b\succ_1 a\succ_1\emptysetmatch,
  \qquad
  b\succ_2 a\succ_2\emptysetmatch,
  \qquad
  a\succ_3\emptysetmatch.
\]

Under truthful reporting, students $1$ and $2$ eventually propose to $b$.
Because student $1$ has highest priority at $b$ but is infeasible there,
cutoff choice stops at student $1$ and selects neither student. Students $1$
and $2$ therefore reach $a$, where student $2$ has higher priority. The
outcome assigns student $2$ to $a$ and leaves students $1$ and $3$ unmatched.

If student $1$ instead reports that only $a$ is acceptable, student $2$ is
selected by $b$, while $a$ selects student $1$ over student $3$. Student $1$
therefore obtains $a$ rather than remaining unmatched, so cutoff
COP-and-Choose is manipulable.

Under truthful reporting, greedy COP-and-Choose produces the latter matching:
at $b$, greedy skips the infeasible student $1$ and continues to student $2$.
\end{example}

\clearpage
\section{Additional HIAS Analyses}\label{app:hias-additional}

\autoref{fig:hias-dominance} presents the beneficiary shares in
\autoref{tab:hias-by-type} with Monte Carlo confidence intervals.
\autoref{fig:hias-outcomes} summarizes the calibrated outcomes behind
\autoref{tab:hias-main}, comparing cutoff COP-and-Choose, KDA, and greedy
COP-and-Choose under one and three resources. The subsections that follow
report the supporting incidence, severity, and robustness checks.

\textbf{Data and code availability.}
The calibration uses the public data and code archive of \citet{DKT2023},
available from the AEA Data and Code Repository:
\begin{center}
  \url{https://doi.org/10.3886/E191062V1}
\end{center}
The code and
derived results needed to reproduce the tables and figures in
\autoref{sec:hias-simulations}, including the strong-envy and completion
incentive audits, accompany this submission as supplementary material.

\begin{figure}[H]
  \centering
  \tagstructbegin{tag=Figure,alttext={Bar charts report the share of families
  strictly better off under greedy COP-and-Choose relative to KDA.
  Separate panels show one-dimensional and three-dimensional constraints under
  four simulated preference specifications.}}
  \tagmcbegin{tag=Figure}
  \includegraphics[width=\textwidth,height=.68\textheight,keepaspectratio]{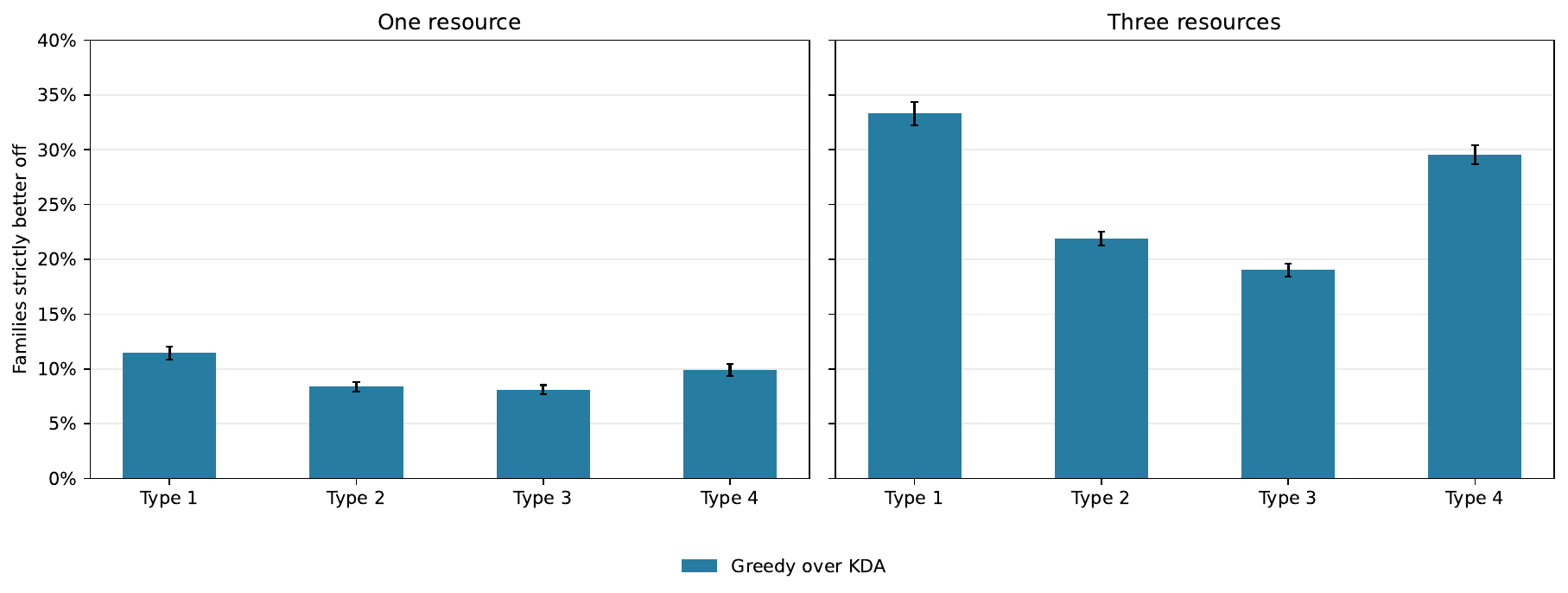}
  \tagmcend
  \tagstructend
  \caption{Pareto gains of greedy COP-and-Choose relative to KDA. Each bar is
  the share of families receiving a strictly preferred locality under greedy
  according to their simulated preferences; the reverse share is zero in every
  comparison shown. Left: one resource; right: three resources.
  Error bars are 95-percent Monte Carlo confidence intervals
  across the 100 simulated preference draws.}
  \label{fig:hias-dominance}
\end{figure}

\begin{figure}[H]
  \centering
  \tagstructbegin{tag=Figure,alttext={Bar charts compare cutoff COP-and-Choose
  with KDA and with greedy COP-and-Choose across four preference
  specifications under one-dimensional and three-dimensional HIAS resource
  constraints. Panels report placement and predicted-employment outcomes with
  confidence intervals across simulation draws.}}
  \tagmcbegin{tag=Figure}
  \includegraphics[width=\textwidth,height=.68\textheight,keepaspectratio]{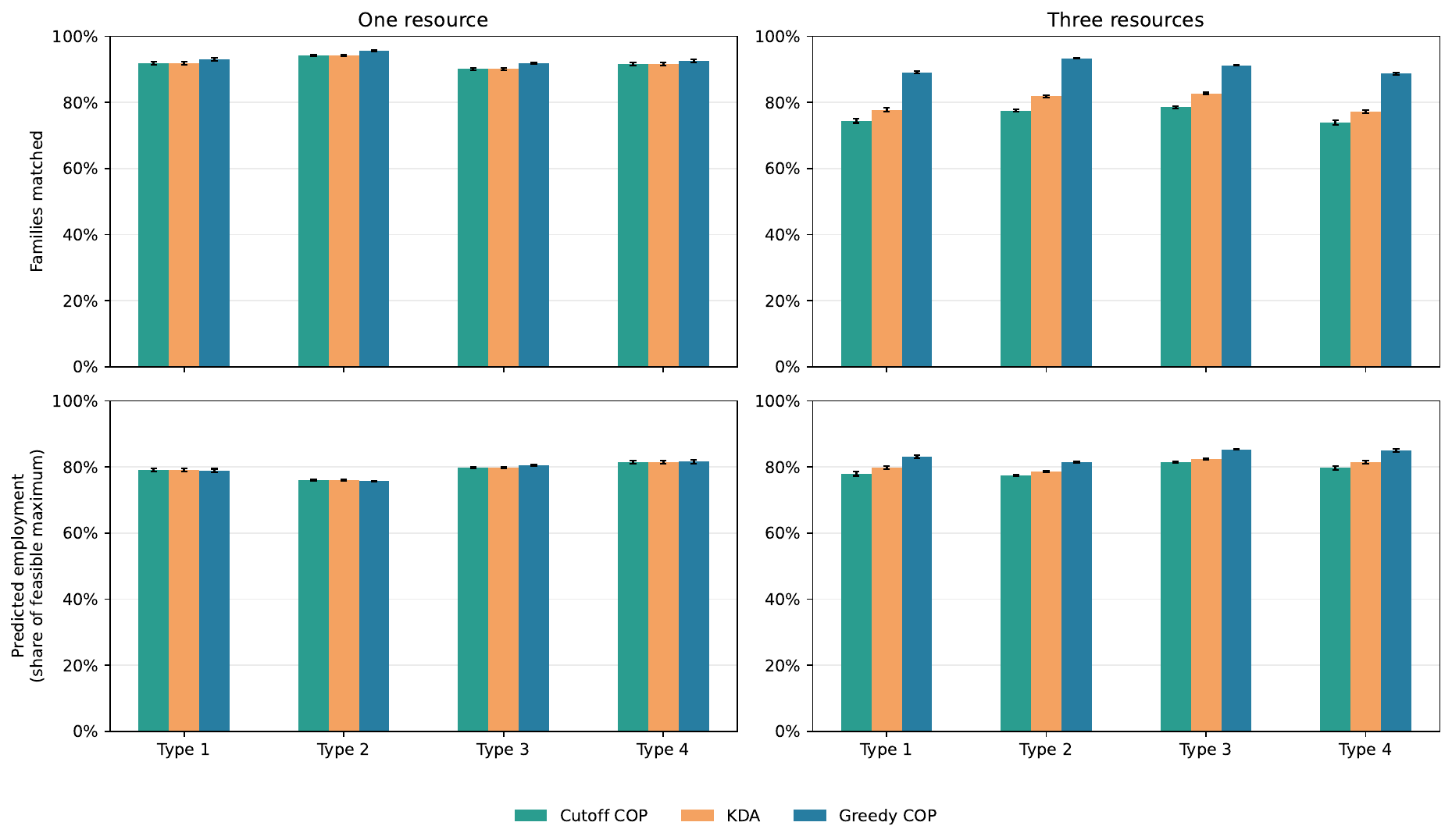}
  \tagmcend
  \tagstructend
  \caption{HIAS-calibrated outcomes. Left: one resource; right: three resources.
  Bars are means across 100 preference
  draws, with greedy COP-and-Choose first averaged over 12 proposal orders.
  Error bars are 95-percent Monte Carlo confidence intervals across simulated
  preference draws.}
  \label{fig:hias-outcomes}
\end{figure}

\Needspace{.5\textheight}
\subsection{Priority-blocker incidence}\label{sec:blocker-incidence}

\autoref{tab:blocker-incidence} reports incidence in terminal proposal sets.

\begin{table}[H]
\centering
\caption{Priority blockers in terminal greedy scans}
\label{tab:blocker-incidence}
\begin{tabular}{llrrr}
\toprule
Dim. & Type & Distinct blockers & Localities & Assigned below (\%) \\
\midrule
1 & Type 1 & 38.39 & 10.27 & 5.24 \\
  & Type 2 & 23.16 & 10.18 & 5.70 \\
  & Type 3 & 22.08 & 9.96 & 5.79 \\
  & Type 4 & 40.79 & 10.08 & 4.98 \\
\addlinespace
3 & Type 1 & 273.96 & 13.83 & 26.26 \\
  & Type 2 & 72.85 & 12.54 & 19.92 \\
  & Type 3 & 93.49 & 12.35 & 18.06 \\
  & Type 4 & 251.35 & 13.53 & 25.03 \\
\bottomrule
\end{tabular}
\par\medskip
\begin{minipage}{.98\linewidth}\footnotesize
\emph{Notes:} Means over 100 preference draws and 12 orders per draw.
A terminal blocker is a family skipped as infeasible in a locality's final
priority scan before a lower-priority family is selected. Distinct blockers
counts each family once per run, even if it blocks at several localities;
Localities counts localities with a terminal blocker. Assigned below is the
share of all 329 families ultimately assigned below a blocker at their own
locality. All 9,600 runs contain a terminal blocker; all 20 localities exhibit
one in at least one run under each resource specification. The counts refer
to each locality's terminal scan.
\end{minipage}
\end{table}

\Needspace{.45\textheight}
\subsection{Strong-envy severity}\label{sec:envy-severity}

\autoref{tab:envy-severity} counts distinct strongly envied assignees
(\emph{targets}) and their localities, conditional on strong envy.

\begin{table}[H]
  \centering
  \caption{Severity of strong envy under greedy COP-and-Choose}
  \label{tab:envy-severity}
  \begin{tabular}{llrrrrr}
    \toprule
    Dim. & Type & Affected (\%) & Mean targets & 90th & Maximum & Mean localities \\
    \midrule
    1 & Type 1 & 11.4 & 1.83 & 3 & 13 & 1.27 \\
      & Type 2 &  6.9 & 2.00 & 4 & 18 & 1.25 \\
      & Type 3 &  6.6 & 2.17 & 4 & 14 & 1.29 \\
      & Type 4 & 12.3 & 1.72 & 3 & 11 & 1.28 \\
    \addlinespace
    3 & Type 1 & 83.0 & 7.37 & 23 & 78 & 2.60 \\
      & Type 2 & 21.3 & 7.20 & 18 & 49 & 2.07 \\
      & Type 3 & 26.7 & 4.48 & 11 & 40 & 1.89 \\
      & Type 4 & 75.9 & 7.65 & 24 & 73 & 2.56 \\
    \bottomrule
  \end{tabular}
  \par\medskip
  \begin{minipage}{0.98\linewidth}\footnotesize
  \emph{Notes:} Each observation is a family in one greedy outcome. Each cell
  pools 100 preference draws and 12 proposal orders. Affected means that the
  family strongly envies at least one assigned family. Mean targets and mean
  localities are conditional on being affected. The 90th percentile and
  maximum refer to target counts among affected family-outcome observations;
  quantiles use the empirical inverse distribution function.
  \end{minipage}
\end{table}

\subsection{TKDA and ordinary justified envy}\label{sec:tkda-envy}

\autoref{tab:tkda-pairwise} reports the familywise welfare comparisons
with TKDA under employment priorities. Our implementation follows
Algorithms~4--5 of \citet{DKT2023}. Preferences rank compatible localities by
simulated utility, and locality priorities rank families by employment
weight; ties follow family or locality index. Capacities are aligned by
locality name, as in the main experiment.

Algorithm~5 assigns an infinite final threshold to a family whose temporary
threshold is infinite. For finite thresholds, our implementation uses the
algorithm's sets of higher-priority families and current proposers. The
replication package includes independent checks of the thresholds and final
assignments against this set formulation.

\begin{table}[H]
\centering
\caption{Welfare comparisons under employment priorities}
\label{tab:tkda-pairwise}
\small
\begin{tabular}{lllrrrr}
\toprule
Resources & First & Second & Better & Worse & \shortstack{Matched\\first} & \shortstack{Matched\\second} \\
\midrule
1 & greedy & KDA & 9.5 & 0.0 & 93.3 & 91.9 \\
1 & greedy & TKDA & 66.0 & 0.0 & 93.3 & 57.0 \\
1 & KDA & TKDA & 63.6 & 0.0 & 91.9 & 57.0 \\
\addlinespace
3 & greedy & KDA & 25.9 & 0.0 & 90.6 & 79.9 \\
3 & greedy & TKDA & 70.8 & 0.0 & 90.6 & 44.5 \\
3 & KDA & TKDA & 59.0 & 0.0 & 79.9 & 44.5 \\
\bottomrule
\end{tabular}
\par\medskip
\begin{minipage}{0.98\linewidth}\footnotesize
\emph{Notes:} All entries are percentages. Better and Worse compare the first mechanism with the second family by family; Matched gives placement rates. TKDA follows Algorithms~4--5 of \citet{DKT2023}. Preferences rank localities by simulated utility, priorities follow employment weights, and capacities are aligned by locality name. Results average over 400 profiles per resource specification, with greedy first averaged over 12 proposal orders.
\end{minipage}
\end{table}

\begin{table}[H]
\centering
\caption{Ordinary justified envy and strong envy}
\label{tab:ordinary-envy}
\small
\begin{tabular}{llrr}
\toprule
Resources & Mechanism & Justified envy & Strong envy \\
\midrule
1 & cutoff & 0.0 & 0.0 \\
1 & KDA & 0.0 & 0.0 \\
1 & TKDA & 0.0 & 0.0 \\
1 & greedy & 9.3 & 9.3 \\
\addlinespace
3 & cutoff & 0.0 & 0.0 \\
3 & KDA & 6.1 & 0.0 \\
3 & TKDA & 0.0 & 0.0 \\
3 & greedy & 51.9 & 51.8 \\
\bottomrule
\end{tabular}
\par\medskip
\begin{minipage}{0.98\linewidth}\footnotesize
\emph{Notes:} A family is counted if it envies at least one assignee under the indicated definition. Ordinary justified envy requires a preferred locality with a lower-priority assignee, irrespective of replacement feasibility. Shares are expressed as percentages and averaged as in \autoref{tab:hias-main}.
\end{minipage}
\end{table}

Ordinary justified envy and strong envy give nearly identical incidence for
greedy. With three resources, however, 6.1 percent of families have ordinary
justified envy under KDA despite its interference-free guarantee. Neither
measure in \autoref{tab:ordinary-envy} imposes the feasible-replacement test
used for weak fairness in \autoref{app:greedy-stability}.

\subsection{Locality-alignment sensitivity}\label{sec:join-sensitivity}

Our baseline aligns capacities with employment and compatibility data by
locality name, using the named public data to construct a crosswalk for the
anonymized openICPSR V1 files. Employment and compatibility columns 14 and 15
correspond to Columbus and Cleveland Heights; our crosswalk assigns them
capacity rows 15 and 14, respectively. Other positions retain their file order.
We also run the full 12-order experiment with all capacities aligned by
position. Preference and proposal seeds are held fixed, and initial-fit
compatibility and the employment-maximizing denominator are recomputed.
The compatibility masks remain unchanged.

Under positional alignment, greedy makes 8.5--11.7 percent of families strictly
better off than KDA with one resource and 20.8--33.1 percent with three,
compared with 8.1--11.5 and 19.0--33.3 percent under name alignment.
The largest type-specific change in the three-dimensional beneficiary share
is 1.83 percentage points. All 9,600 positional-alignment greedy outcomes
continue to give every family a weakly preferred locality assignment relative
to cutoff and KDA under its simulated preference. Waste
reductions remain substantial, and the three-dimensional employment gain
relative to KDA remains positive under every preference specification,
ranging from 1.8 to 4.2 percentage points under positional alignment.

All 12,000 mechanism--order outcomes change across alignments. Name alignment
remains the baseline.

\subsection{Additional priority-design detail}\label{sec:priority-design-detail}

Under the employment-based priorities, all 40 constraint-priority pairs
(20 localities times two resource specifications) fail substitutability.
In 37 cases, the search finds a violation of the containment condition in
\autoref{prop:fixed-priority}(i) at the first level of the recursion.
The replication package reports the corresponding certificates.

\autoref{prop:ascending} applies in one dimension but not in the
three-dimensional specification. Although 91.6 percent of family pairs have
componentwise comparable requirement vectors, the full collection is not
totally ordered, so the proposition provides no analogous guarantee.

Under ascending-size priority, 63.0 percent of one-person families have a
better order-averaged rank and 0.9 percent a worse one, compared with 7.9 and
61.6 percent among families of five or more. These comparisons use ranks
1--20 for acceptable localities and 21 for unassignment, averaged over the
12 proposal orders under employment-based priorities. Mean placements of
families of five or more fall by 6.5.

An additive employment accounting separates newly placed families, families
losing placement, and families placed under both priorities. Their mean
contributions to the change in total normalized employment weight are
$+1.318$, $-1.315$, and $-2.247$, respectively, summing to $-2.243$ up to
rounding. The decline is accounted for by reassignment among continuously
placed families.

For three resources, a bounded screen tests 2,000 independently sampled
rankings per locality, together with employment, total-size, normalized-load,
and three lexicographic rankings. For each ranking it tests 33 available sets
and at most six selected-family deletions per set, stopping at the first
violation. Explicit violations of substitutability or size monotonicity occur
for 38,772 of the 40,000 random rankings and 68 of the 120 specified rankings.
The remaining rankings are unclassified. The package records
each tested ranking and witness. The deletion--contraction decomposition in
\autoref{prop:fixed-priority} does not supply a polynomial-time guarantee for
compactly represented resource constraints.

\subsection{Finite profitable-deviation search}
\label{sec:deviation-search}

The lowest-indexed active family proposes next under every tested report.
For each of the 800 profiles, the search selects the 20 non-top-choice
families with the worst truthful ranks, breaking ties by family index.
Holding other reports truthful, it tests every nonempty strict top truncation,
single-locality deletion, adjacent transposition, and promotion of one
locality to the top. Duplicate reports are removed, and the search stops at
the first profitable report per profile. This finite search gives lower
bounds on manipulability under the fixed proposal convention.

\begin{table}[H]
  \centering
  \caption{Finite search for profitable reports under a fixed proposal order}
  \label{tab:deviation-search}
  \begin{tabular}{llrrrr}
    \toprule
    Resources & Preference & Profiles & Witnesses & Share (\%) & Mean rank gain \\
    \midrule
    One & Type 1 & 100 & 0  & 0.0  & -- \\
        & Type 2 & 100 & 3  & 3.0  & 3.33 \\
        & Type 3 & 100 & 3  & 3.0  & 3.67 \\
        & Type 4 & 100 & 4  & 4.0  & 1.50 \\
    \addlinespace
    Three & Type 1 & 100 & 4  & 4.0  & 1.50 \\
          & Type 2 & 100 & 10 & 10.0 & 4.20 \\
          & Type 3 & 100 & 1  & 1.0  & 5.00 \\
          & Type 4 & 100 & 1  & 1.0  & 1.00 \\
    \bottomrule
  \end{tabular}
  \par\medskip
  \begin{minipage}{0.98\linewidth}\footnotesize
  \emph{Notes:} A witness is a tested
  unilateral report that gives the family a strictly better locality under
  its true simulated ranking. Mean rank gain is conditional on a witness.
  Ranks number acceptable localities from one, with unassignment immediately
  below the last acceptable locality.
  The search examines at most 20 families per profile and stops after its
  first witness. The detected shares are lower bounds on the proportions of
  these simulated profiles at which greedy COP-and-Choose is manipulable under
  the fixed lowest-index-active-family convention.
  \end{minipage}
\end{table}

The search tests about 775,000 unilateral reports and finds witnesses in 10
of 400 one-resource profiles and 16 of 400 three-resource profiles. Conditional
on a witness, the mean improvements are 2.70 and 3.38 true preference-rank
positions, and the maximum improvements are 6 and 8. An independent
implementation reproduces every witness.

\subsection{Completion and incentives}\label{sec:completion-incentives}

Under the lowest-index-active-family proposal convention, completion begins
when greedy COP-and-Choose ends. The lowest-indexed family with a feasible
strict improvement moves to its most-preferred such locality, releasing its
old assignment and leaving every other family's assignment unchanged. The
scan restarts from the lowest index after each move and ends when no move is
available. This defines a deterministic rule at every reported preference
profile; \autoref{prop:completion} guarantees
termination, feasibility, non-wastefulness, and weak Pareto improvement under
the reported preferences.

\begin{table}[H]
\centering
\caption{Welfare before and after unilateral completion}
\label{tab:completion-outcomes}
\small
\begin{tabular}{llrrrr}
\toprule
Resources & Mechanism & \shortstack{Better\\vs. KDA} & \shortstack{Worse\\vs. KDA} & Matched & Mean rank \\
\midrule
1 & greedy & 8.0 & 0.0 & 93.0 & 5.10 \\
1 & greedy + completion & 10.7 & 0.0 & 93.7 & 4.92 \\
\addlinespace
3 & greedy & 25.1 & 0.0 & 90.3 & 5.59 \\
3 & greedy + completion & 26.2 & 0.0 & 90.9 & 5.45 \\
\bottomrule
\end{tabular}
\par\medskip
\begin{minipage}{0.98\linewidth}\footnotesize
\emph{Notes:} The proposal convention always selects the lowest-indexed active family. Completion repeatedly moves the lowest-indexed family with an available improvement to its most-preferred feasible locality. All priorities are employment based. Entries except mean rank are percentages, averaged over 400 profiles per resource specification. These fixed-order outcomes differ from the 12-order averages in \autoref{tab:hias-main}. Mean rank uses that table's convention.
\end{minipage}
\end{table}

With employment priorities, completion makes roughly ten moves per profile
with one resource and seven with three; the respective maxima are 33
and 24. Beneficiary shares relative to KDA increase from 8.0 to 10.7 percent
and from 25.1 to 26.2 percent (\autoref{tab:completion-outcomes}). The
proposal convention is fixed here, so the
uncompleted figures differ from the proposal-order averages in the main
tables. No family loses relative to greedy or KDA.

For each profile, the search selects the 20 worst-ranked non-top-choice
families under their truthful \emph{completed} assignments, with family-index
ties. It uses the report classes and order in \autoref{sec:deviation-search}
and stops at the first profitable report. Both COP-and-Choose and completion
are rerun using the reported preferences; gains are evaluated under the
deviating family's true preference. Selection by completed assignment can
change the candidate families, so comparison with the uncompleted search
does not hold the tested reports fixed.

\begin{table}[H]
\centering
\caption{Finite profitable-deviation search for completed greedy}
\label{tab:completion-incentives}
\small
\begin{tabular}{llrrrr}
\toprule
Resources & Preference & Profiles & Witnesses & Share (\%) & Mean rank gain \\
\midrule
1 & Type 1 & 100 & 0 & 0.0 & -- \\
1 & Type 2 & 100 & 10 & 10.0 & 4.80 \\
1 & Type 3 & 100 & 8 & 8.0 & 3.50 \\
1 & Type 4 & 100 & 6 & 6.0 & 1.33 \\
\addlinespace
3 & Type 1 & 100 & 5 & 5.0 & 1.40 \\
3 & Type 2 & 100 & 17 & 17.0 & 4.35 \\
3 & Type 3 & 100 & 4 & 4.0 & 3.50 \\
3 & Type 4 & 100 & 1 & 1.0 & 1.00 \\
\addlinespace
1, ascending size & All types & 400 & 0 & 0.0 & -- \\
\bottomrule
\end{tabular}
\par\medskip
\begin{minipage}{0.98\linewidth}\footnotesize
\emph{Notes:} Employment priorities apply except in the ascending-size control. Candidate families are selected using their truthful completed assignments. Every tested report is evaluated by rerunning both COP-and-Choose and completion. Mean rank gain is conditional on a witness and uses acceptable-list ranks, with unassignment just below the last acceptable locality. Detected shares are finite-search lower bounds; candidate sets may differ from those in \autoref{tab:deviation-search}.
\end{minipage}
\end{table}

The search reported in \autoref{tab:completion-incentives} tests about
760,000 reports under employment priorities and finds witnesses in
24 of 400 one-resource profiles (6.0 percent) and 27 of 400 three-resource
profiles (6.8 percent). An independent implementation reproduces all
truthful outcomes and every witness.

The one-resource ascending-size exercise is a control. With family-index
ties common across localities, greedy is serial dictatorship and is already
non-wasteful at every report. Completion therefore leaves it unchanged and
preserves strategy-proofness and proposal-order independence. In all 400
truthful profiles and about 377,000 tested reports, completion makes no moves,
and no profitable report is found.

\subsection{Proposal-order dispersion}
\label{sec:proposal-order-dispersion}

Across the original 12 orders, the mean within-market range is
0.089--0.148 rank positions and 0.2--0.6 percentage points of predicted
employment, depending on the cell. Beneficiary counts vary by 7.8 families on
average with one resource and 4.9 with three; maximum ranges are 26 and 21.

\autoref{tab:order-convergence} extends the same 800 market profiles (two
resource specifications and 400 preference draws) from 12 to 50 independently
seeded orders.
The largest absolute change in any cell mean is 0.029 percentage points for
the beneficiary share, 0.007 for placement and employment, 0.068 for strong
envy, and 0.451 for waste; the largest mean-rank change is 0.0012 positions.
The table also compares the disjoint first 12 and next 38 orders. These
checks assess the sensitivity of cell means to additional sampled orders;
they do not bound outcomes over all proposal-order conventions. The main
tables retain the original 12-order specification.

\begin{table}[t]
\centering
\caption{Proposal-order extension on the same 800 preference profiles}
\label{tab:order-convergence}
\begin{tabular}{lrr}
\toprule
Statistic & 12 versus 50 orders & First 12 versus next 38 \\
\midrule
Beneficiary share (pp) & 0.029 & 0.038 \\
Matched share (pp) & 0.007 & 0.010 \\
Mean rank (positions) & 0.0012 & 0.0016 \\
Employment share (pp) & 0.007 & 0.009 \\
Strong-envy incidence (pp) & 0.068 & 0.090 \\
Waste (pp) & 0.451 & 0.593 \\
\bottomrule
\end{tabular}
\par\medskip
\begin{minipage}{.98\linewidth}\footnotesize
\emph{Notes:} Maximum absolute difference in a cell mean across the eight
resource-dimension--preference-type cells. Orders are independently seeded
within each of 100 draws per type, using the same preference draws throughout.
The second comparison uses disjoint order samples. pp denotes percentage
points. The original 12-order estimates remain the main-table specification.
\end{minipage}
\end{table}

\FloatBarrier

\subsection{Overlapping terminal offers}\label{sec:overlap}

At least one family holds terminal offers from more than one locality in
94.4 percent of greedy runs under the one-dimensional constraints and
93.1 percent under the three-dimensional constraints. The mean number of
families holding multiple terminal offers is 3.60 and 2.79 per run,
respectively, with a maximum of 13 under either specification.

\FloatBarrier
\Needspace{6\baselineskip}
\subsection{Retaining rejected proposals}\label{sec:retention}

We compare greedy COP-and-Choose with the deferred-acceptance analogue that
applies greedy choice to current holdings and each new proposal, permanently
discarding rejections. The comparison uses all 9,600 greedy runs in the main
experiment. The two procedures receive the same initial
proposal-randomization seed, but
their realized proposal histories may diverge as their active sets change.

\begin{table}[H]
  \centering
  \caption{Families whose assignments change when greedy rejections are
  discarded}
  \label{tab:retention}
  \begin{tabular}{lrrrrrr}
    \toprule
    & Mean & Median & 90th & 95th & 99th & Maximum \\
    \midrule
    One resource & 19.39 & 19 & 30 & 33 & 39 & 46 \\
    Three resources & 16.09 & 15 & 28 & 31 & 39 & 56 \\
    \bottomrule
  \end{tabular}
  \par\medskip
  \begin{minipage}{0.98\linewidth}\footnotesize
  \emph{Notes:} Statistics are over 4,800 paired runs per resource specification,
  including identical outcomes. Quantiles use the empirical inverse distribution
  function. Each market contains 329 families.
  \end{minipage}
\end{table}

The outcomes differ in 99.8 percent of one-resource runs and 99.9 percent of
three-resource runs. Conditional on a difference, the mean numbers of changed
assignments are 19.4 and 16.1; the quantiles in
\autoref{tab:retention} are unchanged.

The welfare effects are mixed. With one resource, an average of 10.4 families
prefer COP-and-Choose and 9.0 prefer the discarded-rejection procedure; with
three resources, the corresponding means are 8.3 and 7.8. The paired
outcomes are Pareto incomparable in 99.2 and 97.6 percent of runs,
respectively. COP-and-Choose changes mean placement by $-0.17$ and $+0.80$
families and mean simulated-preference rank by $+0.009$ and $-0.050$ positions; lower
ranks are better.

\FloatBarrier
\Needspace{.45\textheight}
\subsection{KDA-preserving improvement ceilings}\label{sec:oracle}

No greedy run attains the beneficiary ceiling in \autoref{tab:oracle}.
Mean maxima are 66.7 and 123.2 beneficiaries with one and three resources;
maximum placement gains are 15.9 and 52.2 families, against 4.4 and 35.2 under
greedy. A beneficiary-maximizing matching reaches the unconditional placement
ceiling in only 9 and 24 of the respective 400 markets.

\begin{table}[H]
  \centering
  \caption{Greedy COP-and-Choose relative to KDA-preserving improvement
  ceilings}
  \label{tab:oracle}
  \begin{tabular}{lccc}
    \toprule
    & Beneficiaries (\%) & Placement (\%) & Employment (\%) \\
    \midrule
    One resource     & $48.9$ & $27.7$ & $0.3$ \\
    Three resources  & $70.1$ & $67.4$ & $44.4$ \\
    \bottomrule
  \end{tabular}
  \par\medskip
  \begin{minipage}{0.98\linewidth}\footnotesize
  \emph{Notes:} Entries report greedy's gain as a percentage of the largest gain
  attainable by any matching subject to no family receiving a worse locality
  than under KDA according to its simulated preference. Beneficiary capture is
  greedy's number of beneficiaries divided by the largest attainable number;
  placement and employment capture are greedy's incremental gains over KDA
  divided by the corresponding largest attainable gains. Reported values are
  means of the corresponding market-level capture rates over 400 preference
  draws for each resource specification, with greedy first averaged over 12
  proposal orders within a draw.
  \end{minipage}
\end{table}

With one resource, the mean employment gain available above KDA is 2.07
weight units and greedy realizes 0.13. The 0.3-percent capture entry averages
market-level ratios. Greedy's order-averaged employment
is below KDA's in 152 of 400 draws. With three resources, it realizes 2.27 of
the 5.15 available units and is below KDA in 13 of 400 draws.

Decomposing greedy's employment change into newly placed families and
reassignment of already placed families gives mean contributions of
$+0.333$ and $-0.204$ weight units with one resource, and $+2.581$ and
$-0.315$ with three. No family placed under KDA loses placement under greedy.

\end{document}